%% file: 2DknapsackArxiv.tex
\documentclass[letter,11pt]{article}

\usepackage[left=1in,right=1in,top=1in,bottom=1in]{geometry}

\usepackage[english]{babel}
\usepackage[utf8]{inputenc}
\usepackage{amsmath}
\usepackage{amsthm}
\usepackage{amsfonts}
\usepackage{graphicx}
\usepackage[colorinlistoftodos]{todonotes}
\usepackage{verbatim} 
\usepackage{enumitem}

\usepackage{fetamont}
\usepackage[T1]{fontenc}

\usepackage{times}
\usepackage{fullpage}

\usepackage{tikz,pgfplots}
\usetikzlibrary{patterns}
\usepackage{caption, setspace, subcaption}
\usepackage{authblk}

\newtheorem{thm}{Theorem}
\newtheorem{lem}[thm]{Lemma}
\newtheorem{cor}[thm]{Corollary}
\newtheorem{claim}{Claim}
\newtheorem{remark}{Remark}
\newtheorem{pro}[thm]{Proposition}

\usepackage{xspace}
\newcommand{\eps}{\varepsilon}
\newcommand{\tdk}{$2DK$\xspace}
\newcommand{\tdkr}{$2DKR$\xspace}
\newcommand{\epst}{\varepsilon_{box}}
\newcommand{\epss}{\varepsilon_{small}}
\newcommand{\epsl}{\varepsilon_{large}}

\newcommand{\epsb}{\varepsilon_{box}}
\newcommand{\epsr}{\varepsilon_{ring}}

\newcommand{\epsau}{\varepsilon_{ra}}

\newcommand{\apx}{APX}

\newcommand{\opt}{OPT}
\newcommand{\optco}{OPT_{corr}}
\newcommand{\optla}{OPT_{large}}
\newcommand{\optsm}{OPT_{small}}
\newcommand{\optsk}{OPT_{skew}}
\newcommand{\optho}{OPT_{hor}}
\newcommand{\optve}{OPT_{ver}}
\newcommand{\optln}{OPT_{long}}

\newcommand{\optfa}{OPT_{fat}}
\newcommand{\optth}{OPT_{thin}}

\newcommand{\optki}{OPT_{kill}}
\newcommand{\optin}{OPT_{int}}
\newcommand{\optbo}{OPT_{box}}
\newcommand{\optrc}{OPT_{\fontL\&C}}
\newcommand{\Rsm}{I_{small}}
\newcommand{\Rla}{I_{large}}
\newcommand{\Rho}{I_{hor}}
\newcommand{\Rve}{I_{ver}}
\newcommand{\Rsk}{I_{skew}}
\newcommand{\Rin}{I_{int}}

\newcommand{\R}{I}
\renewcommand{\L}{\mathcal{L}}

\newcommand{\fontL}{L}
\newcommand{\T}{OPT_T}
\newcommand{\il}{I_{\mathrm{long}}}
\newcommand{\ilopt}{OPT_{\mathrm{long}}}
\newcommand{\is}{I_{\mathrm{short}}}
\newcommand{\ilong}{\il}
\newcommand{\ishort}{\is}
\newcommand{\isopt}{OPT_{\mathrm{short}}}
\newcommand{\efl}{\varepsilon_{\fontL}}
\newcommand{\ilthin}{\ilopt \cap \T}
\newcommand{\ilfat}{\ilopt \setminus \T}
\newcommand{\isfat}{\isopt \setminus \T}
\newcommand{\isthin}{\isopt \cap \T}
\newcommand{\isfhor}{(\isopt \setminus \T)_{hor}}
\newcommand{\isfver}{(\isopt \setminus \T)_{ver}}

\newcommand{\F}{OPT_F}
\newcommand{\OK}{OPT_K}
\newcommand{\LF}{OPT_{LF}}
\newcommand{\ST}{OPT_{ST}}
\newcommand{\LT}{OPT_{LT}}
\newcommand{\SF}{OPT_{SF}}

\global\long\def\cell{C}
\global\long\def\C{\mathcal{C}}
\global\long\def\K{\mathcal{K}}

\newcommand{\bottomc}{bottom}
\newcommand{\topc}{top}
\newcommand{\leftc}{left}
\newcommand{\rightc}{right}
\newcommand{\height}{h}
\newcommand{\width}{w}
\newcommand{\profit}{p}
\newcommand{\shift}{shift}
\newcommand{\base}{base}
\newcommand{\cT}{{\cal T}}
\newcommand{\cR}{{\cal R}}

\ifdefined\DEBUG
\newcommand{\fab}[1]{\textcolor{red}{#1}}
 \newcommand{\sal}[1]{\textcolor{green}{#1}}
 \newcommand{\ari}[1]{\textcolor{blue}{#1}}
 \newcommand{\wal}[1]{\textcolor{magenta}{#1}}
 \newcommand{\andy}[1]{\textcolor{cyan}{#1}}
 \newcommand{\san}[1]{\textcolor{orange}{#1}}
  \def\rem#1{{\marginpar{\raggedright\scriptsize #1}}}
  \newcommand{\fabr}[1]{\rem{\textcolor{red}{$\bullet$ #1}}}
  \newcommand{\salr}[1]{\rem{\textcolor{green}{$\bullet$ #1}}}
  \newcommand{\arir}[1]{\rem{\textcolor{blue}{$\bullet$ #1}}}
  \newcommand{\walr}[1]{\rem{\textcolor{magenta}{$\bullet$ #1}}}
  \newcommand{\andyr}[1]{\rem{\textcolor{cyan}{$\bullet$ #1}}}
  \newcommand{\sanr}[1]{\rem{\textcolor{orange}{$\bullet$ #1}}}
\else
  \newcommand{\fab}[1]{#1}
  \newcommand{\sal}[1]{#1}
  \newcommand{\ari}[1]{#1}
  \newcommand{\wal}[1]{#1}
  \newcommand{\andy}[1]{#1}
  \newcommand{\san}[1]{#1}
  \newcommand{\fabr}[1]{}
  \newcommand{\salr}[1]{}
  \newcommand{\arir}[1]{}
  \newcommand{\walr}[1]{}
  \newcommand{\andyr}[1]{}
  \newcommand{\sanr}[1]{}
\fi

\title{Approximating Geometric Knapsack \fab{via L-packings}\thanks{{\fontsize{11}{20}The authors from IDSIA are partially supported by ERC Starting Grant NEWNET 279352 and SNSF Grant APXNET 200021$\_$159697$/$1. 
Arindam Khan is  supported in part by the European Research Council, Grant Agreement No. 691672, the work was primarily done when the author was at IDSIA.
Sandy Heydrich is in part supported by the Google Europe PhD Fellowship.}}}
\author[*]{Waldo G\'alvez}
\author[*]{Fabrizio Grandoni}
\author[**]{Sandy Heydrich}
\author[*]{Salvatore Ingala}
\author[***]{Arindam Khan}
\author[****]{Andreas Wiese}
\affil[*]{IDSIA, USI-SUPSI, Switzerland\\
	\texttt{[waldo,fabrizio,salvatore]@idsia.ch}}
\affil[**]{MPI for Informatics and Saarbr\"ucken Graduate School of Computer Science, Germany,
	\texttt{heydrich@mpi-inf.mpg.de}}
\affil[***]{Department of Computer Science, Technical University of Munich, Munich, Germany,
	\texttt{arindam.khan@in.tum.de}}
\affil[****]{Department of Industrial Engineering and Center for Mathematical Modeling, Universidad de Chile, Chile,
	\texttt{awiese@dii.uchile.cl}}

\date{\empty}

\begin{document}
\maketitle

\thispagestyle{empty}

\begin{abstract}
\noindent We study the two-dimensional geometric knapsack problem (\tdk) in which we
are given a set of $n$ axis-aligned rectangular items, each one with an associated profit, and an axis-aligned square knapsack. The goal is to find a \fab{(non-overlapping) packing} of a maximum profit subset of items inside the knapsack \fab{(without rotating items)}.
The best-known polynomial-time approximation factor for this problem (even just in the cardinality case) is $2+\eps$
{[}Jansen and Zhang, SODA 2004{]}. In this paper we break the $2$ approximation barrier, achieving a polynomial-time $\frac{17}{9}+\eps<1.89$ approximation, which improves to $\frac{558}{325}+\eps<1.72$ in the cardinality case. 

Essentially all prior work on \tdk~approximation packs items inside a constant number of rectangular containers, \fab{where items inside each container are packed  using a simple greedy strategy}.
We deviate for the first time from this setting: we show that there exists a large profit solution where items are packed inside a constant number of containers \emph{plus} one \fontL-shaped region at the boundary of the knapsack 
which contains items that are high and narrow and items that are wide and thin. The items of these two types possibly
interact in a complex manner at the corner of the $\fontL$. 

The above structural result is not enough however: the best-known approximation ratio for the subproblem in the \fontL-shaped 
region is $2+\eps$ \fab{(obtained via a trivial reduction to one-dimensional knapsack by considering tall or wide items only).}\fabr{I think it is important to stress that the best known here is trivial}
Indeed this is one of the simplest special settings of the problem for which this is the best known approximation factor. 
As a second major, \fab{and the main algorithmic} contribution of this paper, we present a PTAS for this case.
We believe that this will turn out to be useful in future work in geometric packing problems. 

We also consider the variant of the problem \emph{with rotations} (\tdkr), where items can be rotated by $90$ degrees. Also in this case the best-known polynomial-time approximation factor (even for the cardinality case) is $2+\eps$
{[}Jansen and Zhang, SODA 2004{]}. Exploiting part of the machinery developed for \tdk plus a few additional ideas, we obtain a polynomial-time $3/2+\eps$-approximation for \tdkr, which improves to $4/3+\eps$ in the cardinality case. 
%
%
%
%
\end{abstract}

\newpage

\setcounter{page}{1}

\section{Introduction}

The \emph{(\fab{two}-dimensional) geometric knapsack} problem (\tdk) is the geometric variant of the classical (\fab{one}-dimensional) knapsack problem. We are given a set of $n$ items $I=\{1,\ldots,n\}$, where each item $i\in I$ is an axis-aligned open rectangle $(0,\width(i))\times (0,\height(i))$ in the \fab{two}-dimensional plane, and has an associated profit $\profit(i)$. Furthermore, we are given an axis-aligned square knapsack $K=[0,N]\times [0,N]$. W.l.o.g. we next assume that all values $\width(i)$, $\height(i)$, $\profit(i)$ and $N$ are positive integers. 
Our goal is to select a subset of items $OPT\subseteq I$ of maximum total profit $opt=\profit(OPT):=\sum_{i\in OPT}\profit(i)$ and to place them so that the selected rectangles are pairwise disjoint and fully contained in the knapsack. More formally, for each $i\in OPT$ we have to define a pair of coordinates $(\leftc(i),\bottomc(i))$ that specify the position of the bottom-left corner of $i$ in the packing. In other words, $i$ is mapped into a rectangle $R(i):=(\leftc(i),\rightc(i))\times (\bottomc(i),\topc(i))$, with $\rightc(i)=\leftc(i)+\width(i)$ and $\topc(i)=\bottomc(i)+\height(i)$. For any two $i,j\in OPT$, we must have $R(i)\subseteq K$ and $R(i)\cap R(j)=\emptyset$.

Besides being a natural mathematical problem, \tdk~is well-motivated by practical applications. For instance, one might want to place advertisements on a board or a website, or cut rectangular pieces from a sheet of some material. Also, it models
a scheduling setting where each rectangle corresponds to a job that
needs some ``consecutive amount'' of a given resource (memory storage, frequencies, etc.). In all these cases, dealing with rectangular shapes only is a reasonable simplification and often the developed techniques can be extended to deal with more general instances.

\tdk is NP-hard \cite{leung1990packing}, and it was intensively studied from the point of view of approximation algorithms. The best known polynomial time approximation
algorithm for it is due to Jansen and Zhang and yields a $(2+\eps)$-approximation~\cite{Jansen2004}.
This is the best known result even in the \emph{cardinality} case (with all profits being $1$).
However, there are reasons to believe that much better polynomial
time approximation ratios are possible: there is a QPTAS under the assumption that  $N=n^{\fab{\mathrm{poly}(\log n)}}$~\cite{adamaszek2015knapsack},
and there are PTASs if the profit of each item equals its area~\cite{bansal2009structural},
if the size of the knapsack can be slightly increased (resource augmentation)~\cite{fishkin2005packing,jansen2007new}, if all items are relatively small \cite{FishkinGJ05}
and if all input items are squares~\cite{jansen2008ipco,HeydrichWiese2017}. \fab{Note that, with no restriction on $N$, the current best approximation for \tdk is $2+\eps$ \emph{even in quasi-polynomial time}}\footnote{\fab{The role of $N$ in the running time is delicate, as shown by recent results on the related \emph{strip packing} problem \cite{AKPP16,GGIK16,HarrenJPS14,JR16,nadiradze2016approximating}}.}.

All prior polynomial-time approximation algorithms for \tdk~implicitly or explicitly exploit a \emph{container-based} packing approach. The idea is to partition the knapsack into a constant number of axis-aligned rectangular regions (\emph{containers}). The sizes (and therefore positions) of these containers can be \emph{guessed} in polynomial time. Then items are packed inside the containers in a simple way: either one next to the other from left to right or from bottom to top (similarly to the one-dimensional case), or by means of the \fab{simple greedy} Next-Fit-Decreasing-Height algorithm. 
Indeed, also the QPTAS in~\cite{adamaszek2015knapsack} can be cast in this framework, with the relevant difference that the number of containers in this case is poly-logarithmic (leading to a quasi-polynomial running time).\fabr{I think here we wish to put the focus on L\&C-packings rather than on N versus n.}

One of the major bottlenecks to achieve approximation factors better than $2$ (in polynomial-time) is that items that are high and narrow (\emph{vertical} items) and items that are wide and thin (\emph{horizontal}
items) can interact in a very complicated way. Indeed, consider the following seemingly simple \fontL\emph{-packing} problem: we are given a set of items $i$ with either $\width(i)>N/2$ (horizontal items) or $\height(i)>N/2$ (vertical items). 
Our goal is to pack a maximum profit subset of them inside an $\fontL$-shaped region $\fontL=([0,N]\times [0,\height_\fontL]) \cup ([0,\width_\fontL]\times [0,N])$, so that horizontal (resp., vertical) items are packed in the bottom-right (resp., top-left) of $\fontL$. 
To the best of our knowledge, the best-known approximation ratio for L-packing is $2+\eps$: Remove either all vertical or all horizontal items, and then pack the remaining items by a simple reduction to one-dimensional knapsack (for which an FPTAS is known).
\fab{It is unclear whether a container-based packing can achieve a better approximation factor, and we conjecture that this is not the case.}\fabr{Is this ok?}
As we will see, a better understanding of L-packing will play a major role in the design of improved approximation algorithms for \tdk.

%
%
%

\subsection{Our contribution}

In this paper we break the $2$-approximation barrier for \tdk. In
order to do that, we substantially deviate for the first time from
\emph{pure} container-based packings, \fab{which are, either implicitly or explicitly, at the hearth of prior work}. Namely,
we consider \emph{L\&C-packings} that combine $O_\eps(1)$ containers
\emph{plus} one L-packing of the above type (\fab{see Fig.\ref{fig:packing+ring}.(a)}), \fab{and show that one such packing has large enough profit.} 

\fabr{I was not happy with the flow of the previous version. I tried another variant. I think the PTAS for L\&C-packing can be left as implicit}
While it is easy to pack almost optimally items into containers, the mentioned $2+\eps$ approximation for L-packings is not sufficient to achieve altogether a better than $2$ approximation factor: indeed, the items of the L-packing might carry all the profit! The \fab{main algorithmic contribution} of this paper is a PTAS for the L-packing problem.\fabr{I stressed that there is some algorithmic content, and also that the hard thing is poly-time} 
\fab{It is easy to solve this problem optimally in pseudo-polynomial time $(Nn)^{O(1)}$ by means of dynamic programming}. We show that a $\fab{1+\eps}$ approximation can be obtained by restricting the top (resp., right) coordinates of horizontal (resp., vertical) items to a proper set that can be
computed in polynomial time \fab{$n^{O_\eps(1)}$}. \fab{Given that, one can adapt the above dynamic program to run in polynomial time.}


\begin{thm}\label{thm:main:Lpacking} There is a PTAS for the L-packing
problem. 
\end{thm}


In order to illustrate the power of our approach, we next sketch a simple $\frac{16}{9}+O(\eps)$ approximation for the cardinality case of \tdk (details in Section~\ref{sec:tdk_car:simple}).\fabr{Some discussion on large items helps the reader IMO} By standard arguments\footnote{There can be at most $O_\eps(1)$ such items in any feasible solution, and if the optimum solution contains only $O_\eps(1)$ items we can solve the problem optimally by brute force.} it is possible to discard \emph{large} items with both sides longer than $\eps\cdot N$. The remaining items have height or width smaller than $\eps\cdot N$ (\emph{horizontal} and \emph{vertical} items, resp.). Let us delete all items intersecting a random vertical or horizontal strip of width $\eps\cdot N$ inside the knapsack. We can pack the remaining items into $O_\eps(1)$ containers by exploiting the PTAS under one-dimensional resource augmentation for \tdk in
\cite{jansen2007new}\footnote{Technically this PTAS is not container-based, however in Section \ref{sec:resource-augmentation} we show that it can be cast in that framework. Our version of the PTAS simplifies the algorithms and \fab{works also in the case with rotations}: this might be a handy black-box tool.}. \fab{A vertical strip deletes vertical items with $O(\eps)$ probability, and horizontal ones with probability roughly proportional to their width, and symmetrically for a horizontal strip. In particular, let us call \emph{long} the items with longer side larger than $N/2$, and \emph{short} the remaining items. Then the above argument gives in expectation roughly one half of the profit $opt_{long}$ of long items, and three quarters of the profit $opt_{short}$ of short ones. This is already good enough unless $opt_{long}$ is large compared to $opt_{short}$.}

At this point L-packings and our PTAS come into play. We shift long items such that they form $4$ stacks
at the sides of the knapsack in a \emph{ring-shaped} region, see Fig.\ref{fig:packing+ring}\fab{.(b)-(c)}: this is possible since any vertical long item cannot have
a horizontal long item \emph{both} at its left and at its right, and vice
versa. Next we delete the least profitable of these stacks and rearrange
the remaining long items into an L-packing, \fab{see Fig.\ref{fig:packing+ring}.(d)}. Thus using our PTAS for L-packings, we can compute a solution of profit \fab{roughly three quarters of $opt_{long}$}. The reader might check that the \fab{combination of these two algorithms} gives the claimed approximation factor. 

\fab{Above we used either $O_\eps(1)$ containers or one L-packing: by combining the two approaches together} and with a more sophisticated case analysis we achieve the following result (see Section \ref{sec:tdk_car:refined}).

\begin{thm}\label{trm:tdk_car:refined} There is a polynomial-time
$\frac{558}{325}+\eps<1.72$ approximation algorithm for cardinality \tdk.
\end{thm}

\begin{figure}
\begin{centering}
\includegraphics[height=4cm]{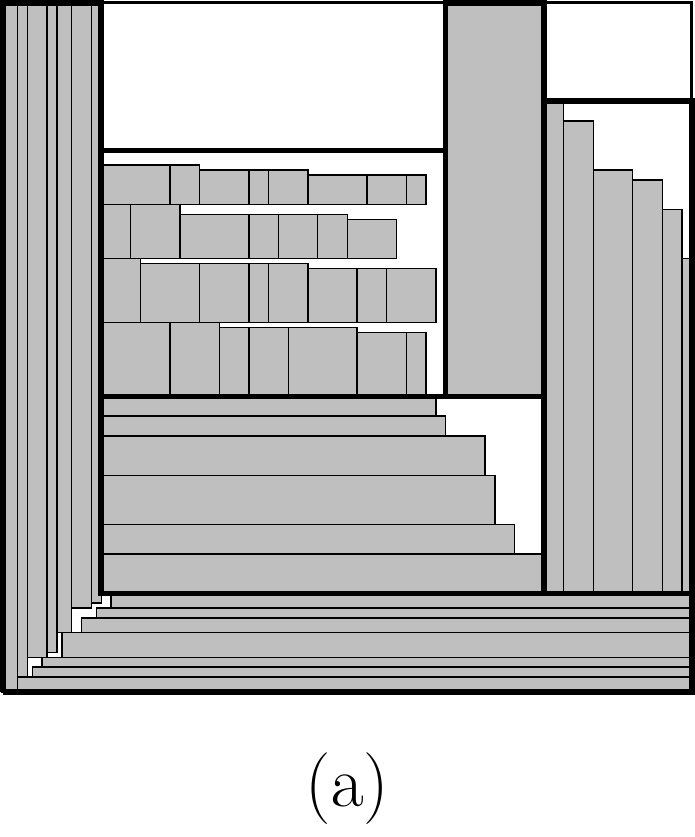}~~~~~~~~~~~~\includegraphics[height=4cm]{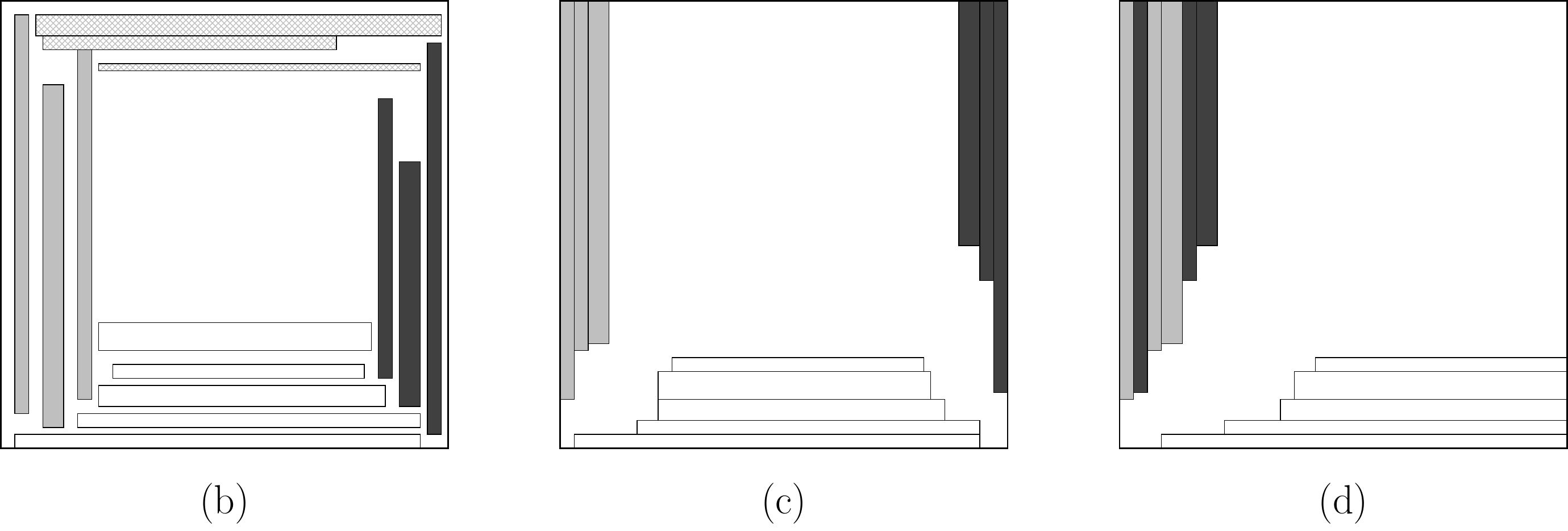}
\par\end{centering}
\caption{\fab{(a) An L\&C-packing with $4$ containers, where the top-left container is packed by means of Next-Fit-Decreasing-Height. (b) A subset of long items. (c) Such items are shifted into $4$ stacks at the sides of the knapsack, and the top stack is deleted. (d) The final packing into an L-shaped region.}}
\label{fig:packing+ring} 
\end{figure}\sanr{Removed the term ``skewed'' from the caption of fig. 1 as this was not introduced. ``Long items'' should be clear enough.}

For weighted \tdk we face severe technical complications for proving
that there is a profitable L\&C-packing. One key reason is that in \fab{the weighted case we cannot discard large items since even one such item might contribute a large fraction to the optimal profit}.
In order to circumvent these difficulties, we
exploit the \emph{corridor-partition} at the hearth of the QPTAS for
\tdk~in \cite{adamaszek2015knapsack} (in turn inspired by prior
work in \cite{AW2013}). Roughly speaking, there exists a partition
of the knapsack into $O_{\eps}(1)$ \emph{corridors}, \fab{consisting of the \emph{concatenation} of $O_{\eps}(1)$ \san{(partially overlapping)} rectangular regions (\emph{subcorridors}).}\fabr{Subcorridors are used later so we need to define them.}
In \cite{adamaszek2015knapsack} the authors partition
the corridors into a \emph{poly-logarithmic} number of containers. Their main algorithm then guesses these containers
in time $n^{\mathrm{poly}(\log n)}$. However, we can only \fab{handle} a \emph{constant} number of containers in polynomial time.
Therefore, we present a different way to partition the corridors into
containers: \fab{here} we lose the profit of a set of \emph{thin} items,
which in some sense play the role of long items in the previous discussion.
These thin items fit in a \emph{very narrow} ring at the boundary of the
knapsack and we map them to an \fontL-packing in the same way as
in the cardinality case above. Some of the remaining non-thin items
are then packed into $O_{\eps}(1)$ containers that are placed in
the (large) part of the knapsack not occupied by the \fontL-packing.
Our partition of the corridors is based on a somewhat intricate case
analysis that exploits the fact that \fab{\emph{long} consecutive subcorridors are arranged in the shape of \emph{rings} or \emph{spirals}: this is used to show the existence of a profitable L\&C-packing.} 
\begin{thm}\label{trm:tdk_weight} There is a polynomial-time
$\frac{17}{9}+\eps<1.89$ approximation algorithm for (weighted) \tdk.
\end{thm}

\paragraph{Rotation setting.}

In the variant of \tdk \emph{with rotations} (\tdkr),
we are allowed to rotate any rectangle $i$ by $90$ degrees. This
means that $i$ can also be placed in the knapsack as a rectangle
of the form $(\leftc(i),\leftc(i)+\height(i))\times(\bottomc(i),\bottomc(i)+\width(i))$.
The best known polynomial time approximation factor for \tdkr (even
for the cardinality case) is again $2+\eps$ due to~\cite{Jansen2004}
and the mentioned QPTAS in \cite{adamaszek2015knapsack} works also
for this case.

By using the techniques described above and exploiting a
few more ideas, we are also able to improve the approximation factor
for \tdkr~(see Sections \ref{sec:cardRot} and \ref{sec:weightedRot}
for the cardinality and general case, resp.). The basic idea
is that any thin item can now be packed inside a narrow vertical strip
(say at the right edge of the knapsack) by possibly rotating it. This
way we do not lose one quarter of the profit due to the mapping to
an \fontL-packing and instead place all items from the ring into
the mentioned strip (while we ensure that their total width is small).
The remaining short items are packed by means of a novel \emph{resource
contraction} lemma: unless there is one \emph{huge item} that occupies almost
the whole knapsack (a case that we consider separately), we can pack
almost one half of the profit of non-thin items in a \emph{reduced}
knapsack where one of the two sides is shortened by a factor $1-\eps$
(hence leaving enough space for the vertical strip). We remark that
here we heavily exploit the possibility to rotate items. Thus, roughly
speaking, we obtain either all profit of non-thin items, or all profit
of thin items plus one half of the profit of non-thin items:
this gives a $3/2+\eps$ approximation. A further refinement of this
approach yields a $4/3+\eps$ approximation in the cardinality case.
We remark that, while resource augmentation is a well-established
notion in approximation algorithms, resource contraction seems to
be a rather novel direction to explore.

\begin{thm}\label{thm:mainNoRotation} For any constant $\eps>0$,
there exists a polynomial-time $\frac{3}{2}+\eps$ approximation algorithm
for \tdkr. In the cardinality case the approximation factor can be
improved to $\frac{4}{3}+\eps$. \end{thm}

\subsection{Other related work}

The mentioned $(2+\eps)$-approximation for two-dimensional knapsack~\cite{Jansen2004}
works in the weighted case of the problem. However, in the unweighted
case a simpler $(2+\eps)$-approximation is known~\cite{jansen2004maximizing}.
If one can increase the size of the knapsack by a factor $1+\eps$
in both dimensions then one can compute a solution of optimal weight,
rather than an approximation, in time $f(1/\eps)\cdot n^{O(1)}$
where the exponent of $n$ does not depend on $\eps$~\cite{HeydrichWiese2017}
(for some suitable function $f$). Similarly, for the case of squares
there is a $(1+\eps)$-approximation algorithm known with such
a running time, i.e., an EPTAS~\cite{HeydrichWiese2017}. This improves
previous results such as a $(5/4+\eps)$-approximation~\cite{Harren06}
and the mentioned PTAS~\cite{jansen2008ipco}. Two-dimensional knapsack
is the separation problem when we want to solve the configuration-LP
for two-dimensional bin-packing. Even though we do not have a PTAS
for the former problem, Bansal et al.~\cite{bansal2009structural}
show how to solve the latter LP to an $(1+\eps)$-accuracy using
their PTAS for two-dimensional knapsack for the special case where
the profit of each item equals its area. The best known (asymptotic) result
for two-dimensional bin packing is due to Bansal and Khan and it is
based on this configuration-LP, achieving an approximation ratio of
$1.405$~\cite{bansal2014binpacking} which improves a series of
previous results \cite{jansen2007new,bansal2009new,caprara2002packing,kenyon2000near,chung1982packing}. See also the recent survey in \cite{CKPT17} and \cite{Khan16}.



\section{A PTAS for \fontL-packings}
\label{sec:ptasL}

In this section we present a PTAS for the problem of finding an optimal \fontL-packing. In this problem we are given a set of \emph{horizontal} items $I_{hor}$ with width larger than $N/2$, and a set of \emph{vertical} items $I_{ver}$ with height larger than $N/2$. Furthermore, we are given an \fontL-shaped region $\fontL=([0,N]\times [0,\height_\fontL]) \cup ([0,\width_\fontL]\times [0,N])$. Our goal is to pack a  subset $OPT\subseteq I:=I_{hor}\cup I_{ver}$ of maximum total profit $opt=\profit(OPT):=\sum_{i\in OPT}\profit(i)$, such that $OPT_{hor}:=OPT\cap I_{hor}$ is packed inside the \emph{horizontal box} $[0,N]\times [0,\height_\fontL]$ and $OPT_{ver}:=OPT\cap I_{ver}$ is packed inside the \emph{vertical box} $[0,\width_\fontL]\times [0,N]$. We remark that packing horizontal and vertical items independently is not possible due to the possible overlaps in the intersection of the two boxes: this is what makes this problem non-trivial, in particular harder than standard (one-dimensional) knapsack.

Observe that in an optimal packing we can assume w.l.o.g. that items in $OPT_{hor}$ are pushed as far to the right/bottom as possible. Furthermore, the items in $OPT_{hor}$ are packed from bottom to top in non-increasing order of width. Indeed, it is possible to \emph{permute} any two items violating this property while keeping the packing feasible. A symmetric claim holds for $OPT_{ver}$. See Fig. \ref{fig:packing+ring}.(d) for an illustration.

Given the above structure, it is relatively easy to define a dynamic program (DP) that computes an optimal L-packing in pseudo-polynomial time $(Nn)^{O(1)}$. The basic idea is to scan items of $I_{hor}$ (resp. $I_{ver}$) in decreasing order of width (resp., height), and each time \emph{guess} if they are part of the optimal solution $OPT$. At each step either both the considered horizontal item $i$ and vertical item $j$ are not part of the optimal solution, or there exist a \emph{guillotine cut}\footnote{A guillotine cut is an \san{infinite, axis-parallel} line $\ell$ that partitions the items in a given packing in two subsets without intersecting any item.} separating $i$ or $j$ from the rest of $OPT$. Depending on the cases, one can define a smaller L-packing sub-instance (among $N^2$ choices) for which the DP table already contains a solution.

In order to achieve a $(1+\eps)$-approximation in polynomial time $n^{O_\eps(1)}$, we show that it is possible (with a small loss in the profit) to restrict the possible top coordinates of $OPT_{hor}$ and right coordinates of $OPT_{ver}$ to proper polynomial-size subsets $\cT$ and $\cR$, resp. We call such an L-packing \emph{$(\cT,\cR)$-restricted}. By adapting the above DP one obtains:   


\begin{lem}\label{lem:DPrestricted}
An optimal $(\cT,\cR)$-restricted \fontL-packing can be computed in time polynomial in $m:=n+|\cT|+|\cR|$ using dynamic programming.
\end{lem} 

\begin{proof}
For notational convenience we assume $0\in \cT$ and $0\in \cR$.\fabr{Removed $\height_L\in \cT$ and $\width_L\in \cR$ not needed. Right?} 
\fabr{Switch to a and b, so you can use h for height}
Let $h_1,\ldots,h_{n(h)}$ be the items in $I_{hor}$ in decreasing order of width and $v_1,\ldots,v_{n(v)}$ be the items in $I_{ver}$ in decreasing order of height (breaking ties arbitrarily). For $\width\in [0,\width_\fontL]$ and $\height\in [0,\height_\fontL]$, let $L(\width,\height)=([0,\width]\times [0,N])\cup ([0,N]\times [0,\height])\subseteq \fontL$. Let also $\Delta \fontL(\width,\height)=([\width,\width_\fontL]\fab{\times} [\height,N])\fab{\cup} ([\width,N]\times [\height,\height_\fontL])\subseteq \fontL$. Note that $\fontL=\fontL(\width,\height)\cup \Delta \fontL(\width,\height)$.

We define a dynamic program table \fab{$DP$} indexed by $i\in [1,n(h)]$ and $j\in [1,n(v)]$, by a top coordinate $t \in \cT$, and a right coordinate $r \in \cR$. The value of $DP(i,t,j,r)$ is the maximum profit of a $(\cT,\cR)$-restricted packing of a subset of $\{h_i,\ldots,h_{n(h)}\}\cup \{v_j,\ldots,v_{n(v)}\}$ inside $\Delta \fontL(\fab{r,t})$. The value of $DP(1,0,1,0)$ is the value of the optimum solution we are searching for. 
Note that the number of table entries is upper bounded by $m^4$.

We fill in $DP$ according to the partial order induced by vectors $(i,t,j,r)$, processing larger vectors first. The base cases are given by 
$(i,j)=(n(h)+1,n(v)+1)$ and $(\fab{r,t})=(\width_\fontL,\height_\fontL)$, in which case the table entry has value $0$.\fabr{Check base cases}

In order to compute any other table entry $DP(i,t,j,r)$, with optimal solution $OPT'$, we take the maximum of the following few values:
\begin{itemize}\itemsep0pt
\item[$\bullet$] If $i\leq n(h)$, the value $DP(i+1,t,j,r)$. This covers the case that $h_i\notin OPT'$;
\item[$\bullet$] If $j\leq n(v)$, the value $DP(i,t,j+1,r)$. This covers the case that $v_j\notin OPT'$;
\item[$\bullet$] Assume that there exists $t'\in \cT$ such that $t'-\height(h_i)\geq t$ and that $\width(h_i)\leq N-r$. Then for the minimum such $t'$ we consider the value $\profit(h_i)+DP(i+1,t',j,r)$. This covers the case that $h_i\in OPT'$, and there exists a (horizontal) guillotine cut separating $h_i$ from $OPT'\setminus \{h_i\}$.
\item[$\bullet$] Assume that there exists $r'\in \cR$ such that $r'-\width(v_j)\geq r$ and that $\height(v_j)\leq N-t$. Then for the minimum such $r'$ we consider the value $\profit(v_j)+DP(i,t,j+1,r')$. This covers the case that $v_j\in OPT'$, and there exists a (vertical) guillotine cut separating $v_j$ from $OPT'\setminus \{v_j\}$.
\end{itemize}
We observe that the above cases (which can be explored in polynomial time) cover all the possible configurations in $OPT'$. Indeed, if the first two cases do not apply, we have that $h_i,v_j\in OPT'$. Then either the line containing the right side of $v_j$ does not intersect $h_i$ (hence any other item in $OPT'$) or the line containing the top side of $h_i$ does not intersect $v_j$ (hence any other item in $OPT'$). Indeed, the only remaining case is that $v_j$ and $h_i$ overlap, which is impossible since they both belong to $OPT'$.  
\end{proof}

We will show that there exists a $(\cT,\cR)$-restricted \fontL-packing with the desired properties. 
\begin{lem}\label{lem:Lpacking:structural}
There exists a $(\cT,\cR)$-restricted \fontL-packing solution of profit at least $(1-2\eps)opt$, where the sets $\cT$ and $\cR$ have cardinality at most $n^{O(1/\eps^{1/\eps})}$ and can be computed in polynomial time based on the input (without knowing $OPT$).
\end{lem}
Lemmas \ref{lem:DPrestricted} and \ref{lem:Lpacking:structural} together immediately imply a PTAS for L-packings (showing Theorem \ref{thm:main:Lpacking}). The rest of this section is devoted to the proof of Lemma \ref{lem:Lpacking:structural}.

\begin{figure*}
\begin{centering}
\includegraphics[height=4cm]{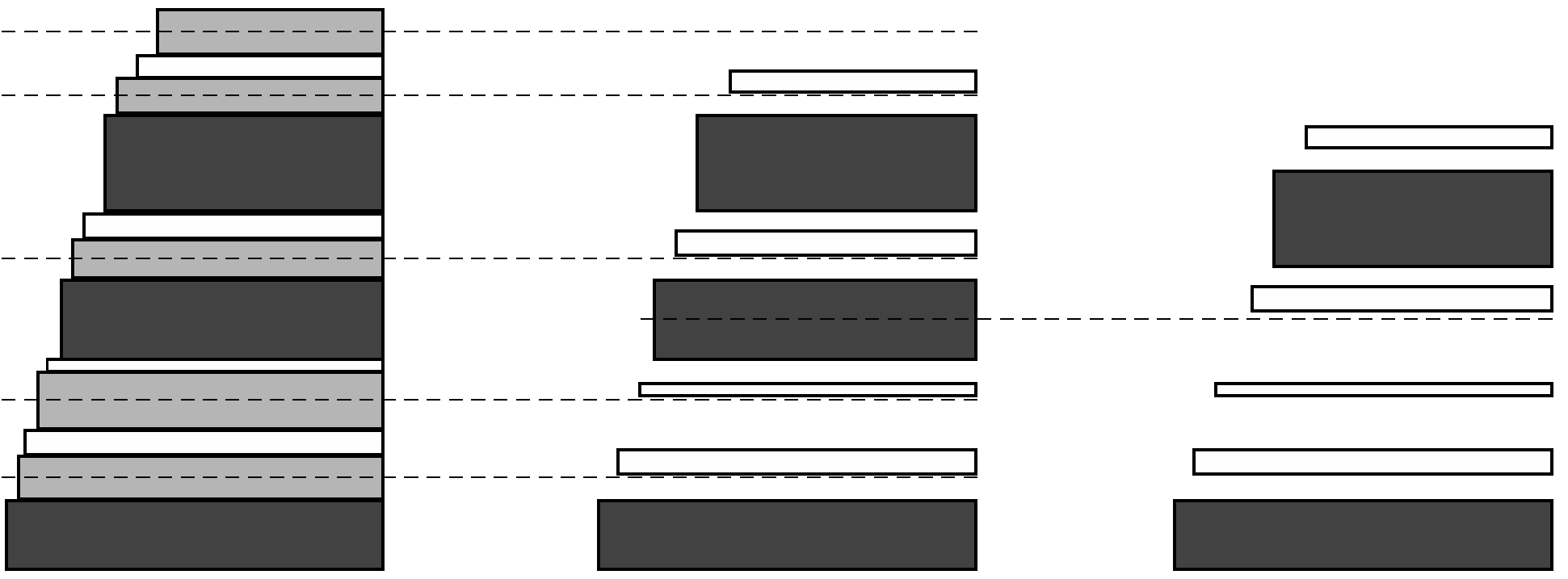}~~~~~~~~~~~~%
\par\end{centering}
\caption{Illustration of the {\tt delete\&shift} procedure with $r_{hor}=2$. The dashed lines indicate the value of the new baselines in the different stages of the algorithm. (Left) The starting packing. Dark and light grey items denote the growing sequences for the calls with $r=2$ and $r=1$, resp. (Middle) The shift of items at the end of the recursive calls with $r=1$. Note that light grey items are all deleted, and dark grey items are not shifted. (Right) The shift of items at the end of the process. Here we assume that the middle dark grey item is deleted.}\label{fig:Lpacking}
\end{figure*}

We will describe a way to delete a subset of items $D_{hor}\subseteq OPT_{hor}$ with $\profit(D_{hor})\leq 2\eps\profit(OPT_{hor})$, and \emph{shift down} the remaining items $OPT_{hor}\setminus D_{hor}$ so that their top coordinate belongs to a set $\cT$ with the desired properties. Symmetrically, we will delete a subset of items $D_{ver}\subseteq OPT_{ver}$ with $\profit(D_{ver})\leq 2\eps\profit(OPT_{ver})$, and \emph{shift to the left} the remaining items $OPT_{ver}\setminus D_{ver}$ so that their right coordinate belongs to a set $\cR$ with the desired properties. We remark that shifting down (resp. to the left) items of $OPT_{hor}$ (resp., $OPT_{ver}$) cannot create any overlap with items of $OPT_{ver}$ (resp., $OPT_{hor}$). This allows us to reason on each such set separately. 

We next focus on $OPT_{hor}$ only: the construction for $OPT_{ver}$ is symmetric. For notational convenience we let $1,\ldots,n_{hor}$ be the items of $OPT_{hor}$ in non-increasing order of width \emph{and} from bottom to top in the starting optimal packing. We remark that this sequence is not necessarily sorted (increasingly or decreasingly) in terms of item heights: this makes our construction much more complicated.


Let us first introduce some useful notation. Consider any subsequence $B=\{b_{start},\ldots,b_{end}\}$ of consecutive items (\emph{interval}). For any $i\in B$, we define $\topc_B(i):=\sum_{k\in B,k\leq i}\height(k)$ and $\bottomc_B(i)=\topc_B(i)-\height(i)$. The \emph{growing subsequence} $G=G(B)=\{g_1,\ldots,g_h\}$ of $B$ (with possibly non-contiguous items) is defined as follows. We initially set $g_1=b_{start}$. Given the item $g_i$, $g_{i+1}$ is the smallest-index (i.e., lowest) item in $\{g_i+1,\ldots,b_{end}\}$ such that $\height(g_{i+1})\geq \height(g_i)$. We halt the construction of $G$ when we cannot find a proper $g_{i+1}$. For notational convenience, define $g_{h+1}=b_{end}+1$. We let $B^G_i:=\{g_i+1,\ldots,g_{i+1}-1\}$ for $i=1,\ldots,h$. Observe that the sets $B^G_i$ partition $B\setminus G$. We will crucially exploit the following simple property.
\begin{lem}\label{lem:propertiesG}
For any $g_i\in G$ and any $j\in \{b_{start},\ldots,g_{i+1}-1\}$, $\height(j)\leq \height(g_i)$.
\end{lem}
\begin{proof}
The items $j\in B^G_i=\{g_i+1,\ldots,g_{i+1}-1\}$ have $\height(j)<\height(g_i)$. Indeed, any such $j$ with $\height(j)\geq \height(g_i)$ would have been added to $G$, a contradiction. 

Consider next any $j\in \{b_{start},\ldots g_i-1\}$. If $j\in G$ the claim is trivially true by construction of $G$. Otherwise, one has $j \in B^G_k$ for some $g_k\in G$, $g_k<g_i$. Hence, by the previous argument  and by construction of $G$, $\height(j)<\height(g_k)\leq \height(g_i)$.
\end{proof}

The intuition behind our construction is as follows. Consider the growing sequence $G=G(OPT_{hor})$, and suppose that $\profit(G)\leq \eps \cdot \profit(OPT_{hor})$. Then we might simply delete $G$, and shift the remaining items $OPT_{hor}\setminus G=\cup_j B^G_j$ as follows. Let $\lceil x\rceil_y$ denote the smallest multiple of $y$ not smaller than $x$. We consider each set $B^G_j$ separately. For each such set, we define a baseline vertical coordinate $\base_j=\lceil \bottomc(g_j)\rceil_{\height(g_j)/2}$, where $\bottomc(g_j)$ is the bottom coordinate of $g_j$ in the original packing. We next round up the height of $i\in B^G_j$ to $\hat{\height}(i):=\lceil \height(i)\rceil_{\height(g_j)/(2n)}$, and pack the rounded items of $B^G_j$ as low as possible above the baseline. The reader might check that the possible top coordinates for rounded items fall in a polynomial size set (using Lemma \ref{lem:propertiesG}). It is also not hard to check that items are \emph{not} shifted up.  

We use recursion in order to handle the case $\profit(G)> \eps \cdot \profit(OPT_{hor})$. Rather than deleting $G$, we consider each $B^G_j$ and build a new growing subsequence for each such set. We repeat the process recursively for $r_{hor}$ many rounds. Let ${\cal G}^r$ be the union of all the growing subsequences in the recursive calls of level $r$. Since the sets ${\cal G}^r$ are disjoint by construction, there must exist a value $r_{hor}\leq \frac{1}{\eps}$ such that $\profit({\cal G}^{r_{hor}})\leq \eps\cdot \profit(OPT_{hor})$. Therefore we can apply the same shifting argument to all growing subsequences of level $r_{hor}$ (in particular we delete all of them). In the remaining growing subsequences we can afford to delete $1$ out of $1/\eps$ consecutive items (with a small loss of the profit), and then apply a similar shifting argument.

We next describe our approach in more detail. 
We exploit a recursive procedure {\tt delete\&shift}. This procedure takes as input two parameters: 
an interval $B=\{b_{start},\ldots,b_{end}\}$, and an integer \emph{round parameter} $r\geq 1$. 
Procedure {\tt delete\&shift} returns a set $D(B)\subseteq B$ of deleted items, and a 
shift function $\shift:B\setminus D(B)\rightarrow \mathbb{N}$. Intuitively, $\shift(i)$ is the value of the top coordinate of $i$ in the shifted packing w.r.t. a proper baseline value which is implicitly defined.
We initially call {\tt delete\&shift}$(OPT_{hor},r_{hor})$, for a proper $r_{hor}\in \{1,\ldots,\frac{1}{\eps}\}$ to be fixed later. Let $(D,\shift)$ be the output of this call. The desired set of deleted items is $D_{hor}=D$, and in the final packing $\topc(i)=\shift(i)$ for any $i\in OPT_{hor}\setminus D_{hor}$ (the right coordinate of any such $i$ is $N$).

The procedure behaves differently in the cases $r=1$ and $r>1$.
If $r=1$, we compute the growing sequence $G=G(B)=\{g_1=b_{start},\ldots,g_h\}$, and set $D(B)=G(B)$. Consider any set $B^G_j=\{g_{j}+1,\ldots,g_{j+1}-1\}$, $j=1,\ldots,h$. Let 
$\base_j:= \lceil \bottomc_{B}(g_j) \rceil_{\height(g_j)/2}$. We define for any $i\in B^G_j$,
$$
\shift(i)= \base_j+\sum_{k\in B^G_j,k\leq i}\lceil \height(k) \rceil_{\height(g_j)/(2n)}.
$$
Observe that $\shift$ is fully defined since $\cup_{j=1}^{h}B^G_j=B\setminus D(B)$.

If instead $r>1$, we compute the growing sequence $G=G(B)=\{g_1=b_{start},\ldots,g_h\}$. We next delete a subset of items $D'\subseteq G$. If $h<\frac{1}{\eps}$, we let $D'=D'(B)=\emptyset$. Otherwise, let $G_k=\{g_j\in G: j = k \pmod{1/\eps}\}\subseteq G$, for $k\in \{0,\ldots,1/\eps-1\}$. We set $D'=D'(B)=\{d_1,\ldots,d_p\}=G_{x}$ where $x=\arg\min_{k\in \{0,\ldots,1/\eps-1\}}\profit(G_k)$. 
\begin{pro}\label{pro:deleted}
One has $\profit(D')\leq \eps\cdot \profit(G)$. Furthermore, any subsequence $\{g_x,g_{x+1},\ldots,g_y\}$ of $G$ with at least $1/\eps$ items contains at least one item from $D'$.
\end{pro}

Consider each set $B^G_j=\{g_{j}+1,\ldots,g_{j+1}-1\}$, $j=1,\ldots,h$: We run {\tt delete\&shift}$(B^G_j,r-1)$. Let $(D_j,\shift_j)$ be the output of the latter procedure, and $\shift^{max}_j$ be the maximum value of $\shift_j$. We set the output set of deleted items to $D(B)=D'\cup (\cup_{j=1}^{h}D_j)$. 

It remains to define the function $\shift$. Consider any set $B^G_j$, and let $d_q$ be the deleted item in $D'$ with largest index (hence in topmost position) in $\{b_{start},\ldots,g_{j}\}$: define $\base_q = \lceil \bottomc_B(d_q)\rceil_{\height(d_q)/2}$. If there is no such $d_q$, we let $d_q=0$ and $\base_q=0$. For any $i\in B^G_j$ we set: 
$$\begin{array}{rl}
\shift(i) & = \base_q + \sum_{g_k\in G,d_q< g_k\leq g_j}\height(g_k)
\\+&\sum_{g_k\in G,d_q\leq g_k< g_j}\shift^{max}_k+\shift_j(i). \end{array}$$
Analogously, if $g_j\neq d_q$, we set
$$\begin{array}{ll}
\shift(g_j) & = \base_q + \sum_{g_k\in G,d_q< g_k\leq g_j}\height(g_k)\\ & +\sum_{g_k\in G,d_q\leq g_k< g_j}\shift^{max}_k. \end{array}
$$ 
This concludes the description of {\tt delete\&shift}. We next show that the final packing has the desired properties. Next lemma shows that the total profit of deleted items is small for a proper choice of the starting round parameter $r_{hor}$.
\begin{lem}\label{lem:Lpacking:costDeleted}
There is a choice of $r_{hor}\in \{1,\ldots,\frac{1}{\eps}\}$ such that the final set $D_{hor}$ of deleted items satisfies $\profit(D_{hor})\leq 2\eps\cdot \profit(OPT_{hor})$.
\end{lem}
\begin{proof}
Let ${\cal G}^r$ denote the union of the sets $G(B)$ computed by all the recursive calls with input round parameter $r$. Observe that by construction these sets are disjoint.
Let also ${\cal D}^r$ be the union of the sets $D'(B)$ on those calls (the union of sets $D(B)$ for $r=r_{hor}$). By Proposition \ref{pro:deleted} and the disjointness of sets ${\cal G}^r$ one has
$$\begin{array}{ll}
\profit(D_{hor}) & =\sum_{1\leq r\leq r_{hor}}\profit({\cal D}^r) \\ &\leq \eps\cdot \sum_{r< r_{hor}}\profit({\cal G}^r)+\profit({\cal D}^{r_{hor}}) \\ & \leq \eps\cdot \profit(OPT_{hor})+\profit({\cal D}^{r_{hor}}). \end{array}
$$ 
Again by the disjointness of sets ${\cal G}^r$ (hence ${\cal D}^r$), there must exists a value of $r_{hor}\in \{1,\ldots,\frac{1}{\eps}\}$ such that $\profit({\cal D}^{r_{hor}})\leq \eps\cdot \profit(OPT_{hor})$. The claim follows. 
\end{proof}
Next lemma shows that, intuitively, items are only shifted down w.r.t. the initial packing.\fabr{Added proof of Lem \ref{lem:Lpacking:shiftDown}}
\begin{lem}\label{lem:Lpacking:shiftDown}
Let $(D,\shift)$ be the output of some execution of {\tt delete\&shift}$(B,r)$. Then, for any $i\in B\setminus D$, $\shift(i)\leq \topc_B(i)$.
\end{lem}
\begin{proof}
We prove the claim by induction on $r$. Consider first the case $r=1$. In this case, for any $i\in B^G_j$:
\begin{align*}
& \shift(i)  \\
= & \lceil \bottomc_{B}(g_j) \rceil_{\height(g_j)/2}+\sum_{k\in B^G_j,k\leq i}\lceil \height(k) \rceil_{\height(g_j)/(2n)} \\
            \leq & \topc_{B}(g_j) - \frac{1}{2}\height(g_j)+\sum_{k\in B^G_j,k\leq i}\height(k)+n\cdot \frac{\height(g_j)}{2n}\\
            = & \topc_B(i).
\end{align*}
Assume next that the claim holds up to round parameter $r-1\geq 1$, and consider round $r$. For any $i\in B^G_j$ with $\base_q = \lceil \bottomc_B(d_q)\rceil_{\height(d_q)/2}$, one has
\begin{align*}
& \shift(i) \\
= & \lceil \bottomc_B(d_q)\rceil_{\height(d_q)/2} + \sum_{g_k\in G,d_q< g_k\leq g_j}\height(g_k)\\
+ &\sum_{g_k\in G,d_q\leq g_k< g_j}\shift^{max}_k+\shift_j(i) \\
            \leq & \topc_B(d_q)+\sum_{g_k\in G,d_q< g_k\leq g_j}\height(g_k)\\
            + & \sum_{g_k\in G,d_q\leq g_k< g_j}\topc_{B^G_k}(g_{k+1}-1)+\topc_{B^G_j}(i)\\
            = & \topc_B(i).
\end{align*}
An analogous chain of inequalities shows that $\shift(g_j)\leq \topc_B(g_j)$ for any $g_j\in G\setminus D'$. A similar proof works for the special case $\base_q=0$.
\end{proof}

It remains to show that the final set of values of $\topc(i)=\shift(i)$ has the desired properties. This is the most delicate part of our analysis. We define a set $\cT^r$ of candidate top coordinates recursively in $r$. Set $\cT^1$ contains, for any item $j\in I_{hor}$, and any integer $1\leq a\leq 4n^2$, the value $a\cdot \frac{\height(j)}{2n}$. Set $\cT^r$, for $r>1$ is defined recursively w.r.t. to $\cT^{r-1}$. For any item $j$, any
integer $0\leq a\leq 2n-1$, any tuple of $b\leq 1/\eps-1$ items $j(1),\ldots,j(b)$, and any tuple of $c\leq 1/\eps$ values $s(1),\ldots,s(c)\in \cT^{r-1}$, $\cT^r$ contains the sum $a\cdot \frac{\height(j)}{2}+\sum_{k=1}^{b}\height(j(k))+\sum_{k=1}^{c}s(k)$. Note that sets $\cT^r$ can be computed based on the input only (without knowing $OPT$). It is easy to show that $\cT^r$ has polynomial size for $r=O_\eps(1)$.\fabr{Added this proof}
\begin{lem}\label{lem:sizeTr}
For any integer $r\geq 1$, $|\cT^r|\leq (2n)^{\frac{r+2+(r-1)\eps}{\eps^{r-1}}}$.
\end{lem}
\begin{proof}
We prove the claim by induction on $r$. The claim is trivially true for $r=1$ since there are $n$ choices for item $j$ and $4n^2$ choices for the integer $a$, hence altogether at most $n\cdot 4n^2<8n^3$ choices. For $r>1$, the number of possible values of $\cT^r$ is at most
\begin{align*}
& n\cdot 2n \cdot (\sum_{b=0}^{1/\eps-1}n^b)\cdot (\sum_{c=0}^{1/\eps}|\cT^{r-1}|^c)\leq 4n^2\cdot n^{\frac{1}{\eps}-1}\cdot |\cT^{r-1}|^{\frac{1}{\eps}} \\
& \leq (2n)^{\frac{1}{\eps}+1}((2n)^{\frac{r+1+(r-2)\eps}{\eps^{r-2}}})^{\frac{1}{\eps}}\leq (2n)^{\frac{r+2+(r-1)\eps}{\eps^{r-1}}}.  
\end{align*}
\end{proof}
Next lemma shows that the values of $\shift$ returned by {\tt delete\&shift} for round parameter $r$ belong to $\cT^r$, hence the final top coordinates belong to $\cT:=\cT^{r_{hor}}$.
\begin{lem}\label{lem:Lpacking:possibleHeights}
Let $(D,\shift)$ be the output of some execution of {\tt delete\&shift}$(B,r)$. Then, for any $i\in B\setminus D$, $\shift(i)\in \cT^r$.
\end{lem}
\begin{proof}
We prove the claim by induction on $r$. For the case $r=1$, recall that for any $i\in B^G_j$ one has 
\begin{align*}
\shift(i) & =  \lceil \bottomc_{B}(g_j) \rceil_{\height(g_j)/2}\\ 
&+\sum_{k\in B^G_j,k\leq i}\lceil \height(k) \rceil_{\height(g_j)/(2n)}.
\end{align*}
By Lemma \ref{lem:propertiesG}, $\bottomc_B(g_j)=\sum_{k\in B,k<g_j}\height(k)\leq (n-1)\cdot \height(g_j)$. By the same lemma, $\sum_{k\in B^G_j,k\leq i} \height(k)\leq (n-1)\cdot \height(g_j)$. It follows that 
\begin{align*}
\shift(i) &\leq 2(n-1)\cdot \height(g_j)+\frac{\height(g_j)}{2} +(n-1)\cdot \frac{\height(g_j)}{2n} \\ &\leq 4n^2\cdot \frac{\height(g_j)}{2n}.
\end{align*}
Hence $\shift(i)=a\cdot \frac{\height(g_j)}{2n}$ for some integer $1\leq a\leq 4n^2$, and $\shift(i)\in \cT^1$ for $j=g_j$ and for a proper choice of $a$.

Assume next that the claim is true up to $r-1\geq 1$, and consider the case $r$. Consider any $i\in B^G_j$, and assume $0<\base_q = \lceil \bottomc_B(d_q)\rceil_{\height(d_q)/2}$. One has:
\begin{align*}
\shift(i) & = \lceil \bottomc_B(d_q)\rceil_{\height(d_q)/2}  + \sum_{g_k\in G,d_q< g_k\leq g_j}\height(g_k)
\\& + \sum_{g_k\in G,d_q\leq g_k< g_j}\shift^{max}_k+\shift_j(i) .
\end{align*}
By Lemma \ref{lem:propertiesG}, $\bottomc_B(d_q)\leq (n-1)\height(d_q)$, therefore $\lceil \bottomc_B(d_q)\rceil_{\height(d_q)/2}=a\cdot \frac{\height(d_q)}{2}$ for some integer $1\leq a\leq 2(n-1)+1$. By Proposition \ref{pro:deleted}, $|\{g_k\in G,d_q< g_k\leq g_j\}|\leq 1/\eps-1$. Hence 
$\sum_{g_k\in G,d_q< g_k\leq g_j}\height(g_k)$ is a value contained in the set of sums of $b\leq 1/\eps-1$ item heights.  By inductive hypothesis $\shift^{max}_k,\shift_j(i)\in \cT^{r-1}$. Hence by a similar argument the value of $\sum_{g_k\in G,d_q\leq g_k< g_j}\shift^{max}_k+\shift_j(i)$ is contained in the set of sums of $c\leq 1/\eps-1+1$ values taken from $\cT^{r-1}$. 
Altogether, $\shift(i)\in \cT^r$. A similar argument, without the term $\shift_j(i)$,  shows that $\shift(g_j)\in \cT^r$ for any $g_j\in G\setminus D'$. The proof works similarly in the case $\base_q=0$ by setting $a=0$. The claim follows.
\end{proof}

\begin{proof}[Proof of Lemma \ref{lem:Lpacking:structural}]
We apply the procedure {\tt delete\&shift} to $OPT_{hor}$ as described before, and a symmetric procedure to $OPT_{ver}$. In particular the latter procedure computes a set $D_{ver}\subseteq OPT_{ver}$ of deleted items, and the remaining items are shifted to the left so that their right coordinate belongs to a set $\cR:=\cR^{r_{ver}}$, defined analogously to the case of $\cT:=\cT^{r_{hor}}$, for some integer $r_{ver}\in \{1,\ldots,1/\eps\}$ (possibly different from $r_{hor}$, though by averaging this is not critical).

It is easy to see that the profit of non-deleted items satisfies the claim by Lemma \ref{lem:Lpacking:costDeleted} and its symmetric version. Similarly, the sets 
$\cT$ and $\cR$ satisfy the claim by Lemmas \ref{lem:sizeTr} and \ref{lem:Lpacking:possibleHeights}, and their symmetric versions. Finally, w.r.t. the original packing non-deleted items in $OPT_{hor}$ and $OPT_{ver}$ can be only shifted to the bottom and to the left, resp., by Lemma \ref{lem:Lpacking:shiftDown} and its symmetric version. This implies that the overall packing is feasible.
\end{proof}

\section{A Simple Improved Approximation for Cardinality \tdk}
\label{sec:tdk_car:simple}

In this section we present a simple improved approximation
for the cardinality case of \tdk. We can assume that the optimal
solution $OPT\subseteq I$ satisfies that $|OPT|\geq1/\eps^{3}$
since otherwise we can solve the problem optimally by brute force
in time $n^{O(1/\eps^{3})}$. Therefore, we can discard from
the input all \emph{large} items with both sides larger than $\eps\cdot N$:
any feasible solution can contain at most $1/\eps^{2}$ such items,
and discarding them decreases the cardinality of $OPT$ at most by
a factor $1+\eps$. Let $OPT$ denote this slightly sub-optimal
solution obtained by removing large items.

We will need the following technical lemma, that holds also in the weighted case (see also Fig.\ref{fig:packing+ring}.(b)-(d)).

\begin{lem}\label{lem:LoftheRing} Let $H$ and $V$ be given subsets
of items from some feasible solution with width and height strictly
larger than $N/2$, resp. Let $\height_{H}$ and $\width_{V}$
be the total height and width of items of $H$ and $V$, resp.
Then there exists an $\fontL$-packing of a set $APX\subseteq H\cup V$
with $\profit(APX)\geq\frac{3}{4}(\profit(H)+\profit(V))$ into the
area $\fontL=([0,N]\times[0,\height_{H}])\cup([0,\width_{V}]\times[0,N])$.
\end{lem} 

\begin{proof} Let us consider the packing of $H\cup V$. Consider
each $i\in H$ that has no $j\in V$ to its top (resp., to its bottom)
and shift it up (resp. down) until it hits another $i'\in H$ or the
top (resp, bottom) side of the knapsack. Note that, since $\height(j)>N/2$
for any $j\in V$, one of the two cases above always applies. We iterate
this process as long as possible to move any such $i$. We perform
a symmetric process on $V$. At the end of the process all items in
$H\cup V$ are stacked on the $4$ sides of the knapsack\footnote{It is possible to permute items in the left stack so that items appear from left to right in non-increasing order of height, and symmetrically for the other stacks. This is not crucial for this proof, but we implemented this permutation in Fig.\ref{fig:packing+ring}.(c).}.

Next we remove the least profitable of the $4$ stacks: by a simple permutation
argument we can guarantee that this is the top or right stack. We
next discuss the case that it is the top one, the other case being
symmetric. We show how to repack the remaining items in a boundary
$\fontL$ of the desired size by permuting items in a proper order.
In more detail, suppose that the items packed on the left (resp.,
right and bottom) have a total width of $\width_{l}$ (resp., total
width of $\width_{r}$ and total height of $\height_{b}$). We next
show that there exists a packing into $\fontL'=([0,N]\times[0,\height_{b}])\cup([0,\width_{l}+\width_{r}]\times[0,N])$.
We prove the claim by induction. Suppose that we have proved it for
all packings into left, right and bottom stacks with parameters $\width'_{l}$,
$\width'_{r}$, and $\height'$ such that $\height'<\height_{b}$
or $\width'_{l}+\width'_{r}<\width_{l}+\width_{r}$ or $\width'_{l}+\width'_{r}=\width_{l}+\width_{r}$
and $\width'_{r}<\width_{r}$.

In the considered packing we can always find a guillotine cut $\ell$,
such that one side of the cut contains precisely one \emph{lonely}
item among the leftmost, rightmost and bottommost items. Let $\ell$
be such a cut. First assume that the lonely item $j$ is the bottommost
one. Then by induction the claim is true for the part above $\ell$
since the part of the packing above $\ell$ has parameters $\width_{l},\width_{r}$,
and $\height-\height(j)$. Thus, it is also true for the entire packing.
A similar argument applies if the lonely item $j$ is the leftmost
one.

It remains to consider the case that the lonely item $j$ is the rightmost
one. We remove $j$ temporarily and move \emph{all} other items by
$\width(j)$ to the right. Then we insert $j$ at the left (in the
space freed by the previous shifting). By induction, the claim is
true for the resulting packing since it has parameters $\width_{l}+\width(j)$,
$\width_{r}-\width(j)$, and $\height$, resp. \end{proof}

For our algorithm, we consider the following three packings.
The first uses an $L$ that occupies the full knapsack, i.e., $\width_{\fontL}=\height_{\fontL}=N$.
Let $OPT_{long}\subseteq OPT$ be the items in $OPT$ with height
or width strictly larger than $N/2$ and define $OPT_{short}=OPT\setminus OPT_{long}$.
We apply Lemma~\ref{lem:LoftheRing} to $OPT_{long}$ and hence obtain
a packing for this $L$ with a profit of at least $\frac{3}{4}\andy{p}(OPT_{long})$.
We run our PTAS for \fontL-packings from Theorem \ref{thm:main:Lpacking}
on this \fontL, the input consisting of all items in $I$ having
one side longer than $N/2$. Hence we obtain a solution with profit
at least \mbox{$(\frac{3}{4}-O(\eps))\andy{p}(OPT_{long})$}.

For the other two packings we employ the one-sided resource
augmentation PTAS from \cite{jansen2007new}. We apply this
algorithm to the slightly reduced knapsacks $[0,N]\times[0,N/(1+\eps)]$
and $[0,N/(1+\eps)]\times[0,N]$ such that in both cases it outputs
a solution that fits in the full knapsack $[0,N]\times[0,N]$ and
whose profit is by at most a factor $1+O(\eps)$ worse than the
optimal solution for the respective reduced knapsacks. We will prove
in Theorem~\ref{thm:16/9-apx} that one of these solutions yields
a profit of at least $(\frac{1}{2}-O(\eps))\profit(OPT)+(\frac{1}{4}-O(\eps))\profit(OPT_{short})$
and hence one of our packings yields a $(\frac{16}{9}+\eps)$-approximation.

\begin{thm}\label{thm:16/9-apx} There is a $\frac{16}{9}+\eps$
approximation for the cardinality case of \tdk. \end{thm} \begin{proof}
Let $OPT$ be the considered optimal solution with $opt=\profit(OPT)$.
Recall that there are no large items. Let also $OPT_{vert}\subseteq OPT$
be the (\emph{vertical}) items with height more than $\eps\cdot N$
(hence with width at most $\eps\cdot N$), and $OPT_{hor}=OPT\setminus OPT_{ver}$
(\emph{horizontal} items). Note that with this definition both sides
of a horizontal item might have a length of at most $\eps\cdot N$.
We let $opt_{long}=\profit(OPT_{long})$ and $opt_{short}=\profit(OPT_{short})$.

As mentioned above, our $\fontL$-packing PTAS achieves a profit
of at least $(\frac{3}{4}-O(\eps))opt_{long}$ which can be seen by
applying Lemma \ref{lem:LoftheRing} with $H=OPT_{long}\cap OPT_{hor}$
and $V=OPT_{long}\cap OPT_{ver}$. In order to show that the other
two packings yield a good profit, consider a \emph{random horizontal
strip} $S=[0,N]\times[a,a+\eps\cdot N]$ (fully contained in the knapsack)
where $a\in[0,(1-\eps)N)$ is chosen unformly at random. We remove
all items of $OPT$ intersecting $S$. Each item in $OPT_{hor}$ and
$OPT_{short}\cap OPT_{ver}$ is deleted with probability at most $3\eps$
and $\frac{1}{2}+2\eps$, resp. Therefore the total profit of the
remaining items is in expectation at least $(1-3\eps)\profit(OPT_{hor})+(\frac{1}{2}-2\eps)\profit(OPT_{short}\cap OPT_{vert})$.
Observe that the resulting solution can be packed into a restricted
knapsack of size $[0,N]\times[0,N/(1+\eps)]$ by shifting down the
items above the horizontal strip. Therefore, when we apply the resource
augmentation algorithm in~\cite{jansen2007new} to the knapsack $[0,N]\times[0,N/(1+\eps)]$,
up to a factor $1-\eps$, we will find a solution of (deterministically!)
at least the same profit. In other terms, this profit is at least
$(1-4\eps)\profit(OPT_{hor})+(\frac{1}{2}-\frac{5}{2}\eps)\profit(OPT_{short}\cap OPT_{vert})$.

By a symmetric argument, we obtain a solution of profit at
least $(1-4\eps)\profit(OPT_{ver})+(\frac{1}{2}-\frac{5}{2}\eps)\profit(OPT_{short}\cap OPT_{hor})$
when we apply the algorithm in~\cite{jansen2007new} to the knapsack
$[0,N/(1+\eps)]\times[0,N]$. Thus the best of the latter two solutions
has profit at least $(\frac{1}{2}-2\eps)opt_{long}+(\frac{3}{4}-\frac{13}{4}\eps)opt_{short}=(\frac{1}{2}-2\eps)opt+(\frac{1}{4}-\frac{5}{4}\eps)opt_{short}$.
The best of our three solutions has therefore value at least $(\frac{9}{16}-O(\eps))opt$
where the worst case is achieved for roughly $opt_{long}=3\cdot opt_{short}$.
\end{proof}


In the above result we use either an L-packing or a container packing. The $\frac{558}{325}+\eps$ approximation claimed in Theorem \ref{trm:tdk_car:refined} is obtained by a careful combination of these two packings. In particular, we consider configurations where long items (or a subset of them) can be packed into a relatively small $L$, and pack part of the remaining short items in the complementary rectangular region (using container packings and Steinberg's algorithm \cite{steinberg1997strip}). See Section \ref{sec:tdk_car:refined} for details.

\section{Open Problems}

\fabr{Reformulated open problems in question form to be more positive}\andyr{If we need more space we can change the paper format to A4 (it did not say anything about this in the call) \fab{Nothing changes. We are really tight at the moment}}
The main problem that we left open is to find a PTAS, if any, for \tdk and \tdkr. This would be interesting even in the cardinality case. We believe that a better understanding of natural generalizations of \fontL-packings might be useful. For example, is there are PTAS for \emph{ring-packing} instances arising by shifting of long items? This would directly lead to an improved approximation factor for \tdk (though not \fab{to} a PTAS). Is there a PTAS for \fontL-packings \emph{with rotations}? Our improved approximation algorithms for \tdkr\ are indeed based on a different approach. Is there a PTAS for $O(1)$ instances of \fontL-packing? This would also lead to an improved approximation factor for \tdk, and might be an important step towards a PTAS. 
\newpage



\bibliographystyle{plain}
\bibliography{bibliography}

\appendix


\input{2Dknapsack8f-Weighted-case}

\input{2Dknapsack8f-cardinality-No-Rotation}

\newpage
\input{2Dknapsack8f-With-rotations}


\input{2Dknapsack8f-Weighted-case-Rotations}


\input{2Dknapsack8f-Tools}
\input{2Dknapsack9a-Resource-augmentation}

\end{document}

%% file: 2Dknapsack8f-Weighted-case.tex
\section{Weighted Case Without Rotations\label{sec:weighted}}

In this section we show how to extend the reasoning of the unweighted
case to the weighted case. This requires much more complicated technical
machinery than the algorithm presented in Section~\ref{sec:tdk_car:simple}.

Our strategy is to start with a partition of the knapsack into thin
corridors as defined in~\cite{adamaszek2015knapsack}. Then, we partition
these corridors into a set of rectangular boxes and an L-packing.
We first present a simplified version of our argumentation in which
we assume that we are allowed to drop $O_{\eps}(1)$ many items
at no cost, i.e., we pretend that we have the right to remove $O_{\eps}(1)$
items from $\opt$ and compare the profit of our computed solution
with the remaining set. Building on this, we give an argumentation
for the general case which will involve some additional shifting arguments.

\subsection{Item classification}

We start with a classification of the input items according to their
heights and widths. For two given constants $1\geq\epsl>\epss>0$,
we classify an item $i$ as:\fabr{We should choose between item
and rectangle} 

\itemsep0pt 
\begin{itemize}
\item[$\bullet$] \emph{small} if $h_{i},w_{i}\leq\epss N$; 
\item[$\bullet$] \emph{large} if $h_{i},w_{i}>\epsl N$; 
\item[$\bullet$] \emph{horizontal} if $w_{i}>\epsl N$ and $h_{i}\leq\epss N$; 
\item[$\bullet$] \emph{vertical} if $h_{i}>\epsl N$ and $w_{i}\leq\epss N$; 
\item[$\bullet$] \emph{intermediate} otherwise (i.e., at least one side has length
in $(\epss N,\epsl N]$). 
\end{itemize}
We also call \emph{skewed} items that are horizontal or vertical.
We let $\Rsm$, $\Rla$, $\Rho$, $\Rve$, $\Rsk$, and $\Rin$ be
the items which are small, large, horizontal, vertical, skewed, and
intermediate, respectively. The corresponding intersection with $\opt$
defines the sets $\optsm$, $\optla$, $\optho$, $\optve$, $\optsk$,
$\optin$, respectively. 

Observe that $|\optla|=O(1/\epsl^{2})$ and since we are allowed to
drop $O_{\eps}(1)$ items from now on we ignore $\optla$. The
next lemma shows that we can neglect also $\optin$. 
\begin{lem}
\label{lem:item-classification}For any constant $\eps>0$ and \fab{positive}
increasing function $f(\cdot)$, \fab{$f(x)>x$,} there exist constant
values $\epsl,\epss$, with $\eps\geq\epsl\geq\fab{f(\epss)}\ge\Omega_{\eps}(1)$
and $\epss\in\Omega_{\eps}(1)$ such that the total profit of intermediate
rectangles is bounded by $\eps p(OPT)$. \andy{The pair $(\epsl,\epss)$
is one pair from a set of $O_\eps(1)$ pairs and this set can be computed in polynomial time.}
\end{lem}
\begin{proof}
Define $k+1=2/\eps+1$ constants $\eps_{1},\ldots,\eps_{k+1}$, \san{such that}
\andy{$\eps=f(\eps_{1})$} and $\eps_{i}=f(\eps_{i+1})$ \andy{for
each~$i$}. Consider the $k$ ranges of widths and heights of type
$(\eps_{i+1}N,\eps_{i}N]$. By an averaging argument there exists
one index $j$ such that the total profit of items in $\opt$ with
at least one side length in the range $(\eps_{j+1}N,\eps_{j}N]$ is
at most $2\frac{\eps}{2}p(\opt)$. It is then sufficient to set $\epsl=\eps_{j}$
and $\epss=\eps_{j+1}$. 
\end{proof}
We transform now the packing of the optimal solution $\opt$. To this
end, we temporarily remove the small items $\optsm$. We will add
them back later. Thus, the reader may now assume that we need to pack
only the skewed items from $\optsk$.

\subsection{Corridors, Spirals and Rings}

We build on a partition of the knapsack into corridors as used in
\cite{adamaszek2015knapsack}. We define an \emph{open corridor} to
be a face on the 2D-plane bounded by a simple rectilinear polygon
with $2k$ edges $e_{0},\ldots,e_{2k-1}$ for some integer $k\geq2$,
such that for each pair of horizontal (resp., vertical) edges $e_{i},e_{2k-i}$,
$i\in\{1,...,k-1\}$ there exists a vertical (resp., horizontal) line
segment $\ell_{i}$ such that both $e_{i}$ and $e_{2k-i}$ intersect
$\ell_{i}$ and $\ell_{i}$ does not intersect any other edge. Note
that $e_{0}$ and $e_{k}$ are not required to satisfy this property:
we call them the \emph{boundary edges} of the corridor. Similarly
a \emph{closed corridor} (or \emph{cycle}) is a face on the 2D-plane
bounded by two simple rectilinear polygons defined by edges $e_{0},\ldots,e_{k-1}$
and $e'_{0},\ldots,e'_{k-1}$ such that the second polygon is contained
inside the first one, and for each pair of horizontal (resp., vertical)
edges $e_{i},e'_{i}$, $i\in\{0,...,k-1\}$, there exists a vertical
(resp., horizontal) line segment $\ell_{i}$ such that both $e_{i}$
and $e'_{i}$ intersect $\ell_{i}$ and $\ell_{i}$ does not intersect
any other edge. See Figures \san{\ref{fig:packing}} and \andy{\ref{fig:corridors}}
for examples. Let us focus on \fab{minimum} length such $\ell_{i}$'s:
then the \emph{width} $\alpha$ of the corridor is the \fab{maximum}
length of any such $\ell_{i}$. We say that an open (resp., closed)
corridor of the above kind has $k-2$ (resp., $k$) \emph{bends}.
A corridor partition is a partition of the knapsack into corridors.

\begin{figure}
\centering
\includegraphics[width=.25\textwidth]{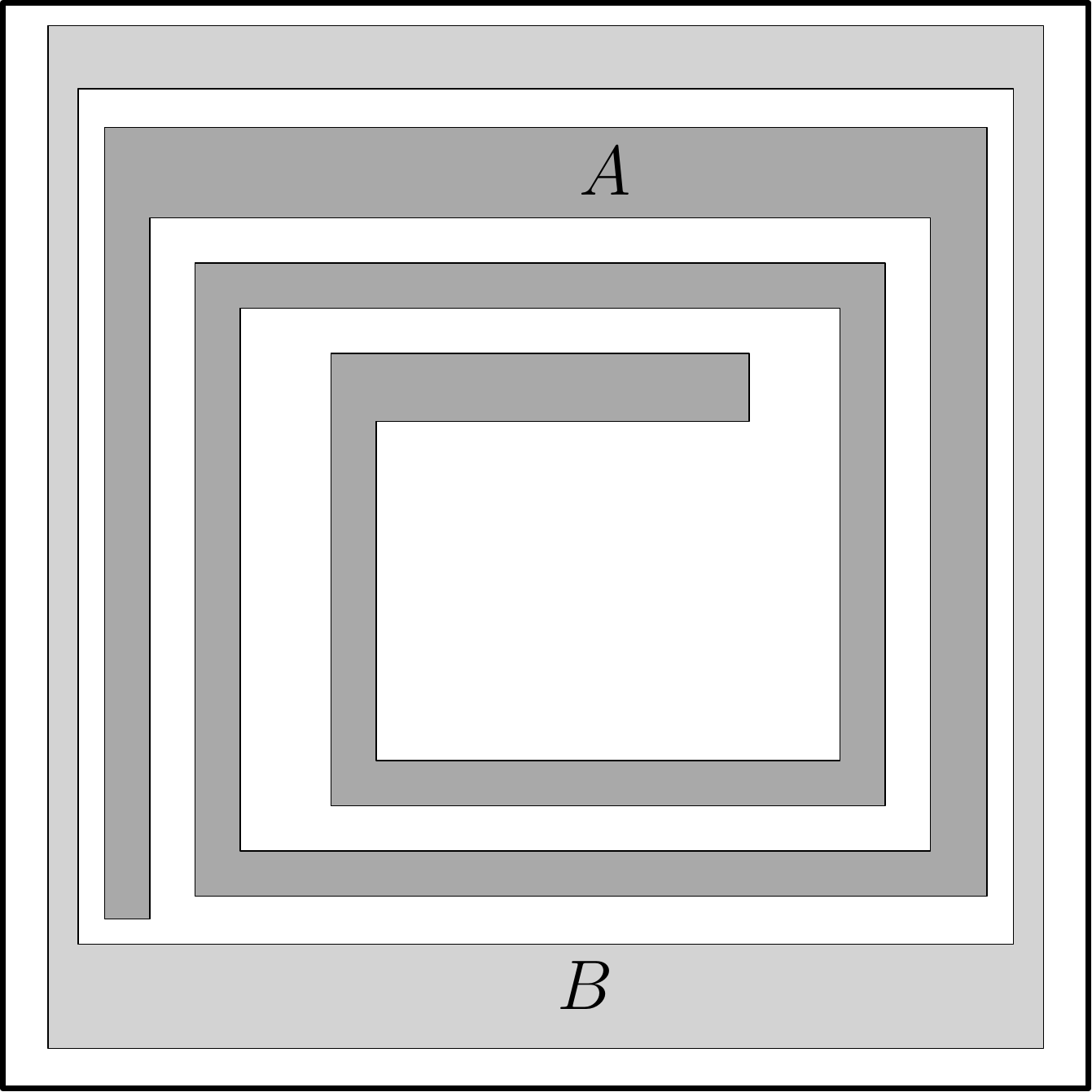}
\caption{Illustration of two specific types of corridors: spirals (A) and rings (B).\label{fig:packing}.}
\end{figure}
 
\begin{lem}[Corridor Packing Lemma \cite{adamaszek2015knapsack}]
\label{lem:corridorPack-weighted}There exists a corridor partition
and \andy{a set} $\optco\subseteq\optsk$ such that: 
\end{lem}
\begin{enumerate}
\item there is a subset \san{$\optco^{cross}\subseteq\optco$} with \san{$|\optco^{cross}|\le O_{\eps}(1)$}
such that each item \san{$i\in\optco\setminus\optco^{cross}$} is fully contained
in some corridor, 
\item $p(\optco)\geq(1-O(\eps))p(\andy{\optsk})$, 
\item the number of corridors is $O_{\eps,\epsl}(1)$ and each corridor
has width at most $\epsl N$ and has at most $1/\eps$ bends. 
\end{enumerate}
Since we are allowed to drop $O_{\eps}(1)$ items from now on
we ignore \san{$\optco^{cross}$}. We next identify some structural properties
of the corridors that are later exploited in our analysis. 
Observe that an open (resp., closed) corridor of the above type is
the union of $k-\san{1}$ (resp., $k$) boxes, that we next call \emph{subcorridors}
(see also Figure \ref{fig:corridors}). Each such box is a maximally
large rectangle that is contained in the corridor. The subcorridor
$S_{i}$ of an open (resp., closed) corridor of the above kind is
the one containing edges $e_{i},e_{2k-i}$ (resp., $e_{i},e_{i'}$)
on its boundary. The length of $S_{i}$ is the \emph{length} of the
shortest such edge. We say that a subcorridor is \emph{long} if its
length is more than $N/2$, and \emph{short} otherwise. The partition
of subcorridors into short and long will be crucial in our analysis.

\fabr{Should add illustration of corridor width, horizontal/vertical,
(counter)clockwise, U/Z-bent in the figure (possible with same space
usage)}\andyr{Added description of horizontal/vertical, (counter)clockwise,
U/Z-bent, and corridor width in the caption.} We call a subcorridor
\emph{horizontal} (resp., \emph{vertical}) if the corresponding edges
are so. Note that each rectangle in $\optco$ is \fab{univocally
associated} with the only subcorridor that fully contains it: indeed,
the longer side of a skewed rectangle is longer than the width of
any corridor. Consider the sequence of consecutive subcorridors $S_{1},\ldots,S_{k'}$
of an open or closed corridor. Consider two consecutive corridors
$S_{i}$ and $S_{i'}$, with $i'=i+1$ in the case of an open corridor
and $i'=(i+1)\pmod{k'}$ otherwise. \san{ First assume that $S_{i'}$ is horizontal. }We say that $S_{i'}$ is to the
right (resp., left) of $S_{i}$ \san{if the right-most (left-most) boundary of $S_{i'}$ is to the right (left) of the right-most (left-most) boundary of $S_i$. If instead $S_{i'}$ is vertical, then $S_i$ must be horizontal and we say that $S_{i'}$ is to the right (left) of $S_i$ if $S_i$ is to the left (right) of $S_{i'}$. Similarly, if $S_{i'}$ is vertical, we say that $S_{i'}$ is above (below) $S_i$ if the top (bottom) boundary of $S_{i'}$ is above (below) the top (bottom) boundary of $S_i$. If $S_{i'}$ is horizontal, we say that it is above (below) $S_i$ if $S_i$ (which is vertical) is below (above) $S_{i'}$.} 
We say that the pair $(S_{i},S_{i'})$ forms a clockwise
bend \fab{if $S_{i}$ is horizontal and $S_{i'}$ is to its bottom-right
or top-left, and the complementary cases if $S_{i}$ is vertical.}
In all the other cases the pairs \ari{form} a counter-clockwise bend. Consider
a triple $(S_{i},S_{i'},S_{i''})$ of consecutive subcorridors in
the above sense. It forms a $U$-bend if $(S_{i},S_{i'})$ and $(S_{i'},S_{i''})$
are both clockwise or counterclockwise bends. Otherwise it forms a
\fab{$Z$}-bend. In both cases $S_{i'}$ is the \emph{center} of
the bend, and $S_{i},S_{i''}$ its \emph{sides}. An open corridor
whose bends are all clockwise (resp., counter-clockwise) is a \emph{spiral}.
A closed corridor with $k=4$ is a \emph{ring}. Note that in a ring
all bends are clockwise or counter-clockwise, hence in some sense
it is the closed analogue of a spiral. \fab{We remark that a corridor
whose subcorridors are all long is a spiral or a ring}\footnote{We leave the simple proof for the ring case to the reader since we
do not explicitly need this claim.}. As we will see, spirals and rings play a crucial role in our analysis.
In particular, we will exploit the following simple fact.\fabr{Changed
into lemma plus proof} 

\begin{lem}\label{lem:spiral} The following properties hold: 

\itemsep0pt 
\begin{enumerate}
\item \label{claim:Zbend} The two sides of a $Z$-bend cannot be long.
In particular, an open corridor whose subcorridors are all long is
a spiral. 
\item \label{claim:Ubends} A closed corridor contains at least $4$ distinct
(possibly overlapping) $U$-bends. 
\end{enumerate}
\end{lem} \begin{proof} \eqref{claim:Zbend} By definition of long
subcorridors and $Z$-bend, the $3$ subcorridors of the $Z$-bend
would otherwise have total width or height larger than $N$. \eqref{claim:Ubends}
Consider the left-most and right-most vertical subcorridords, and
the top-most and bottom-most horizontal subcorridors. These $4$ subcorridors
exist, are distinct, and are centers of a $U$-bend. \end{proof}
\fabr{About Fig. \ref{fig:corridors}. Left: if you add $S_{i-1}$
and $S_{i+1}$ (maybe using dashed areas or similar tricks), we can
define $U/Z$-bends. Right: hard to see the split. Maybe just $3/4$
cutting lines would be sufficient} \andyr{Left: Added description
of U/Z-bends and names of other subcorridors. Right: now only three
boxes, I think it is easier to see now what happens.}

\subsection{Partitioning Corridors into Rectangular Boxes\label{sec:structural:boxes}}

We next describe a routine to partition the corridors into rectangular
boxes such that each item is contained in one such box. We remark
that to achieve this partitioning we sometimes have to sacrifice a
large fraction of $\optco$, hence we do not achieve a $1+\eps$ approximation
as in \cite{AW2013}. On the positive side, we generate only a constant
(rather than polylogarithmic) number of boxes. This is crucial to
obtain a polynomial time algorithm in the later steps.

Recall that each $i\in\optco$ is univocally associated with the only
subcorridor that fully contains it. 
\ari{We will say that we \emph{delete} a sub-corridor, when we delete all rectangles univocally associated with the subcorridor. Note that in deletion of a sub-corridor we do not delete rectangles that are partially contained in that subcorridor but  completely contained in a neighbor sub-corridor.}
\arir{{\bf One of the reviewers wanted more explanation on deletion.}}
Given a corridor, we sometimes
\emph{delete} some of its subcorridors, and consider the \emph{residual}
corridors (possibly more than one) given by the union of the remaining
subcorridors. 
Note that removing any subcorridor from a closed corridor
turns it into an open corridor. We implicitly assume that items associated
with a deleted subcorridor are also removed (and consequently the
corresponding area can be used to pack other items).

\begin{figure}
\begin{centering}
\includegraphics[height=2cm]{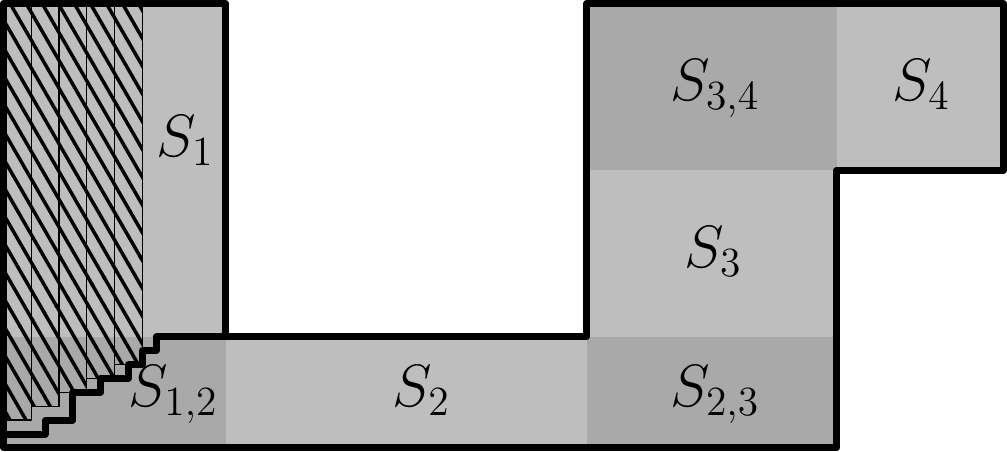}~~~~\includegraphics[height=2cm]{peel-corridor} 
\par\end{centering}
\caption{\label{fig:corridors}Left: \andy{The subcorridors $S_{1}$ and
$S_{3}$ are vertical, $S_{2}$ and $S_{4}$ are horizontal.} The
subcorridor $S_{3}$ is on the top-right of $S_{2}$. \andy{ The
curve on the bottom left shows the boundary curve between $S_{1}$
and $S_{2}$. The pair $(S_{3},S_{4})$ forms a clockwise bend and
the pair $(S_{2},S_{3})$ forms a counter-clockwise bend. The triple
$(S_{1},S_{2},S_{3})$ forms a $U$-bend and the triple $(S_{2},S_{3},S_{4})$
forms a $Z$-bend.} Right: Our operation that divides a corridor
into $O_{\eps}(1)$ boxes and $O_{\eps}(1)$ shorter corridors. The
dark gray items show thin items that are removed in this operation.
\andy{The light gray items are fat items that are shifted to the
box below their respective original box.} The value $\alpha$ denotes
the width of the depicted corridor. }
\end{figure}

Given two consecutive subcorridors $S_{i}$ and $S_{i'}$, we define
the \emph{boundary curve} among them as follows (see also Figure \ref{fig:corridors}).
Suppose that $S_{i'}$ is to the top-right of $S_{i}$, the other
cases being symmetric. Let $S_{i,i'}=S_{i}\cap S_{i'}$ be the rectangular
region shared by the two subcorridors. Then the boundary curve among
them is any simple rectilinear polygon inside $S_{i,i'}$ that decreases
monotonically from its top-left corner to its bottom-right one and
that does not cut any rectangle in these subcorridors. For a boundary
horizontal (resp., vertical) subcorridor of an open corridor (i.e.,
a subcorridor containing $e_{0}$ or $e_{2k-1}$) we define a dummy
boundary curve given by the vertical (resp., horizontal) side of the
subcorridor that coincides with a (boundary) edge of the corridor. 

\begin{remark}\label{rem:boundaryCurves} Each subcorridor has two
boundary curves (including possibly dummy ones). Furthermore, all
its items are fully contained in the region delimited by such curves
plus the two edges of the corridor associated with the subcorridor
\emph{(private region)}. \end{remark} 

Given a corridor, we partition its area into a constant number of
boxes as follows (see also Figure \ref{fig:corridors}, and \cite{AW2013}
for a more detailed description of an analogous construction). Let
$S$ be one of its boundary subcorridors (\san{if} any), or the central
subcorridor of a $U$-bend. Note that one such $S$ must exist (trivially
for an open corridor, otherwise by Lemma \ref{lem:spiral}.\ref{claim:Ubends}).
In the corridor partition, there might \ari{be} several subcorridors fulfilling
the latter condition. We will explain later in which order to process
the subcorridors, here we explain only how to apply our routine to
\emph{one} subcorridor, which we call \emph{processing} of subcorridor. 

Suppose that $S$ is horizontal with height $b$, with the shorter
horizontal associated edge being the top one. The other cases are
symmetric. Let $\epst>0$ be a sufficiently small constant to be defined
later. If $S$ is the only subcorridor in the considered corridor,
\fab{$S$ forms a box and all its items are marked as \emph{fat}}.
Otherwise, we draw $1/\epst$ horizontal lines that partition the
private region of $S$ into subregions of \ari{height} $\epst b$. We mark
as \emph{thin} the items of the bottom-most (i.e., the widest) such
subregion, and as \emph{killed} the items of the subcorridor cut by
these horizontal lines. 
All the remaining items of the subcorridor are marked as \emph{fat}.

For each such subregion, we define an associated (horizontal) box
as the largest axis-aligned box that is contained in the subregion.
Given these boxes, we partition the rest of the corridor into $1/\epst$
corridors as follows. Let $S'$ be a corridor next to $S$, say to
its top-right. Let $P$ be the set of corners of the boxes contained
in the boundary curve between $S$ and $S'$. We project $P$ vertically
on the boundary curve of $S'$ not shared with $S$, hence getting
a set $P'$ of $1/\fab{\epst}$ points. We iterate the process on
the pair $(S',P')$. At the end of the process, we obtain a set of
$1/\epst$ boxes from the starting subcorridor $S$, plus a collection
of $1/\epst$ new (open) corridors each one having one less bend with
respect to the original corridor. \san{Later, we will also apply this process on the latter
corridors.} \andy{Each newly created corridor will have one bend less than the original corridor 
and thus this process eventually terminates.}
Note that, since initially there are $O_{\eps,\epsl}(1)$ corridors
each one with $O(1/\eps)$ bends, the final number of boxes is $O_{\eps,\epsl,\epst}(1)$.
See Figure \ref{fig:corridors} for an illustration.

\begin{remark} Assume that we execute the above procedure on the
subcorridors until there is no subcorridor left on which we can apply
it. Then we obtain a partition of $\optco$ into disjoint sets $\optth$,
$\optfa$, and $\optki$ of thin, fat, and killed items, respectively.
Note that each order to process the subcorridors leads to different
such partition. We will define this order carefully in our analysis.
\end{remark} 
\begin{remark}\label{rem:fatPack} By a simple shifting
argument, there exists a packing of $\optfa$ into the boxes. Intuitively,
in the above construction each subregion \andy{is} fully contained
in the box associated with the subregion immediately below (when no
lower subregion exists, the corresponding items are thin). \end{remark}

\san{We will from now on assume that the shifting of items as described in 
Remark \ref{rem:fatPack} has been done.}\sanr{A reviewer was complaining that
it is not explicitly clear whether we do this or not.}

The following lemma summarises some of the properties of the boxes
and of the associated partition of $\optco$ (independently from the
way ties are broken). \andy{Let $\Rho$ and $\Rve$ denote the set
of horizontal and vertical input items, respectively.} \begin{lem}\label{lem:boxProperties}
The following properties hold: 

\itemsep0pt 
\begin{enumerate}
\item \label{lem:boxProperties:profitKill} \fab{$|\optki|=O_{\eps,\epsl,\epst}(1)$};
\item \label{lem:boxProperties:thin} For any given constant $\epsr>0$
\andy{there is} a sufficiently small $\epst>0$ \andy{such that}
the total height (resp., width) of items in $\optth\cap\Rho$ (resp.,
$\optth\cap\Rve$) is at most $\epsr N$. 
\end{enumerate}
\end{lem} \begin{proof} \eqref{lem:boxProperties:profitKill} Each
horizontal (resp., vertical) line in the construction can kill at
most $1/\epsl$ items, since those items must be horizontal (resp.,
vertical). Hence we kill $O_{\eps,\epsl,\epst}(1)$ items in total.

\eqref{lem:boxProperties:thin} The mentioned total height/width is
at most $\epst N$ times the number of subcorridors, which is $O_{\eps,\epsl}(1)$.
The claim follows for $\epst$ small enough. \end{proof} 

\subsection{Containers\label{sec:structural:containers}}

Assume that we applied the routine described in Section~\ref{sec:structural:boxes}
above until each corridor is partitioned into boxes. We explain how
to partition each box into $O_{\eps}(1)$ subboxes, to which we
refer to as \emph{containers} in the sequel. Hence, we apply the routine
described below to each box.

Consider a box of size $a\times b$ coming from the above construction,
and on the associated set $\andy{\optbo}$ of items from $\optfa$.
We will show how to pack a set $\andy{\optbo'}\subseteq\andy{\optbo}$
with $\andy{p(\optbo')}\geq(1-\eps)\andy{p(\optbo)}$ into $O_{\eps}(1)$
containers packed inside the box, such that both the containers and
the packing of $\andy{\optbo'}$ inside them satisfy some extra properties
that are useful in the design of an efficient algorithm. This part
is similar in spirit to prior work, though here we present a refined
analysis that simplifies the algorithm (in particular, we can avoid
LP rounding).

A \emph{container} is a box labelled as \emph{horizontal}, \emph{vertical},
or \emph{area}. A \emph{container packing} of a set of items $I'$
into a collection of non-overlapping containers has to satisfy the
following properties: 

\itemsep0pt 
\begin{itemize}
\item[$\bullet$] Items in a horizontal (resp., vertical) container are stacked one
on top of the other (resp., one next to the other). 
\item[$\bullet$] Each $i\in I'$ packed in an area container of size $a\times b$
must have $w_{i}\leq\eps a$ and $h_{i}\leq\eps b$. 
\end{itemize}
Our main building block is the following resource augmentation packing
lemma essentially taken from \cite{Jansen2009310}\footnote{In Appendix \ref{sec:resource-augmentation} we reprove this lemma
in a \emph{container-based} form, rather than using LP-based arguments,
since this is more convenient for our final algorithm. Our version
of the lemma might also be a handy tool for future work.}. 

\begin{lem}[Resource Augmentation Packing Lemma]
\label{lem:augmentPack} Let $I'$ be a collection of items that
can be packed into a box of size $a\times b$, and $\epsau>0$ be
a given constant. Then there exists a container packing of $I''\subseteq I'$
inside a box of size $a\times(1+\epsau)b$ (resp., $(1+\epsau)a\times b$)
such that: 
\end{lem}\sanr{Why don't we use our enhanced version of this lemma here right away? Is there any reason for keeping the simpler version?}
\itemsep0pt 
\begin{enumerate}
\item $p(I'')\geq(1-\sal{O(\epsau)})p(I')$; 
\item The number of containers is \fab{$O_{\epsau}(1)$} and their sizes
belong to a set of cardinality \fab{$n^{O_{\epsau}(1)}$} that
can be computed in polynomial time. 
\end{enumerate}
Applying Lemma~\ref{lem:augmentPack} to each box yields the following
lemma.
\begin{lem}[Container Packing Lemma]
\label{lem:containerPack} For a given constant \fab{$\epsau>0$},
there exists a set $\san{\optfa^{cont}}\subseteq\optfa$ such that there is a
container packing for all apart from $O_{\eps}(1)$ items in \san{$\optfa^{cont}$}
such that: 
\end{lem}
\begin{enumerate}
\item \fab{$p(\san{\optfa^{cont}})\geq(1-O(\eps))p(\optfa)$}; 
\item The number of containers is $O_{\eps,\epsl,\epst,\epsau}(1)$ and
their sizes belong to a set of cardinality \fab{$n^{O_{\eps,\epsl,\epst,\epsau}(1)}$}
that can be computed in polynomial time.\andyr{Shouldn't the number
of sizes depend on $\log\max{h_{i},w_{i}}$ somehow? \fab{Not sure}}
\end{enumerate}
\begin{proof} Let us focus on a specific box of size $a\times b$
from the previous construction in Section~\ref{sec:structural:boxes}, and on the items $\andy{\optbo}\subseteq\optfa$
inside it. \andy{If $|\andy{\optbo}|=O_{\eps}(1)$ then we can simply
create one container for each item and we are done.} 
\andyr{It used to say ``we can assume $|\andy{\optbo}|\gg1$ otherwise
we can discard $\andy{\optbo}$ with small loss in the profit. ''.
I believe that it is cleaner to say that we generate one container
for each item if the box contains only few items.}\andy{Otherwise,
assume} w.l.o.g. that this box (hence its items) is horizontal. We
obtain a set $\andy{\overline{\optbo}}$ by removing from $\andy{\optbo}$
all items intersecting a proper horizontal strip of height $\fab{3}\eps b$.
Clearly these items can be repacked in a box of size $a\times(1-3\eps)b$.
By a simple averaging argument, it is possible to choose the strip
so that the items fully contained in it have total profit at most
$O(\eps)p(\andy{\optbo})$. Furthermore, there can be at most $O(1/\epsl)$
items that partially overlap with the strip (since items are skewed).
We drop these items and do not pack them. %

At this point we can use the Resource Augmentation Lemma \ref{lem:augmentPack}
to pack a large profit subset $\andy{\optbo'}\subseteq\andy{\overline{\optbo}}$
into $O_{\epsau}(1)$ containers that can be packed in a box of size
$a\times(1-3\eps)(1+\epsau)b\leq a\times(1-2\eps)b$. 
\andy{We perform the above operation on each box of the previous
construction 
and define \san{$\optfa^{cont}$} to be the union} of the \andy{respective}
sets $\andy{\optbo'}$. The claim follows. \end{proof}

\subsection{\san{A} Profitable Structured Packing\label{sec:structural:lemma}}

We next prove our main structural lemma which yields that there exists
a structured packing which is partitioned into $O_{\eps}(1)$
containers and an L. We will refer to such a packing as an L\&C packing
(formally defined below). Note that in the previous section we did
not specify in which order we partition the subcorridors into boxes.
In this section, we give several such orders which will then result
in different packings. The last such packing is special since we will
modify it a bit to gain some space and then reinsert the thin items
that were removed in the process of partitioning the corridors into
containers. Afterwards, we will show that one of the resulting packings
will yield an approximation ratio of $17/9+\eps$.

A \emph{boundary ring} of width $N'$ is a ring having as external
boundary the \andy{edges} of the knapsack and as internal boundary
the boundary of a square box of size $(N-N')\times(N-N')$ in the
middle of the knapsack. A \emph{boundary $L$} of width $N'$ is the
region covered by two boxes of size $N'\times N$ and $N\times\fab{N'}$
that are placed on the left and bottom boundaries of the knapsack.

An \emph{L\&C} packing is defined as follows. We are given two integer
parameters $N'\in[0,N/2]$ and $\ell\in\andy{(}N/2,N]$. We define
a boundary $L$ of width $N'$, and a collection of non-overlapping
containers contained in the space \andy{not} occupied by the boundary
$L$. The number of containers and their sizes are as in Lemma \ref{lem:containerPack}.
We let $\san{\ilong}\subseteq I$ be the items whose longer side has length
longer than $\ell$ (hence longer than $N/2$), and $\san{\ishort}=I\setminus \san{\ilong}$
be the remaining items. We can pack only items from $\san{\ilong}$ in the
boundary $L$, and only items from $\san{\ishort}$ in the containers (satisfying
the usual container packing constraints). See also Figure \ref{fig:packing+ring}. 
\begin{remark}\label{rem:degenerateRing}
In the analysis sometimes we will not need the boundary $L$. This
case is captured by setting $N'=0$ and $\ell=N$ (\emph{degenerate
$L$} case). \end{remark} 

\begin{lem}
\label{lem:apxNoRotation} Let $\optrc$ \ari{be} the most profitable solution
that is packed by an L\&C packing. Then $p(\optrc)\geq(\frac{9}{17}-O(\eps))p(\opt)$. 
\end{lem}
In the remainder of this section we prove Lemma~\ref{lem:apxNoRotation},
assuming that we can drop $O_{\eps}(1)$ items at no cost. Hence,
formally we will prove that there is an L\&C packing $I'$ and a set
of $O_{\eps}(1)$ items $I_{\mathrm{drop}}$ such that $p(I')+p(I_{\mathrm{drop}})\ge(\frac{9}{17}-O(\eps))p(\opt)$.
Subsequently, we will prove Lemma~\ref{lem:apxNoRotation} in full
generality (without dropping any items).

 The proof of Lemma~\ref{lem:apxNoRotation} involves some case analysis.
Recall that we classify subcorridors into short and long, and horizontal
and vertical. We further partition short subcorridors as follows:
let $S_{1},\ldots,S_{k'}$ be the subcorridors of a given corridor,
and let $S_{1}^{s},\ldots,S_{k''}^{s}$ be the subsequence of short
subcorridors (if any). Mark $S_{i}^{s}$ as \emph{even} if $i$ is
so, and \emph{odd} otherwise. Note that \san{corridors} are subdivided
into several other \san{corridors} during the box construction process \san{(see the right side of Figure \ref{fig:corridors})}, 
\andy{and these new corridors might have fewer subcorridors than the initial corridor.}
\san{However, the marking of the subcorridors (short, long, even, odd, horizontal, vertical) is inherited from} the marking of the original subcorridor. 

We will describe now 7 different ways to partition the subcorridors into
boxes, for some of them we delete some of the subcorridors. \san{Each of these different processing orders will give different sets $\optth, \optki$ and $\optfa^{cont}$, and based on these, we will partition the items into three sets. We will then prove three different lower bounds on $p(\optrc)$ w.r.t. the sizes of these three sets using averaging arguments about the seven cases.}

\begin{figure*}[t!]
    \centering
	\includegraphics[width=\textwidth]{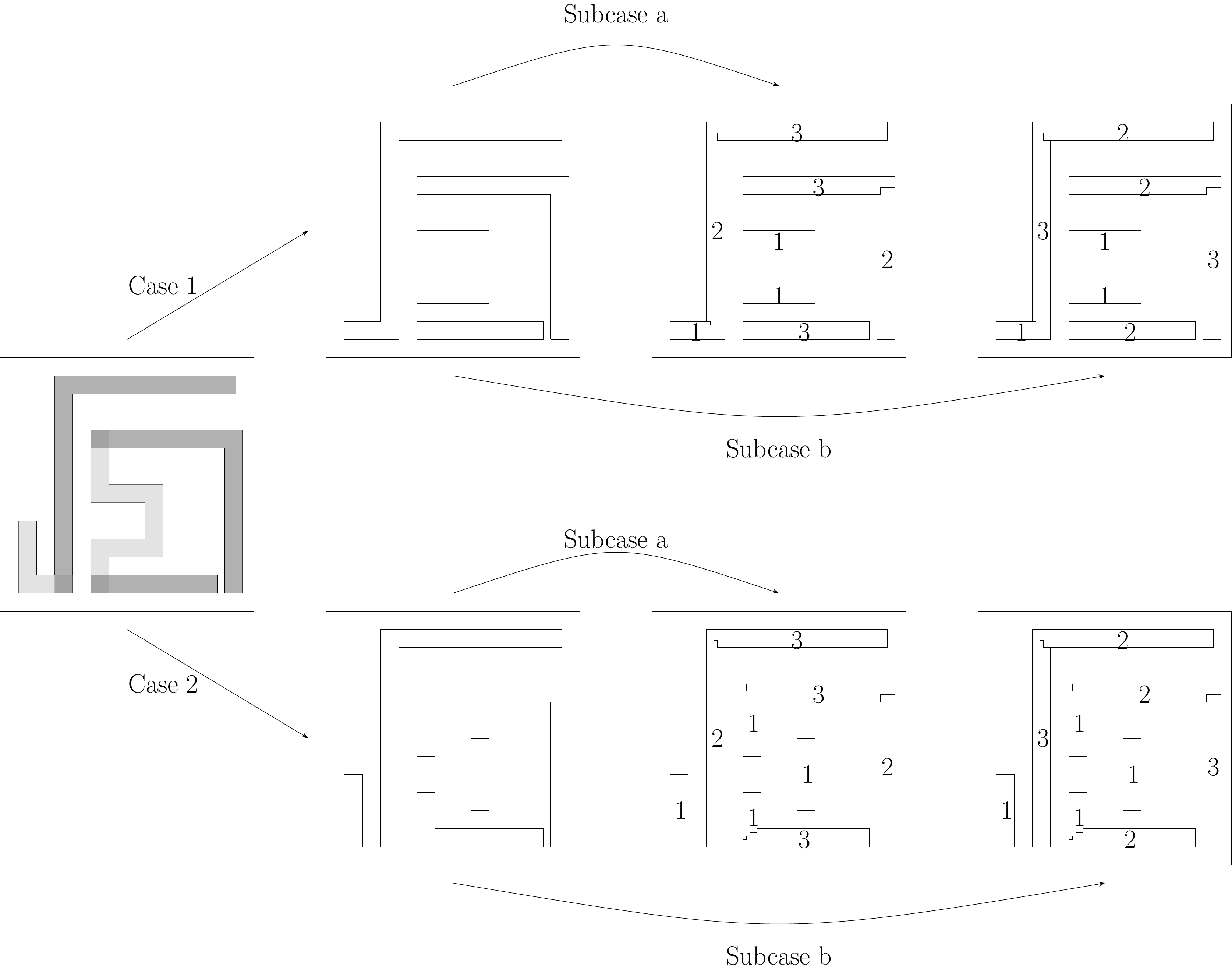}
    
    ~
	\caption{Figure for Case 1 and 2. The knapsack on the left contains two corridors, where short subcorridors are marked light grey and long subcorridors are marked dark grey. In case 1, we delete vertical short subcorridors and then consider two processing orders in subcases a and b. In case 2, we delete horizontal short subcorridors and again consider two processing orders in subcases a and b.}
      \label{fig:case1a}
\end{figure*}

\paragraph{\san{Cases 1a, 1b, 2a, 2b:} Short horizontal/short vertical subcorridors.}

We delete \san{either} all vertical short (case 1) or all horizontal short subcorridors (case 2).
We first process all short subcorridors, then \san{either all vertical (subcases a) or
horizontal long ones (subcases b), and finally the remaining (horizontal or vertical, resp.)} long
ones. We can start by processing all short corridors. Indeed, any
such corridor cannot be the center of a $Z$-bend by Lemma \ref{lem:spiral}.\ref{claim:Zbend}
since its two sides would be long, hence it must be boundary or the
center of a $U$-bend. After processing short subcorridors, by the
same argument the residual (long) subcorridors are the boundary or the
center of a $U$-bend. So we can process the long subcorridors in
any order. \san{This gives in total four cases.}
\ari{See Fig.~\ref{fig:case1a} for deletion/processing of subcorridors for these cases.}

\paragraph{\san{Cases 3a, 3b:} Even/odd short subcorridors.}

We delete the odd (or even) short subcorridors and then process
even (resp., odd) short subcorridors last. We exploit the fact that
each residual corridor contains at most one short subcorridor. Then,
if there is another (long) subcorridor, there is also one which is
boundary (trivially for an open corridor) or the center of a $U$-bend
(by Lemma \ref{lem:spiral}, \ari{Property} \ref{claim:Ubends}). \arir{Addressing one reviewers comment} Hence we can always
process some long subcorridor leaving the unique short subcorridor
as last one. \san{This gives two cases.}


\paragraph{\san{Case 4:} Fat only.}

Do not delete any short subcorridor. Process subcorridors in \fab{any
feasible} order.

In each of the cases, we apply the procedure described in Section~\ref{sec:structural:containers}
to partition each box into $O_{\eps}(1)$ containers. We next
label items as follows. Consider the classification of items into
\san{$\optfa^{cont}$}, $\optth$, and 
\andy{$\optki$} in each one of the \san{$7$} cases above. 
Then: \andyr{I understood that $\optki'$ equals the items in $\optki$
plus the items killed in Lemma \ref{lem:containerPack}. Since in
that lemma now nobody is killed anymore, I replaced $\optki'$ by
$\optki$.} 

\itemsep0pt 

\begin{itemize}
\item \san{$\T$} is the set of items which are in $\optth$ \emph{in at least
one case}; 
\item \san{$\OK$} is the set of items which are in \andy{$\optki$} \emph{in
at least one case}; 
\item \san{$\F$} is the set of items which are in \san{$\optfa^{cont}$} \emph{in all the
cases}. 
\end{itemize}\sanr{As far as I see we never use the definition of $\OK$? Or is this just stated for completeness?}

\begin{remark}\label{rem:processedLast} Consider the subcorridor
of a given corridor that is processed last \andy{in one of the above cases}. None of its items are
assigned to $\optth$ 
\andy{in that case}
and thus essentially all its items are packed
in one of the constructed containers. 
\san{In particular, for an item in set $\T$, in some of the above cases it might be in such a subcorridor and thus marked fat and packed into a container.}
\end{remark}\sanr{Moved this remark after the definition of $T,K,F$ in order to make its implications more explicit.}

\andyr{commented out a sentence here}

\begin{lem}\label{lem:FTprofit} One has $p(\san{\F}\cup \san{\T})+p(\OK)+p(\optla)+\san{p(\optco^{cross})}\geq(1-O(\eps))p(\opt)$.
\end{lem} \begin{proof} Let us initialize $\san{\F}=\san{\optfa^{cont}}$, $\san{\T}=\optth$,
and $\san{\OK}=\andy{\optki}$ by considering one of the above cases. Next
we consider the \ari{aforementioned} cases, hence moving some items in $\san{\F}$ to
either $\san{\T}$ or $\san{\OK}$. Note that initially $p(\san{\F}\cup \san{\T})+p(\optki)+p(\optla)+\san{p(\optco^{cross})}\geq(1-O(\eps))p(\opt)$
by Lemma~\ref{lem:corridorPack-weighted} and hence we keep this
property.\end{proof}

Let $\san{I_{\mathrm{lc}}}$ and $\san{I_{\mathrm{sc}}}$ denote the items in long
and short corridors, respectively. We also let $\san{\LF}=\san{I_{\mathrm lc}}\cap \san{\F}$,
and define analogously $\san{\SF}$, $\san{\LT}$, and $\san{\LF}$. The next \andy{three}
lemmas provide a lower bound on the case of a degenerate L.


\begin{lem}\label{lem:onlyFat} $p(\optrc)\geq p(\san{\LF})+p(\san{\SF}).$ \end{lem}
\begin{proof} \andy{Follows immediately} since we pack a superset
of $\san{\F}$ in \san{case 4}. \end{proof} 
\begin{lem}\label{lem:noShort}
$p(\optrc)\geq p(\san{\LF})+p(\san{\LT})/2+p(\san{\SF})/2.$ \end{lem}
\begin{proof}
Consider the sum of the profit of the packed items corresponding to the in total four subcases of \san{cases 1 and 2}. Each $i\in \san{\LF}$ appears $4$ times in the sum \san{(as items in $\F$ are fat in all cases and all long subcorridors get processed)}, and
each $i\in \san{\LT}$ at least twice by Remark \ref{rem:processedLast}: \san{If a long subcorridor $\mathfrak{L}$ neighbors a short subcorridor, the short subcorridor is either deleted or processed first. Further, all neighboring long subcorridors are processed first in case 1a and 2a (if $\mathfrak{L}$ is horizontal, then its neighbors are vertical) or 1b and 2b (if $\mathfrak{L}$ is vertical and its neighbors are horizontal). Thus, $\mathfrak{L}$ is the last processed subcorridor in at least two cases. Additionally, each item $i\in \SF$ also appears twice in the sum, as it gets deleted either in case 1 (if it is vertical) or in case 2 (if it is horizontal) and is fat otherwise.}

The claim follows by an averaging argument. \end{proof} \begin{lem}\label{lem:evenOdd}
$p(\optrc)\geq p(\san{\LF})+p(\san{\SF})/2+p(\san{\ST})/2.$ \end{lem} 
\begin{proof}
Consider the sum of the number of packed items corresponding to \san{cases 3a and 3b.} Each $i\in \san{\LF}$
appears twice in the sum \san{as it is fat and all long subcorridors get processed}. Each $i\in \san{\SF\cup \ST}$ appears at least
once in the sum by Remark \ref{rem:processedLast}\san{: An item $i\in \SF$ is deleted in one of the two cases (depending on whether it is in an even or odd subcorridor) and otherwise fat. An item $i \in \ST$ is also deleted in one of the two cases and otherwise its subcorridor is processed last}. The claim follows
by an averaging argument. \end{proof}

There is one last (and slightly more involving) case to be considered,
corresponding to a non-degenerate $L$. 

\begin{lem}\label{lem:ringCase} $p(\optrc)\geq\frac{3}{4}p(\san{\LT})+p(\san{\ST})+\frac{1-O(\eps)}{2}p(\san{\SF}).$
\end{lem} \begin{proof} Recall that $\epsl N$ is the maximum width
of a corridor. We consider an execution of the algorithm with boundary
$L$ width $N'=\epsr N$, and threshold length $\ell=(\frac{1}{2}+2\epsl)N$.
We remark that this length guarantees that items in $\san{\ilong}$ are not
contained in short subcorridors.

By Lemma \ref{lem:LoftheRing}, we can pack a subset of $\san{\T\cap \ilong}$
of profit at least $\frac{3}{4}p(\san{\T\cap \ilong})$ in a boundary $L$
of width $\epsr N$. 
By Lemma \ref{lem:boxProperties} the remaining items in $\san{\T}$ can
be packed in two containers of size $\san{\ell}\times\epsr N$ and $\epsr N\times \san{\ell}$
that we place on the two sides of the knapsack not occupied by the
boundary $L$.

In the free area we can identify a square region $K''$ with side
length $(1-\fab{\eps})N$. We next show that there exists a feasible
solution \san{$\SF'\subseteq \SF$} with $p(\san{\SF'})\geq(1-O(\eps))p(\san{\SF})/2$ that
can be packed in a square of side length $(1-\fab{3}\eps)N$. We can
then apply the Resource Augmentation Lemma \ref{lem:augmentPack}
to pack \san{$\SF''\subseteq \SF'$} of cardinality $p(\san{\SF''})\geq(1-\sal{O(\eps)})p(\san{\SF'})$
inside \fab{a central square region} of side length \fab{$(1-3\eps)(1+\epsau)N\leq(1-2\eps)N$}
using containers according to Lemma \ref{lem:containerPack}.\fabr{Small
updates to have enough free area also in this case. Hopefully this
does't confuse the reader (since we save more space that what seems
to be needed)}


Consider the packing of $\san{\SF}$ as in the optimum solution. Choose a
random vertical (resp., horizontal) strip in the knapsack of width
(resp., height) $\fab{3}\eps N$. Delete from $\san{\SF}$ all the items
intersecting the vertical and horizontal strips: clearly the remaining
items $\san{\SF}'$ can be packed into a square of side length $(1-\fab{3}\eps)N$.
Consider any $i\in \san{\SF}$, and assume $i$ is horizontal (the vertical
case being symmetric). Recall that it has height at most $\epss N\leq\eps N$
and width at most $\ell\leq1/2+2\eps$. Therefore $i$ intersects
the horizontal strip with probability at most $5\eps$ and the vertical
strip with probability at most $1/2+\fab{8}\eps$. Thus by the union
bound $i\in \san{\SF'}$ with probability at least $1/2-\fab{13}\eps$. The
claim follows by linearity of expectation. 
\end{proof}


Combining the above Lemmas \ref{lem:FTprofit},\fabr{Let's remind
to cite this lemma in the weighted case as well} \andyr{I referred
to it in the weighted section}\ref{lem:onlyFat}, \ref{lem:noShort},
\ref{lem:evenOdd}, and \ref{lem:ringCase} we achieve the desired
approximation factor, assuming that the (dropped) $O_{\eps}(1)$
items in $\optki\cup\optla\cup\san{\optco^{cross}}$ have zero profit. The worst
case \sal{is} obtained, up to $1-O(\eps)$ factors, for \san{$p(\LT)=p(\SF)=p(\ST)$}
and $p(\san{\LF})=5p(\san{\LT})/4$. This gives $p(\san{\LT})=4/17\cdot p(\san{\T\cup \F})$ and
a total profit of $9/17\cdot p(\san{\T\cup \F})$. 

\subsection{Adding small items}

Note that up to now we ignored the small items $\optsm$. In this
section, we explain how to pack a large fraction of these items. 

We described above how to pack a large enough fraction of $\optsk$
into containers.\fabr{Moved here the container area discussion:
cjeck} We next refine the mentioned analysis to bound the total area
of such containers. It turns out that the residual area is sufficient
to pack almost all the items of $\optsm$ into a constant number of
area containers (not overlapping with the previous containers) for
$\epss$ small enough.


To this aim we use a refined version of the Resource Augmentation
Lemma \ref{lem:augmentPack}, \fab{i.e.,} Lemma \ref{lem:structural_lemma_augm}.\fabr{We
need to add this} Essentially, besides the other properties, we can
guarantee that the total area of the containers is at most $a(I')+\epsau\,a\cdot b$,
where $a\times b$ and $I'$ are the size of the box and the initial
set of items in the claim of Lemma \ref{lem:augmentPack}, respectively.
\begin{lem}\label{lem:areaContainer} In the packings due to Lemmas~\ref{lem:onlyFat},
\ref{lem:noShort}, \ref{lem:evenOdd}, and \ref{lem:ringCase} the
total area occupied by containers is at most $\min\{(1-2\eps)N^{2},a(\optco)+\epsau N^{2}\}$.\fabr{Updated
to consider the non-degenerate L case as well (where containers are
created without passing through boxes)} \end{lem} \begin{proof}
Consider the first upper bound on the area. \fab{We have to distinguish
between the containers considered in Lemma \ref{lem:ringCase} and
the remaining cases. In the first case, there is a region not occupied
by the boundary $L$ nor by the containers of area at least $4\eps N^{2}-4\eps^{2}N^{2}-4\epsr N^{2}\geq2\eps N^{2}$
for $\epsr$ small enough, e.g., $\epsr\le\epsilon^{2}$ suffices.
The claim follows.} For the remaining cases, recall that in each
horizontal box of size $a\times b$ we remove a horizontal strip of
height $3\eps b$, and then use the Resource Augmentation Packing
Lemma to pack the residual items in a box of size $a\times b(1-3\eps)(1+\epsau)\leq a\times b(1-2\eps)$
for $\epsau\leq\eps$. Thus the total area of the containers is at
most a fraction $1-2\eps$ of the area of the original box. A symmetric
argument applies to vertical boxes. Thus the total area of the containers
is at most a fraction $1-2\eps$ of the total area of the boxes, which
in turn is at most $N^{2}$. This gives the first upper bound in the
claim.

For the second upper bound, we just apply the area bound in Lemma
\ref{lem:structural_lemma_augm} to get that the total area of the
containers is at most $a(\optco)$ plus $\epsau\,a\cdot b$ for each
box of size $a\times b$. Summing the latter terms over the boxes
one obtains at most $\epsau N^{2}$. \end{proof}

We are now ready to state a lemma that provides the desired packing
of small items. By slightly adapting the analysis we can guarantee
that the boundary $L$ that we use to prove the claimed approximation
ratio has width at most $\eps^{2}N$. 
\begin{lem}[Small Items Packing Lemma]\label{lem:smallPack} Suppose
we are given a packing of non small items of the above type into $\andy{k}$
containers of total area $A$ and, possibly, a boundary $L$ of width
at most $\eps^{2}N$. 
Then for $\epss$ small enough it is possible to define $O_{\epss}(1)$
area containers of size $\frac{\epss}{\eps}N\times\frac{\epss}{\eps}N$
\andy{neither} overlapping with the containers nor with the boundary
$L$ (if any) such that it is possible to pack $\optsm'\subseteq\optsm$
of profit $p(\optsm')\geq(1-O(\eps))p(\optsm)$ inside these new area
containers. \end{lem} \begin{proof} Let us build a grid of width
$\eps'N=\frac{\epss}{\eps}\cdot N$ inside the knapsack. We delete
any cell of the grid that overlaps with a container or with the boundary
$L$, and call the remaining cells \andy{\emph{free}}. The new
area containers are the free cells.

The total area of the deleted grid cells is, for $\epss$ and $\epsau$
small enough, at most 
\[
(\eps^{2}N^{2}+A)+(\andy{2}N+4N\andy{k})\frac{1}{\eps'N}\cdot(\eps'N)^{2}
\leq A+2\eps^{2}N^{2}\overset{Lem.\ref{lem:areaContainer}}{\leq}\min\{(1-\eps)N^{2},a(\optco)+3\eps^{2}N^{2}\}
\]
\andyr{Shouldn't we refer to Lemma \ref{lem:areaContainer} in the
inequality?\ari{done.}} For the sake of simplicity, suppose that
any empty space in the optimal packing of $\optco\cup\optsm$ is filled
in with dummy small items of profit $0$, so that $a(\optco\cup\optsm)=N^{2}$.
\san{We observe that the area of the free cells is at least $(1-O(\eps))a(\optsm)$: 
Either, $a(\optsm)\geq\eps N^{2}$ and then the area of the free cells is at least 
$a(\optsm)-3\eps^{2}N^{2}\geq(1-3\eps)a(\optsm)$;
otherwise, we have that the area of the free cells is at least $\eps N^2 > a(\optsm)$.}
Therefore we can select a subset of small items $\optsm'\subseteq\optsm$,
with $p(\optsm')\geq(1-O(\eps))p(\optsm)$ and area $a(\optsm)\leq(1-O(\eps))a(\optsm)$
that can be fully packed into free cells using classical Next Fit
Decreasing Height algorithm (NFDH) according to Lemma \ref{lem:nfdhPack}
described later. \andy{The key argument for this is that each free
cell is by a factor $1/\eps$ larger in each dimension than each small
item.}

\end{proof}

Thus, we have proven now that if the items in $\optki\cup\optla\cup\san{\optco^{cross}}$
had zero profit, then there is an L\&C-packing for the skewed and
small items with a profit of at least $9/17\cdot p(\optco)+(1-O(\epsilon))(\optsm)\ge(9/17-O(\epsilon))p(\opt)$.

\subsection{Shifting argumentation}

We remove now the assumption that we can drop $O_{\eps}(1)$ items
from $\opt$. We will add a couple of shifting steps to the argumentation
above to prove Lemma~\ref{lem:apxNoRotation} without that assumption.

It is no longer true that we can neglect the large rectangles $\optla$
since they might contribute a large amount towards the objective,
even though their total number is guaranteed to be small. Also, in
the process of constructing the boxes, we killed up to $O_{\eps}(1)$
rectangles (the rectangles in $\optki$). Similarly, we can no longer
drop the constantly many items in \san{$\optco^{cross}$}. Therefore, we apply
some careful shifting arguments in order to ensure that we can still
\andy{use} a similar construction as above, while losing only a
factor $1+O(\eps)$ due to some items that we will discard.

\san{The general idea is as follows: For $t=0,\ldots, k$ (we will later argue that $k<1/\eps$), 
we define disjoint sets $K(t)$ recursively, each containing at most $O_\eps(1)$ items. 
Each set $\K(t)=\bigcup_{j=0}^t K(j)$ is used to define a grid $G(t)$ in the knapsack. 
Based on an item classification that depends on this grid, we identify a set of skewed items 
and create a corridor partition w.r.t. these skewed items as described in Lemma \ref{lem:corridorPack-weighted}. 
We then create a partition of the knapsack into corridors and a constant (depending on $\eps$) 
number of containers (see Section \ref{sec:weighted:shifting:grid}). 
Next, we decompose the corridors into boxes (Section \ref{sec:weighted:shifting:boxes}) 
and these boxes into containers (section \ref{sec:weighted:shifting:containers}) 
similarly as we did in Sections \ref{sec:structural:boxes} and \ref{sec:structural:containers} 
(but with some notable changes as we did not delete small items from the knapsack and thus need 
to handle those as well). 
In the last step, we add small items to the packing (Section \ref{sec:weighted:shifting:small}). 
During this whole process, we define the set $K(i+1)$ of items that are somehow ``in our way'' 
during the decomposition process (e.g., items that are not fully contained in some corridor of 
the corridor partition), but which we cannot delete directly as they might have large profit. 
\andy{These items are similar to the killed items in the previous argumentation.}
However, using a shifting argument we can simply show that after at most $k<1/\eps$ steps of this 
process, we encounter a set $K(k)$ of low total profit, that we can remove, at which point we 
can apply almost the same argumentation as in Lemmas~\ref{lem:onlyFat},
\ref{lem:noShort}, and \ref{lem:evenOdd} to obtain lower bounds on the profit of an optimal 
L\&C packing (Section \ref{sec:weighted:shifting:packing}).}

\san{We initiate this iterative process as follows: Denote by $K(0)$ a set containing all items that are killed in at
least one of the cases arising in Section~\ref{sec:structural:lemma}
(the set $\OK$ in that Section) and additionally the large items $\optla$
and the $O_{\eps}(1)$ items in \san{$\optco^{cross}$} (in fact $\optla\subseteq\san{\optco^{cross}}$).
Note that $|K(0)|\le O_{\eps}(1)$. If $p(K(0))\le\eps\cdot p(OPT)$
then we can simply remove these rectangles (losing only a factor of $1+\eps$)
and then apply the remaining argumentation exactly as above and we are done.
Therefore, from now on suppose that $p(K(0))>\eps\cdot p(OPT)$.}

\subsubsection{Definition of grid and corridor partition}\label{sec:weighted:shifting:grid}

\san{Assume we are in round $t$ of this process, i.e., we defined $K(t)$ in the previous step (unless $t=0$, then $K(t)$ is defined as specified above) and assume that $p(K(t))>\eps \opt$ (otherwise, see Section \ref{sec:weighted:shifting:packing}).}
We \san{are now going to define the non-uniform grid $G(t)$ and the induced} partition \san{of} the knapsack
into a collection of cells $\C_{t}$. The $x$-coordinates ($y$-coordinates) of the grid cells are the
$x$-coordinates ($y$-coordinates, respectively) of the items in
\san{$\K(t)$}. This yields a partition of the knapsack into $O_{\eps}(1)$
rectangular cells, such that each item of \san{$\K(t)$} covers one or multiple
cells. 
Note that an item might intersect many cells.

Similarly as above, we define constants $1\ge\epsl\ge\epss\ge\Omega_{\eps}(1)$
and apply a shifting step such that we can assume that for each item
$i\in OPT$ touching some cell $\cell$ we have that $w(i\cap\cell)\in(0,\epss w(\cell)]\cup(\epsl w(\cell),w(\cell)]$
and $h(i\cap\cell)\in(0,\epss h(\cell)]\cup(\epsl h(\cell),h(\cell)]$,
where $h(\cell)$ and $w(\cell)$ denote the height and the width
of the cell $\cell$ and $w(i\cap\cell)$ and $h(i\cap\cell)$ denote
the height and the width of the intersection of the rectangle $i$
with the cell $\cell$, respectively. For a cell $\cell$ denote by
$\opt(\cell)$ the set of rectangles that intersect $\cell$ in $\opt$.
We obtain a partition of $\opt(\cell)$ into $\optsm(\cell)$, $\optla(\cell)$,
$\optho(\cell)$, and $\optve(\cell)$: 
\begin{itemize}
\item $\optsm(\cell)$ contains all items $i\in\opt(\cell)$ with $h(i\cap\cell)\le\epss h(\cell)$
and $w(i\cap\cell)\le\epss w(\cell)$, 
\item $\optla(\cell)$ contains all items $i\in\opt(\cell)$ with $h(i\cap\cell)>\epsl h(\cell)$
and $w(i\cap\cell)>\epsl w(\cell)$, 
\item $\optho(\cell)$ contains all items $i\in\opt(\cell)$ with $h(i\cap\cell)\le\epss h(\cell)$
and $w(i\cap\cell)>\epsl w(\cell)$, and 
\item $\optve(\cell)$ contains all items $i\in\opt(\cell)$ with $h(i\cap\cell)>\epsl h(\cell)$
and $w(i\cap\cell)\le\epss w(\cell)$. 
\end{itemize}
We call a rectangle $i$ \emph{intermediate }if there is a cell $\cell$
such that $w(i\cap C)\in(\epss w(C),\epsl w(C)]$ or $h(i\cap C)\in(\epss w(C),\epsl w(C)]$.
Note that a rectangle $i$ is intermediate if and only if the last
condition is satisfied for one of the at most four cells that contain
a corner of $i$. 

\begin{lem} For any constant $\eps>0$ and \fab{positive} increasing
function $f(\cdot)$, \fab{$f(x)>x$,} there exist constant values
$\epsl,\epss$, with $\eps\geq\epsl\geq\fab{f(\epss)}\ge\Omega_{\eps}(1)$
and $\epss\in\Omega_{\eps}(1)$ such that the total profit of intermediate
rectangles is bounded by $\eps p(OPT)$. \end{lem} 

For each cell $C$ that is not entirely covered by some item in \san{$K(t)$}
we \san{add all rectangles in $\optla(C)$ that are not contained in \san{$\mathcal{K}(t)$} to $K(t+1)$}. 
\san{Note that here, in contrast to before, we do \emph{not} remove small items from the packing but keep them.}

Based on the items \san{$\optsk(\C_{t}):=\cup_{\cell\in\C_{t}}\optho(\cell)\cup\optve(\cell)$}\sanr{We do not take care of the items in $\K(t)$, i.e. the ``large'' items that we need to keep? Then it can happen that the corridors cross some of those?}
we create a corridor decomposition and consequently a box decomposition
of the knapsack. To make this decomposition clearer, we assume that
we first stretch the non-uniform grid into a uniform $[0,1]\times[0,1]$
grid. After this operation, for each cell $C$ and for each element
of $\optho(C)\cup\optve(C)$ we know that its height or width is at
least $\epsl\cdot\frac{1}{1+2|\san{\K(t)}|}$. We apply Lemma~\ref{lem:corridorPack-weighted}
on the set $\optsk(\san{\C_{t}})$ which yields a decomposition of the $[0,1]\times[0,1]$
square into at most $O_{\eps,\epsl,\san{\K(t)}}(1)=O_{\eps,\epsl}(1)$ corridors.
The decomposition for the stretched $[0,1]\times[0,1]$ square corresponds
to the decomposition for the original knapsack, with the same items
being intersected. Since we can assume that all items of $OPT$ are
placed within the knapsack so that they have integer coordinates,
we can assume that the corridors of the decomposition also have integer
coordinates. We can do that, because shifting the edges of the decomposition
to the closest integral coordinate will not make the decomposition
worse, i.e., no new items of $OPT$ will be intersected.

We \san{add} all rectangles in $\andy{\optsk(\san{\C_{t}})}$ that are not contained
in a corridor (at most $O_{\eps}(1)$ many) and that are not contained
in \san{$\K(t)$}\san{ to $K(t+1)$}. 
The corridor partition
has the following useful property: we started with a fixed (optimal)
solution $\opt$ for the overall problem with a \emph{fixed placement
}of the items in this solution. Then we considered the items in $\optsk(\san{\C_{t}})$
and obtained the sets $\optco\subseteq\optsk(\san{\C_{t}})$ and $\san{\optco^{cross}}\subseteq\optco$.
With the mentioned fixed placement, apart from the $O_{\eps}(1)$
items in \san{$\optco^{cross}$}, each item in $\optco$ is contained in one corridor.
In particular, the items in $\optco$ do not overlap the items in
$\san{\K(t)}$. We construct now a partition of the knapsack into $O_{\eps}(1)$
corridors and $O_{\eps}(1)$ containers where each container contains
exactly one item from \san{$\K(t)$}. The main obstacle is that there can
be an item $i\in \san{\K(t)}$ that overlaps a \ari{corridor} (see Figure~\ref{fig:weighted-circumvent}).
We solve this problem by applying the following lemma on each such
corridor. 
\begin{lem} \label{lem:divide-open-corridors} Let $S$
be an open corridor with $b(S)$ bends. Let $I'\subseteq OPT$ be
a collection of items which intersect the boundary of $S$ with $I'\cap\optsk(\san{\C_{t}})=\emptyset$.
Then there is a collection of $|I'|\cdot b(S)$ line segments $\mathcal{L}$
within $S$ which partition $S$ into corridors with at most $b(S)$
bends each such that no item from $I'$ is intersected by $\mathcal{L}$
and there are at most $O_{\eps}(|I'|\cdot b(S))$ items of $\optsk(\san{\C_{t}})$
intersected by line segments in $\L$. \end{lem} \begin{proof} Let
$i\in I'$ and assume w.l.o.g.~that $i$ lies within a horizontal
subcorridor $S_{i}$ of the corridor $S$. 
If the top or bottom edge $e$ of $i$ lies within $S_{i}$, we define
a horizontal line segment $\ell$ which contains the edge $e$ and
which is maximally long so that it does not intersect the interior
of any item in $I'$, and such that it does not cross the boundary
curve between $S_{i}$ and an adjacent subcorridor, or an edge of
the boundary of $S$ (we can assume w.l.o.g. that $e$ does not intersect
the boundary curve between $S_{i}$ and some adjacent subcorridor).
We say that $\ell$ \emph{crosses} a boundary curve $c$ (or an edge
$e$ of the boundary of $S$) if $c\setminus\ell$ (or $e\setminus\ell$)
has two connected components.

\begin{figure}
\begin{centering}
\includegraphics[scale=0.3]{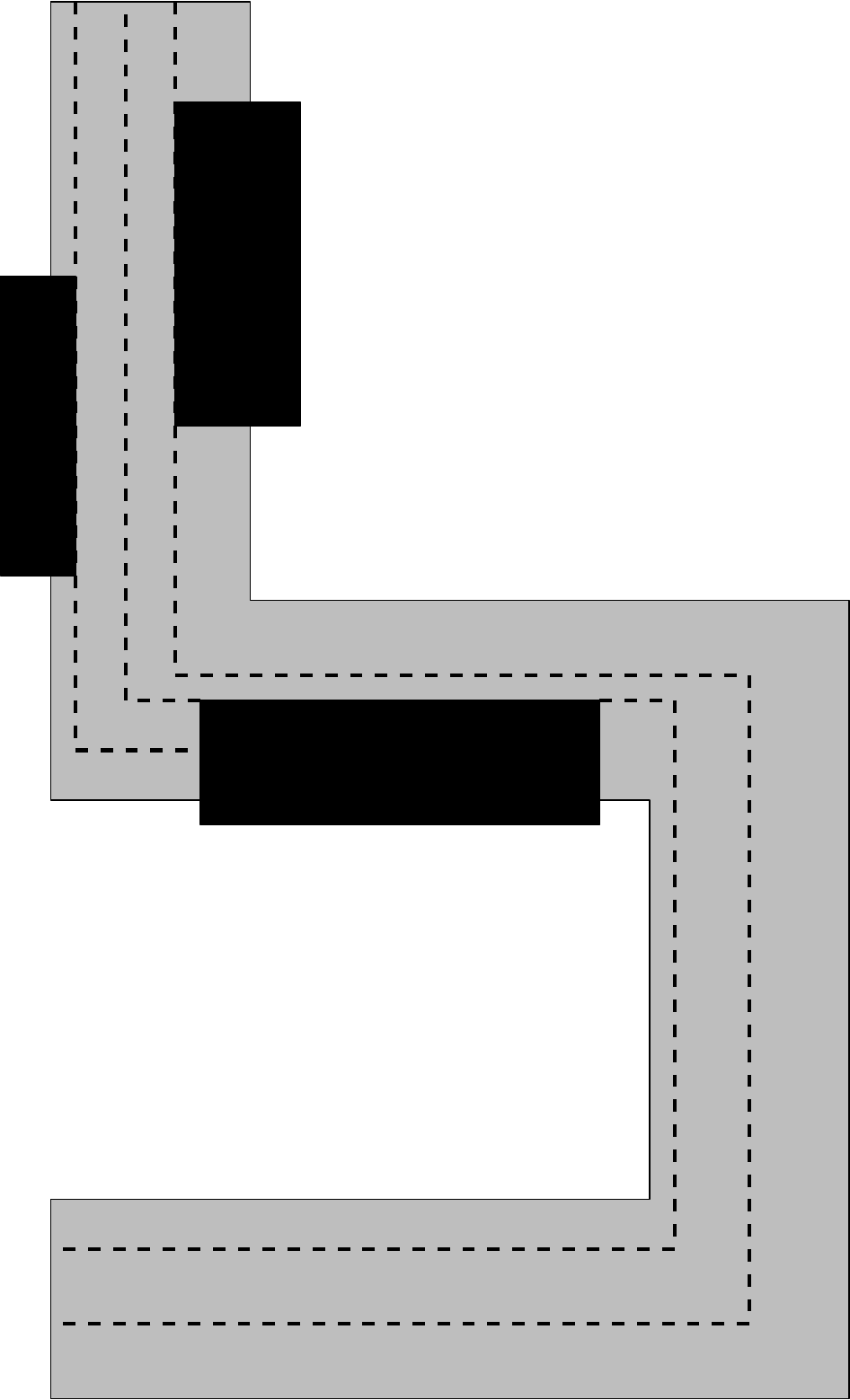} ~~~~~~~~~\includegraphics{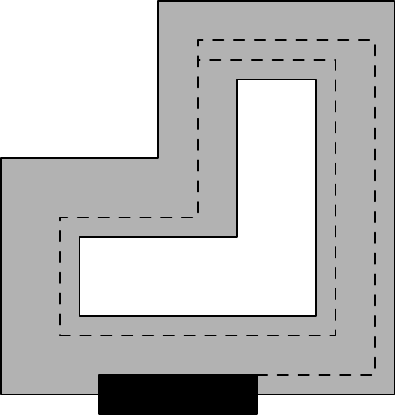}
\par\end{centering}
\caption{\label{fig:weighted-circumvent}Circumventing the items in $I'$,
shown in black. The connected components between the dashed lines
show the resulting new corridors. }
\end{figure}

We now ``extend'' each end-point of $\ell$ which does not lie at
the boundary of some other item of $I'$ or at the boundary of $S$,
we call such an end point a \emph{loose end}. For each loose end $x$
of $\ell$ \san{lying on the boundary} curve $c_{ij}$ partitioning the
subcorridors $S_{i}$ and $S_{j}$, we introduce a new line $\ell'$
perpendicular to $\ell$, starting at $x$ and crossing the subcorridor
$S_{j}$ such that the end point of $\ell'$ is maximally far away
from $x$ subject to the \ari{constraint} that $\ell'$ does not
cross an item in $I'$, another boundary curve inside $S$, or the
boundary of $S$. We continue iteratively. Since the corridor has
$b(S)$ bends, after at most $b(S)$ iterations this operation will
finish. We repeat the above operation for every item $i\in I'$, and
we denote by $\L$ the resulting set of line segments, see Figure~\ref{fig:weighted-circumvent}
for a sketch. Notice that $|\L|=b(S)\cdot|I'|$. By construction,
if an item $i\in\optsk(\san{\C_{t}})$ is intersected by a line in $\L$
then it is intersected parallel to its longer edge. Thus, each line
segment in $\L$ can intersect at most $O_{\eps}(1)$ items of $\optsk(\san{\C_{t}})$.
Thus, in total there are at most $O_{\eps,\epsl}(|I'|\cdot b(S))$
items of $\optsk(\san{\C_{t}})$ intersected by line segments in $\L$.
\end{proof} 
We apply Lemma~\ref{lem:divide-open-corridors} to each
\ari{open} corridor that intersects an item in $\san{\K(t)}$. We \san{add}
all \san{items} of $\Rsk(\san{\C_{t}})$ that are intersected by line segments
in $\L$ \san{to $K(t+1)$}. This adds $O_{\eps}(1)$ items in total to \san{$K(t+1)$} since $|\san{\K(t)}|\in O_{\eps}(1)$
and $b(S)\le1/\eps$ for each corridor $S$. For closed corridors
we prove the following lemma. 
\begin{lem} \label{lem:divide-closed-corridors}
Let $S$ be a closed corridor with $b(S)$ bends. Let $\optsk(S)$
denote the items in $\optsk(\san{\C_{t}})$ that are contained in $S$.
Let $I'\subseteq OPT$ be a collection of items which intersect the
boundary of $S$ with $I'\cap\optsk(\san{\C_{t}})=\emptyset$. Then there
is a collection of $O_{\eps}(|I'|^{2}/\eps)$ line segments $\mathcal{L}$
within $S$ which partition $S$ into a collection of closed corridors
with at most $1/\eps$ bends each and possibly an open corridor with
$b(S)$ bends such that no item from $I'$ is intersected by $\mathcal{L}$
and there is a set of items $\optsk'(S)\subseteq\optsk(S)$ with $|\optsk'(S)|\le O_{\eps}(|I'|^{2})$
such that the items in $\optsk(S)\setminus\optsk'(S)$ intersected
by line segments in $\L$ have a total profit of at most $O(\eps)\cdot p(\optsk(\san{\C_{t}}))$.
\end{lem} \begin{proof} Similarly as for the case of open corridors,
we take each item $i\in I'$ whose edge $e$ is contained in $S$,
and we create a path containing $e$ that partitions $S$. Here the
situation is a bit more complicated, as our newly created paths could
extend over more than $\frac{1}{\eps}$ bends inside $S$. In this
case we will have to do some shortcutting, i.e., some items contained
in $S$ will be crossed parallel to their shorter edge and we cannot
guarantee that their total number will be small. However, we will
ensure that the total weight of such items is small. We proceed as
follows (see Figure~\ref{fig:weighted-circumvent} for a sketch).

Consider any item $i\in I'$ and assume w.l.o.g. that $i$ intersects
a horizontal subcorridor $S_{i}$ of the closed corridor $S$. Let
$e$ be the edge of $i$ within $S_{i}$. For each endpoint of $e$
we create a path $p$ as for the case of closed corridors. If after
at most $b(S)\le1/\eps$ bends the path hits an item of $I'$ (possibly
the same item $i$), the boundary of $S$ or another path created
earlier, we stop the construction of the path. Otherwise, if $p$
is the first path inside of $S$ that did not finish after at most
$b(S)$ bends, we proceed with the construction of the path, only
now at each bend we check the total weight of the items of $\optsk(S)$
that would be crossed parallel to their shorter edge, if, instead
of bending, the path would continue at the bend to hit itself. From
the construction of the boundary curves in the intersection of two
subcorridors we know that for two bends of the constructed path, the
sets of items that would be crossed at these bends of the path are
pairwise disjoint. Thus, after $O(|I'|/\eps)$ bends we encounter
a collection of items $\optsk''(S)\subseteq\optsk(S)$ such that $p(\optsk''(S))\le\frac{\eps}{|I'|}p(\optsk(S))$,
and we end the path $p$ by crossing the items of $\optsk''(S)$.
\andy{This operation creates an open corridor with up to $O(|I'|/\eps)$
bends. We divide it into up to $O(|I'|)$ corridors with up to $1/\eps$
bends each. Via a shifting argument we can argue that this loses at
most a factor of $1+\eps$ in the profit due to these items.} When
we perform this operation for each item $i\in I'$ the total weight
of items intersected parallel to their shorter edge (i.e., due to
the above shortcutting) is bounded by $|I'|\cdot\frac{\eps}{|I'|}p(\optsk(S))=\eps\cdot p(\optsk(S))$.
This way, we introduce at most $O(|I'|^{2}/\eps)$ line segments.
Denote them by $\L$. They intersect at most $O_{\eps}(|I'|^{2})$
items parallel to their respective longer edge, denote them by $\optsk'(S)$.
Thus, the set $\L$ satisfies the claim of the lemma. \end{proof}
Similarly as for Lemma~\ref{lem:divide-open-corridors} we apply
Lemma~\ref{lem:divide-closed-corridors} to each closed corridor.
We \san{add} all items in the respective set $\optsk'(S)$ \san{to the set $\san{K(t+1)}$} which yields
$O_{\eps}(1)$ many items. 
\san{The items in $\optsk(S)\backslash\optsk'(S)$ are removed from the instance, as their total profit is small.}

\subsubsection{Partitioning corridors into boxes}\label{sec:weighted:shifting:boxes}

Then we partition the resulting corridors into boxes according to
the different cases described in Section~\ref{sec:structural:lemma}.
There is one difference to the argumentation above: we define that
the set $\optfa$ contains not only skewed items contained in the
respective subregions of a subcorridor, but \emph{all }items contained
in such a subregion. In particular, this includes items that might
have been considered as small items above. Thus, when we move items
from one subregion to the box associated to the subregion below (see
Remark~\ref{rem:fatPack}) then we move \emph{every }item that is
contained in that subregion. If an item is killed in one of the orderings
of the subcorridors to apply the procedure from Section~\ref{sec:structural:boxes}
then we add it to $\san{K(t+1)}$. Note that $|\san{K(t+1)}|\in O_{\eps,\epsl,\epst}(1)$
and $\san{\K(t)\cap K(t+1)}=\emptyset$.
\san{Also note here that we ignore for the moment small items that cross the boundary curves of the subcorridors; they will be taken care of in Section \ref{sec:weighted:shifting:small}.}

\subsubsection{Partitioning boxes into containers}\label{sec:weighted:shifting:containers}

Then we subdivide the boxes into containers. We apply Lemma~\ref{lem:containerPack}
to each box with a slight modification. Assume that we apply it to
a box of size $a\times b$ containing a set of items $I_{box}$. Like
above we first remove the items in a thin strip of width $3\eps b$
such that via a shifting argument the items (fully!) contained in
this strip have a small profit of $O(\eps)p(I_{box})$. However, in
contrast to the setting above the set $I_{box}$ contains not only
skewed items but also small items. We call an item $i$ \emph{small
}if there is no cell $\cell$ such that $i\in\optla(\cell)\cup\optho(\cell)\cup\optve(\cell)$
and denote by $\optsm(\san{\C_{t}})$ the set of small items. When we choose
the strip to be removed we ensure that the profit of the removed skewed
\emph{and} small items is small. There are $O_{\eps}(1)$ skewed items
that partially (but not completely) overlap the strip whose items
we remove. We \san{add} those $O_{\eps}(1)$ items 
to $\san{K(t+1)}$. \san{Small items that partially overlap the strip are taken care of later in Section \ref{sec:weighted:shifting:small}, we ignore them for the moment.}
Then we apply Lemma~\ref{lem:augmentPack}. In contrast to the setting
above, we do not only apply it to the skewed items but apply it also
to small items that are contained in the box. Denote by $\optsm'(\san{\C_{t}})$
the set of small items that are contained in some box of the box partition.

Thus, we obtain an L\&C packing for the items in $\san{\K(t)}$, for a set
of items $\optsk'(\san{\C_{t}})\subseteq\optsk(\san{\C_{t}})$, and for a set
of items $\optsm''(\san{\C_{t}})\subseteq\optsm'(\san{\C_{t}})$ such that $p(\optsk'(\san{\C_{t}}))+p(\optsm''(\san{\C_{t}}))+p(\san{K(t+1)})\ge(1-O(\eps))p(\optsk(\san{\C_{t}})\cup\optsm'(\san{\C_{t}}))$.

\subsubsection{Handling small items}\label{sec:weighted:shifting:small}

So far we ignored the small items in $\optsm''(\san{\C_{t}}):=\optsm(\san{\C_{t}})\setminus\optsm'(\san{\san{\C_{t}}})$.
This set consists of small items that in the original packing intersect
a line segment of the corridor partition, the boundary of a box, or
a boundary curve within a corridor. We describe now how to add them
into the empty space of the so far computed packing. First, we assign
each item in $\optsm''(\san{\C_{t}})$ to a grid cell. We assign each small
item $i\in\optsm''(\san{\C_{t}})$ to the cell $C$ such that in the original
packing $i$ intersects with $C$ and the area of $i\cap C$ is not
smaller than $i\cap C'$ for any cell $C'$ ($i\cap C'$ denotes the
part of $i$ intersecting $C'$ in the original packing for any grid
cell $C'$).

\fabr{I think we can save an $\eps$ fraction of the area of rescaled
grid, excluding area of \char`\"{}guessed\char`\"{} items that we
cannot touch. The current procedure would for example kill \char`\"{}automatically\char`\"{}
the relatively large items in a given cell. They are a constant number
and we counted them already in some sense. The small packing lemma
should work also with the residual area (not sure about the $L$ case)}
Consider a grid cell $C$ and let $\optsm''(C)$ denote the small
items in $\optsm''(\san{\C_{t}})$ assigned to $C$. Intuitively, we want
to pack them into the empty space in the cell $C$ that is not used
by any of the containers, similarly as above. We first prove an analog
of Lemma~\ref{lem:areaContainer} of the setting above.

\begin{lem}\label{lem:areaContainerWeighted} Let $\cell$ be a cell.
The total area of $C$ occupied by containers is at most $(1-2\eps)a(\cell)$.
\end{lem}

\begin{proof}

In our construction of the boxes we moved some of the items (within
a corridor). In particular, it can happen that we moved some items
into $C$ that were originally in some other grid cell $C'$. This
reduces the empty space in $C$ for the items in $\optsm''(C)$. Assume
that there is a horizontal subcorridor $H$ intersecting $C$ such
that some items or parts of items within $H$ were moved into $C$
\san{that} were not in $C$ before. Then such items were moved vertically
and the corridor containing $H$ must intersect the upper or lower
boundary of $C$. The part of this subcorridor \ari{lying} within
$C$ has a height of at most $\epsl\cdot h(C)$. Thus, the total area
of $C$ lost in this way is bounded by $O(\epsl a(C))$ which includes
analogous vertical subcorridors.

Like in Lemma~\ref{lem:areaContainer} we argue that in each horizontal
box of size $a\times b$ we remove a horizontal strip of height $3\eps b$
and then the created containers lie in a box of height $(1-3\eps)(1+\epsau)b$.
In particular, if the box does not intersect the top or bottom edge
of $C$ then within $C$ its containers use only a box of dimension
$a'\times(1-3\eps)(1+\epsau)b$ where $a'$ denotes the width of the
box within $C$, i.e., the width of the intersection of the box with
$C$. If the box intersects the top or bottom edge of $C$ then we
cannot guarantee that the free space lies within $C$. However, the
total area of such boxes is bounded by $O(\epsl a(C))$. We can apply
a symmetric argument to vertical boxes. Then, the total area of $C$
used by containers is at most $(1-3\eps)(1+\epsau)a(C)+O(\epsl a(C))\le(1-2\eps)a(C)$.
This gives the claim of the lemma. \end{proof}

Next, we argue that the items in $\optsm''(C)$ have very small total
area. Recall that they are the items intersecting $C$ that are not
contained in a box. The total number of boxes and boundary curves
intersecting $C$ is $O_{\eps,\epsl}(1)$ and in particular, this
quantity does not depend on $\epss$. Hence, by choosing $\epss$
sufficiently small, we can ensure that $a(\optsm''(C))\le\eps a(C)$.
Then, similarly as in Lemma~\ref{lem:smallPack} we can argue that
if $\epss$ is small enough then we can pack the items in $\optsm''(C)$
using NFDH into the empty space within $C$.

\subsubsection{L\&C packings}\label{sec:weighted:shifting:packing}

We iterate the above construction, obtaining
pairwise disjoint sets $K(1),K(2),...$ until we find a set $K(\san{k})$\sanr{Changed $\ell$ to $k$ to avoid confusion with the other $\ell$.}
such that $p(K(\san{k}))\le\eps\cdot OPT$. Since the sets $K(0),K(1),...$
are pairwise disjoint there must be such a value \san{$k$} with $\san{k}\le1/\eps$.
Thus, $|\san{\K(k-1)}|\le O_{\eps}(1)$. We build the grid given by the
$x$- and $y$-coordinates of \san{$\K(k-1)$},
\san{giving a set of cells} $\C_{\san{k}}$. As described above we define the corridor
partition, the partition of the corridors into boxes (with the different
orders to process the subcorridors as described in Section~\ref{sec:structural:boxes})
and finally into containers. \andy{Denote by $\optsm(\C_{\san{k}})$
the resulting set of small items.}


We consider the candidate packings based on the grid $\C_{\san{k}}$.
For each of the six candidate packings with a degenerate $L$ we can
pack almost all small items of the original packing. We define \san{$I_{\mathrm{lc}}$
and $I_{\mathrm{sc}}$} the sets of items in long and short subcorridors in the initial
corridor partition, respectively. Exactly as in the cardinality case,
a subcorridor is long if it is longer than $N/2$ and short otherwise.
As before we divide the items into fat and thin items and define the
sets $\san{\SF}$, $\san{\LT}$, and $\san{\ST}$ accordingly. Moreover, we define the
set $\san{\LF}$ to contain all items in \san{$I_{\mathrm{lc}}$} that are fat in all candidate
packings \emph{plus} the items in \san{$\K(k-1)$}.

Thus, we obtain the respective claims of Lemmas~\ref{lem:onlyFat},
\ref{lem:noShort}, and \ref{lem:evenOdd} in the weighted setting.
For the following lemma let $\optsm:=\optsm(\C_{\san{k}})$.\fabr{Can
we avoid to restate the case analysis and reduce to that? Maybe considering
guessed items as fat and long?}\andyr{Considering guessed items
now as fat and long. We could remove Lemma~\ref{lem:weighted-apx}.
However, I think that it helps to keep track of what happened so far.}
\begin{lem} \label{lem:weighted-apx}Let $\optrc$ the most profitable
solution that is packed by an L\&C packing. 
\begin{enumerate}\renewcommand{\theenumi}{\alph{enumi}}\renewcommand\labelenumi{(\theenumi)}
\item $p(\optrc)\ge(1-O(\eps))(p(\san{\LF})+p(\san{\SF})+p(\optsm))$ 
\item $p(\optrc)\geq(1-O(\eps))(p(\san{\LF})+p(\san{\SF})/2+p(\san{\LT})/2+p(\optsm))$ 
\item $p(\optrc)\geq(1-O(\eps))(p(\san{\LF})+p(\san{\SF})/2+p(\san{\ST})/2+p(\optsm)).$ 
\end{enumerate}
\end{lem}
For the candidate packing for the non-degenerate-$L$ case
(\san{Lemma \ref{lem:ringCase} in} Section \ref{sec:structural:lemma}) we first add the small items as described
above. Then we remove the items in $\san{\K(k-1)}$.
Then, like above, with a random shift we delete items touching a horizontal
and a vertical strip of width $3\eps N$. Like before, each item $i$
is still contained in the resulting solution with probability $1/2-15\eps$
(note that we cannot make such a claim for the items in $\san{\K(k-1)}$).
For each small item we can even argue that it still contained in the
resulting solution with probability $1-O(\eps)$ (since it is that
small in both dimensions).
We proceed with the construction of the
boundary $L$ and the assignment of the items into it like in the
unweighted case.
\begin{lem} \label{lem:weighted-apx2}For the solution
$\optrc$ we have that $p(\optrc)\ge(1-O(\eps))(\frac{3}{4}p(\san{\LT})+p(\san{\ST})+\frac{1-O(\eps)}{2}p(\san{\SF})+p(\optsm))$.
\end{lem} When we combine Lemmas~\ref{lem:weighted-apx} and \ref{lem:weighted-apx2}
we conclude that $p(\optrc)\ge(17/9+O(\eps))p(OPT)$. \andy{Similarly
as before, the worst case is obtained, up to $1-O(\eps)$ factors,
when we have that $p(\san{\LT})=p(\san{\SF})=p(\san{\ST})$, $p(\san{\LF})=5p(\san{\LT})/4$, and $p(\optsm)=0$.}\andyr{@Fab:
did you mean something like this in your comment at the proof of Lemma~\ref{lem:apxNoRotation}
in the cardinality case?} This completes the proof of Lemma~\ref{lem:apxNoRotation}.

\subsection{Main algorithm}

In this Section we present our main algorithm for the weighted case
of \tdk. It is in fact an approximation scheme for L\&C packings.
Its approximation factor therefore follows immediately from Lemma~\ref{lem:apxNoRotation}. 

Given $\epsilon>0$, we first guess the quantities $\epsl,\epss$,
the proof of Lemma~\ref{lem:item-classification} reveals that there
are only $2/\epsilon+1$ values we need to consider. We choose $\epsr:=\epsilon^{2}$
and subsequently define $\epsb$ according to Lemma~\ref{lem:boxProperties}.
Our algorithm combines two basic packing procedures. The first one
is the following standard PTAS to pack items into a constant number
of containers. The same basic approach works also with rotations.
The basic idea is to reduce the problem to an instance of the Maximum
Generalized Assignment Problem~\andy{(GAP)} with one bin per container,
and then use a PTAS for the latter problem plus Next Fit Decreasing
Height to repack items in area containers.

\begin{lem}\label{lem:containersPackPTAS} There is a PTAS for the
problem of computing a maximum profit packing of a subset of items
\andy{of} a given set $I'$ into a given set of containers of constant
cardinality, \fab{both with and without rotations}. \end{lem}

The second packing procedure \san{is the} PTAS for the L-packing problem, see Theorem~\san{\ref{thm:main:Lpacking}}.

To use these packing procedures, we first guess whether the optimal
L\&C-packing due to Lemma~\ref{lem:apxNoRotation} uses a non-degenerate
boundary L. 
If yes, we guess a parameter $\ell$ which denotes the
minimum height of the vertical items in the boundary $L$ and the
minimum width of the horizontal items in the boundary $L$. For $\ell$
we allow all heights and widths of the input items that \san{are} larger than
$N/2$, i.e., at most $2n$ values. Let $\san{\ilong}$ be the items whose
longer side has length at least $\ell$ (hence longer than $N/2$).
We set the width of the boundary $L$ to be $\epsilon^{2}N$ and solve
the resulting instance $(L,\san{\ilong})$ optimally using the PTAS for L-packings
due to Theorem~\san{\ref{thm:main:Lpacking}}. Then we enumerate all the possible subsets \andy{of}
non-overlapping containers in the space not occupied by the boundary
$L$ (or in the full knapsack, in the case of a degenerate $L$),
where the number and sizes of the containers are defined as in Lemma
\ref{lem:containerPack}. In particular, there are at most $O_{\eps}(1)$
containers and there is a set of size $n^{O_{\eps}(1)}$ that
we can compute in polynomial time such that the height and the width
of each container is contained in this set. We compute an approximate
solution for the resulting container packing instance with items $\san{\ishort}=I\setminus \san{\ilong}$
using the PTAS from Lemma \ref{lem:containersPackPTAS}. Finally,
we output \san{the} most profitable solution that we computed.

%% file: 2Dknapsack8f-cardinality-No-Rotation.tex
\section{Cardinality case without rotations}
\label{sec:tdk_car:refined}
In this section, we present a refined algorithm with approximation factor of $\frac{558}{325}+\eps<1.717$ for the cardinality case when rotations are not allowed. 

\begin{thm}\label{thm:cardWorot}
There exists a polynomial-time $\frac{558}{325}+\eps<1.717$-approximation algorithm for cardinality 2DK.
\end{thm}

Along this section, since the profit of each item is equal to $1$, instead of $\profit(I)$ for a set of items $I$ we will just write $|I|$. We will use most of the notation defined in Section \ref{sec:weighted}. Recall that for two given constants $0<\eps_{small}<\eps_{large}\le 1$, we partition the instance into:
\begin{itemize}
\item $I_{small}$, the set of rectangles with $h_i, w_i\le \eps_{small}N$, and we denote them as \emph{small} rectangles;
\item $I_{large}$, the set of rectangles with $h_i, w_i>\eps_{large}N$, and we denote them as \emph{large} rectangles;
\item $I_{hor}$, the set of rectangles with $w_i>\eps_{large}N$ and $h_i\le \eps_{small}N$, and we denote them as \emph{horizontal} rectangles;
\item $I_{ver}$, the set of rectangles with $h_i>\eps_{large}N$ and $w_i\le \eps_{small}N$, and we denote them as \emph{vertical} rectangles;
\item $I_{int}$, the set of remaining rectangles, and we denote them as \emph{intermediate} rectangles.
\end{itemize}
The corresponding intersection with $OPT$ defines the sets $OPT_{small}$, $OPT_{large}$, $OPT_{hor}$, $OPT_{ver}$ and $OPT_{int}$, respectively. As discussed in Section~\ref{sec:tdk_car:simple}, since any feasible solution contains at most $\frac{1}{\eps_{large}^2}$ large rectangles, we can assume in this case that $OPT_{large}=\emptyset$. Furthermore, thanks to Lemma~\ref{lem:item-classification}, $\eps_{small}$ and $\eps_{large}$ can be chosen in such a way that $\eps_{small}\le \eps_{large}\le \eps$, $\eps_{small}$ differs from $\eps_{large}$ by a large factor and $|OPT_{int}|\le \eps |OPT|$.
Building upon the corridors decomposition from~\cite{AW2013}, we will again consider $OPT_{T}$ (thin rectangles), $OPT_{F}$ (fat rectangles) and $OPT_K$ (killed rectangles) as defined in Section~\ref{sec:structural:lemma}. Thanks to Lemma~\ref{lem:boxProperties}, $|OPT_{K}| = O_{\eps}(1)$ and all the involved parameters can be fixed in such a way that the total height (resp. width) of $OPT_T\cap I_{hor}$ (resp. $OPT_T\cap I_{ver}$) is at most $\eps N$.
Recall that a subcorridor is called \emph{long} if its shortest edge has length at least $\frac{N}{2}$ and short otherwise. In the analysis of the algorithm we will again use sets $OPT_{LF}$, $OPT_{LT}$, $OPT_{SF}$ and $OPT_{ST}$ as defined in Section~\ref{sec:structural:boxes}, corresponding to rectangles from $OPT_F$ inside long corridors, rectangles from $OPT_T$ inside long corridors, rectangles from $OPT_F$ inside short corridors and rectangles from $OPT_T$ inside short corridors respectively.
For a given $\ell\in (\frac{N}{2},N]$, we let $I_{long}\subseteq I$ be the rectangles whose longest side has length longer than $\ell$ and $I_{short}=I\setminus I_{long}$.
We will assume as in the proof of Lemma~\ref{lem:ringCase} that $\ell = \left(\frac{1}{2}+2\eps_{large}\right)N$. That way we make sure that no long rectangle belongs to a short subcorridor (however it is worth remarking that long corridors may contain short rectangles).

Let us define $\ilopt:=\il\cap OPT$ and $\isopt:=\is\cap OPT$. Let us define $\efl=\sqrt \eps$. Note that $\efl \ge \eps \ge \epsl \ge \epss$. For simplicity and readability of the section, sometimes we will slightly abuse the notation and for any small constant depending on $\eps, \epsl, \epss$, we will just use $O({\efl})$.
Now we give a brief informal overview of the refinement and the cases before we go to the details.

\noindent{\bf Overview of the refined packing.}
For the refined packing we will consider several $\fontL\&C$ packings.
Some of the packings are just extensions  of previous constructions (such as from Theorem~\ref{thm:16/9-apx} and Lemma~\ref{lem:weighted-apx}).
Then we consider several other new $\fontL\&C$ packings where an $\fontL$-region is packed with items from $\il$ and the remaining region is used for packing items from $\is$ using Steinberg's theorem (See Theorem \ref{thm:steinberg}).
\ari{Note that in the definition of $\fontL\&C$ packing in Section \ref{sec:weighted}, we assumed the height of  the horizontal part of $L$-region and width of the vertical part of $L$-region to be the same. However, for these new packings we will consider $L$-regions where the height of  the horizontal part and width of vertical part may differ.}
Now several cases arise depending on the structure and profit of the $\fontL$-region. To pack items in $\isopt$ we have three options:\\
1. We can pack items in $\is$ using Steinberg's theorem into one rectangular region. Then we need both sides of the region to be greater than $\frac12+2\epsl$. \\
2. We can pack items in $\is$ using Steinberg's theorem such that
vertical and horizontal items are packed separately into different vertical and horizontal rectangular regions inside the knapsack.\\
3. If $a(\isopt)$ is large, we might pack only a small region with items in $\ilopt$, and use the remaining larger space in the knapsack to pack a significant fraction of profit from $\isopt$.\\
Now depending on the structure of the $L$-packing and $a(\isopt)$, we arrive at several different cases.
If the $\fontL$-region has very small width and height, we have case (1).
Else if the $\fontL$-region has very large width (or height), we have case (2B), where we pack nearly 
$\frac12|\ilopt|$  in the $\fontL$-region and then pack items from $\is$ in one large rectangular region.
Otherwise, we have case (2A), where either we pack only items from $\ilthin$ (See Lemma \ref{lem4cardgen}, used in case: (2Ai)) or nearly $3/4|\ilopt|$ (See Lemma \ref{lem4cardgen2}, used in cases (2Aii), (2Aiiia)) or in another case, we pack the vertical and horizontal items in $\isopt$ in two different regions through a more complicated packing (See case (2Aiiib)). The details of these cases can be found in the proof of Theorem \ref{thm:cardWorot}.

Now first, we start with some extensions of previous packings.
Note that by using analogous arguments as in the proof of Theorem~\ref{thm:16/9-apx}, we can derive the following inequalities leading to a $\left(\frac{16}{9}+O(\efl)\right)$-approximation algorithm.

\begin{equation} \label{lem1card}
|\optrc|\geq \frac{3}{4}|\ilopt|
\end{equation}

\begin{equation} \label{lem2card}
|\optrc|\geq \left(\frac{1}{2} -O(\efl)\right)|\ilopt|+\left(\frac{3}{4}-O(\efl)\right)|\isopt|
\end{equation}
\arir{Check $\eps$ dependency from mainbody} 

Now from Lemma~\ref{lem:boxProperties}, items in $\isthin$ can be packed into two containers of size $\ell \times \eps N$ and $\eps N \times \ell$. We can adapt part of the results in Lemma~\ref{lem:weighted-apx} to obtain the following inequalities.
\begin{pro} The following inequalities hold:
\begin{equation} \label{lem3card} |\optrc| \geq \wal{(1-O(\eps_{\fontL}))}(|\ilfat|+|\isfat|).\end{equation}
\begin{equation} \label{lem4card} |\optrc| \geq \wal{(1-O(\eps_{\fontL}))}(|\ilfat| + \frac{1}{2}(|\isfat| + |\ilthin|)).\end{equation} \end{pro}
\begin{proof} Inequality~\eqref{lem3card} follows directly from Lemma~\ref{lem:weighted-apx} since $\LF \cup \SF \wal{\cup OPT_{small}}= (\ilfat)\cup(\isfat)$ and both sets are disjoint. Inequality~\eqref{lem4card} follows from Lemma~\ref{lem:noShort}: if we consider the sum of the number of packed rectangles corresponding to the $4$ subcases associated with the case ``short horizontal/short vertical'', then every $i \in \ilfat\subseteq \LF$ appears four times, every $i\in \isopt \cap \LF$ appears four times, every $i\in \SF$ appears twice and every $i\in \ilthin$ appears twice. \wal{After including a $(1-O(\eps_{\fontL}))$ fraction of $OPT_{small}$}, and since $(\isopt\cap \LF) \cup \SF \wal{\cup OPT_{small}}= \isfat$, the inequality follows by taking average of the four packings. \end{proof}
\arir{Proposition might need to adapt the notations of main body}

The following theorem due to Steinberg \cite{steinberg1997strip} will be useful to pack items from $\isopt$ in order to obtain better packings.

\begin{thm}[A. Steinberg \cite{steinberg1997strip}] \label{thm:steinberg}
We are given a set of rectangles $I'$ and a box $Q$ of size $w \times h$.
Let $w_{max}\leq w$ and $h_{max}\leq h$ be the maximum width and maximum height among the items in $I'$ respectively.
Also we denote $x_+:=max(x,0)$.
If 
$$
2a(I') \le wh-(2w_{max}-w)_+(2h_{max}-h)_+
$$
then $I'$ can be packed into $Q$.
\end{thm}

\begin{cor}\label{lem:smallStein}
Let $I'$ be a set of rectangles such that $\displaystyle\max_{i \in I'} \height(i) \le \left(\frac{1}{2}+2\eps_{large}\right)N$ and $\displaystyle\max_{i \in I'} \width(i) \le \left(\frac{1}{2}+2\eps_{large}\right)N$.
Then for any $\alpha, \beta \le \frac{1}{2}-2\eps_{large}$, all of $I'$ can be packed into a knapsack  of width $(1-\alpha)N$ and height  $(1-\beta)N$ if
$$a(I') \le \Big(\frac12 -(\alpha+\beta)\left(\frac12+2\epsl\right)-8\epsl^2 \Big)N^2.$$
\end{cor}
\arir{removed the proof and other corollary}


Now we prove a more general version of Lemma~\ref{lem:ringCase} which holds for the cardinality case. 
\begin{lem}\label{lem4cardgen}
If $a(\isfat) \le \gamma N^2$ for any $\gamma \le 1$, then 
$$|\optrc| \geq \frac{3}{4}|\ilthin|+|\isthin|+\min\left\{1,\frac{1-O(\efl)}{2 \gamma}\right\}|\isfat|.$$
\end{lem}
\begin{proof}
As in Lemma~\ref{lem:ringCase}, we can pack $\frac34|\ilthin|+|\isthin|$ many rectangles in a boundary $\fontL$-region plus two boxes on the other two sides of the knapsack
and then a free square region with side length $(1- 3 \eps)N$ can be used to pack items from $\isfat$.
From Corollary~\ref{lem:smallStein}, any subset of rectangles of $\isfat$ with total area at most $(1-O(\efl))N^2/2$ can be packed into that square region of length $(1- 3 \eps)N$.
Thus we sort rectangles from $\isfat$ in the order of nondecreasing area and iteratively pick them until their total area reaches $(1- O(\efl) -\eps_{small})N^2/2$. Using Steinberg's theorem, there exists a packing of the selected rectangles. If $2\gamma \le 1- O(\efl) -\eps_{small}$ then the profit of this packing is $|\isfat|$, and otherwise the total profit is at least $\frac{1-O(\efl)}{2 \gamma}|\isfat|$. The packing coming from Steinberg's theorem may not be container-based, but we can then use resource augmentation as in Lemma~\ref{lem:ringCase} to obtain an $\fontL\&C$ packing.
\end{proof}

Now the following  lemma will be useful when $a(\isopt)$ is large. 

\begin{lem}\label{lem4cardgen2}
If $a(\isopt) > \gamma N^2$ for any $\gamma \ge \frac34+\eps+\eps_{large}$, then 
$$|\optrc|\geq \frac{3}{4}|\ilopt|+\frac{(3\gamma - 1 -O(\efl))}{4\gamma}|\isopt|.$$
\end{lem}
\begin{proof}
Similar to Lemma \ref{lem:LoftheRing} in Section~\ref{sec:tdk_car:simple}, we start from the optimal packing and move all rectangles in $\ilopt$ to the boundary such that all of them are contained in a boundary ring.  Note that unlike the case when we only pack $\ilthin$ in the boundary region, the boundary ring formed by $\ilopt$ may have width or height $\gg \eps N$.
\ari{Let us call the 4 stacks in the ring to be subrings.}
Let us assume that left and right subrings have width $\alpha_{left} N$ and $\alpha_{right} N$ respectively and 
bottom and top subrings have height $\beta_{bottom} N$ and $\beta_{top} N$ respectively. 
\wal{By possibly killing one of the long rectangles, subrings can be arranged such that $\alpha_{left}, \alpha_{right}, \beta_{bottom}, \beta_{top} \le 1/2$: If no vertical rectangle intersects the vertical line $x=\frac{N}{2}$ and no horizontal rectangle intersects the horizontal line $y=\frac{N}{2}$ this property holds directly. If one of the previous cases is not satisfied, by deleting such rectangle we can ensure the desired property at a negligible loss of profit, and notice that it is not possible that both cases happen at the same time since rectangles are long.}

As $a(\isopt) > \gamma N^2$, then $a(\ilopt) < (1-\gamma)N^2$.
Let us define $\alpha=\alpha_{left}+\alpha_{right}$ and $\beta=\beta_{bottom}+\beta_{top}$.
Then $(\alpha+\beta)N \cdot \frac{N}{2}  \le a(\ilopt)$, which implies that $\frac{\alpha+\beta}{2} < 1-\gamma$.
Hence, we get the following two inequalities:
\begin{equation} \label{steinLar1}
(\alpha+\beta) \le 2(1-\gamma);
\end{equation} 
\begin{equation} \label{steinLar2}
 a(\isopt) \le N^2-a(OPT_{long}) \le \left(1- \frac{(\alpha+\beta)}{2}\right)N^2.
 \end{equation} 
 Now consider the case when we remove the top horizontal subring and construct a boundary \fontL-region as in Lemma~\ref{lem:LoftheRing}. We will assume that rectangles in the $\fontL$-region are pushed to the left and bottom as much as possible. Then, the boundary $\fontL$-region has width $(\alpha_{left}+\alpha_{right})N$ and height $\beta_{bottom} N$. We will use Steinberg's theorem to show the existence of a packing of rectangles from $\isopt$ in a subregion of the remaining space with width $N-(\alpha_{left}+\alpha_{right}+\eps)N$ and height $N-(\beta_{bottom}+\eps)N$, and use the rest of the area for resource augmentation to get an $\fontL\&C$-based packing.
 Since $\gamma\ge \frac34+\eps+\eps_{large}$, we have from (\ref{steinLar1}) that $\alpha+\beta+2\eps \le 2(1-\gamma) + 2\eps \le 1/2-2\eps_{large}$.
 So $\alpha+\eps \le 1/2-2\eps_{large}$ and $\beta+\eps \le 1/2-2\eps_{large}$.
 Thus from Corollary \ref{lem:smallStein}, in the region with width $N-(\alpha_{left}+\alpha_{right}+\eps)N$ and height $N(1-\beta_{bottom}-\eps)$ we can pack rectangles from $\isopt$ of total area at most $\Big(\frac12 -\frac{(\alpha_{left}+\alpha_{right}+\beta_{bottom})}{2} - O(\efl) \Big)N^2$. 
 Hence, we can take the rectangles in $\isopt$ in the order of nondecreasing area until their total area reaches 
 $\Big(\frac12 -\frac{(\alpha_{left}+\alpha_{right}+\beta_{bottom})}{2} -O(\efl) - \eps_{small} \Big)N^2$ and pack at least $|\isopt| \cdot \frac{(\frac 12 - \frac{(\alpha_{left} +\alpha_{right}+\beta_{bottom})}{2}-O(\efl)-\eps_{small})}{(1-\frac{(\alpha+\beta)}{2})}$ using Steinberg's theorem. \arir{we need -$\epss$ as exactly 1/2 area might not have all integral items. feel free to add explanations.}
 If we now consider all the four different cases corresponding to removal of the four different subrings and take the average of profits obtained in each case, 
 we pack at least \begin{eqnarray*} & & \frac3{4}|\ilopt|+ |\isopt| \cdot \left(\frac{(\frac 12 - \frac38(\alpha_{left} +\alpha_{right}+\beta_{bottom}+\beta_{top})-O(\efl)}{(1-\frac{(\alpha+\beta)}{2})}\right) \\ & = & \frac3{4}|\ilopt|+ |\isopt| \cdot \left(\frac{(\frac 12 - \frac38(\alpha+\beta) -O(\efl))}{(1-\frac{(\alpha+\beta)}{2})}\right) \\ & \ge & \frac3{4}|\ilopt|+ |\isopt| \cdot \frac{3\gamma -1-O(\efl)}{4\gamma},
 \end{eqnarray*} where the last inequality follows from \eqref{steinLar1} and the fact that the expression is decreasing as a function of $(\alpha + \beta)$.\end{proof}
Now we start with the proof of Theorem \ref{thm:cardWorot}.
\begin{proof}[Proof of Theorem~\ref{thm:cardWorot}]
In the refined analysis, we will consider different solutions and show that the best of them always achieves the claimed approximation guarantee. We will pack some rectangles in a boundary $\fontL$-region (either a subset of only $\ilthin$ or a subset of $\ilopt$) using the PTAS for $\fontL$-packings described in Section~\ref{sec:ptasL}, and in the remaining area of the knapsack (outside of the boundary $\fontL$-region), we will pack a subset of rectangles from $\isopt$.
\arir{make consistent with main body}

\begin{figure}[t!]
            \centering
        \includegraphics[width=6in]{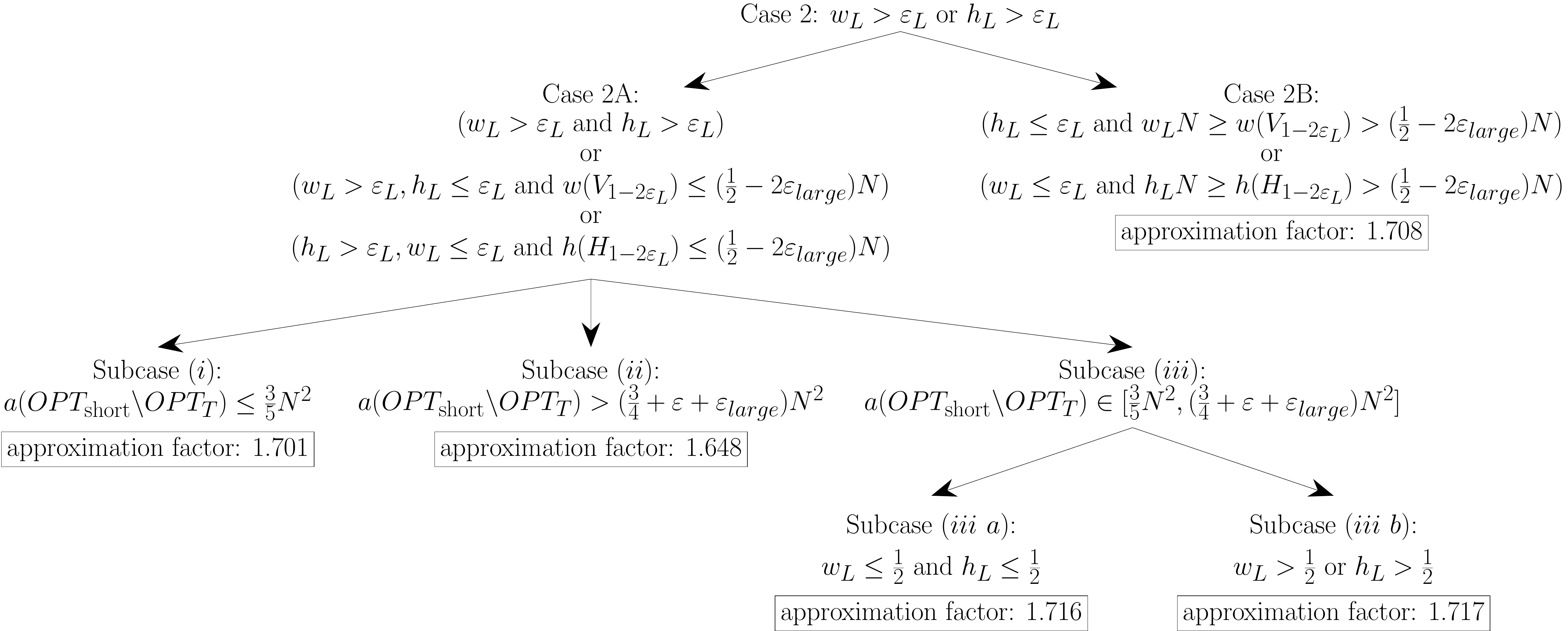}
        \caption{Summary of the cases.}
        \label{fig:cardworotcaseover}
    \end{figure}%

Consider the ring as constructed in the beginning of the proof of Lemma \ref{lem4cardgen2}.
Then we remove the least profitable subring and repack the remaining rectangles from $\ilopt$ in a boundary $\fontL$-region.  
W.l.o.g. assume that the horizontal top subring was the least profitable subring. The other cases are analogous.
We will use the same notation as in Lemma~\ref{lem4cardgen2}, and also define $w_{\fontL}=(\alpha_{left}+\alpha_{right}), h_{\fontL}= \beta_{bottom}$.
Now let us consider two cases \ari{(see Figure \ref{fig:cardworotcaseover} for an overview of the subcases of case 2).}\\

\noindent \textbf{$\bullet$ Case 1. $w_{\fontL} \le \eps_{\fontL} , h_{\fontL} \le \eps_{\fontL} $.}   \arir{$\eps$ coming from Lemma \ref{lem4card}.}\\
In this case, following the proof of Lemma \ref{lem4cardgen} (using $\gamma=1$), we can pack $\frac{3}{4}|\ilopt|+|\isthin|+\frac{1-O(\eps_{\fontL})}{2}|\isfat|$.
This along with inequalities \eqref{lem2card}, \eqref{lem3card} and \eqref{lem4card} will give us a solution with good enough approximation factor. Check Section~\ref{sec:bound_apx} and Table~\ref{tab:summary} for details. 

\noindent \textbf{$\bullet$ Case 2. $w_{\fontL} > \eps_{\fontL} $ or  $h_{\fontL} > \eps_{\fontL} $.}\\
Let $V_{1- 2 \eps_{\fontL}}$ be the set of vertical rectangles having height strictly larger than $(1-2 \eps_{\fontL})N$. Let us define $w(V_{1-2 \eps_{\fontL}})=\sum_{i \in V_{1-2 \eps_{\fontL}}} \width(i)$. 
Similarly, let $H_{1- 2 \eps_{\fontL}}$ be the set of horizontal rectangles of width \wal{strictly larger than} $(1-2 \eps_{\fontL})N$ and $h(H_{1-2 \eps_{\fontL}})=\sum_{i \in H_{1-2 \eps_{\fontL}}} \height(i)$. \\
\textit{$\Diamond$  Case 2A.  \Big($w_{\fontL} > \eps_{\fontL} $ and  $h_{\fontL}  > \eps_{\fontL} $\Big) or \Big($w_{\fontL} > \eps_{\fontL} , h_{\fontL} \le \eps_{\fontL} , $ and $w(V_{1-2\eps_{\fontL}})\le \left(\frac{1}{2}-2\eps_{large}\right) N$\Big) or \Big($h_{\fontL}  > \eps_{\fontL} , w_{\fontL} \le \eps_{\fontL}$, and $h(H_{1-2\eps_{\fontL}})\le \left(\frac{1}{2}-2\eps_{large}\right) N$\Big).}\\
We will show that if any of the above three conditions is met, then we can pack $\frac{3(1-O(\eps))}{4}|\ilopt|+|\isthin|$ in a boundary $\fontL$-region of width close to $ w_{\fontL} N$ and height close to $h_{\fontL} N$, and 
then in the remaining area we will pack some rectangles from $\isfat$ using Steinberg's theorem and resource augmentation. 

\begin{figure*}[t!]
    \centering
    \begin{subfigure}[b]{0.4\textwidth}
            \centering
        \includegraphics[width=2.5in]{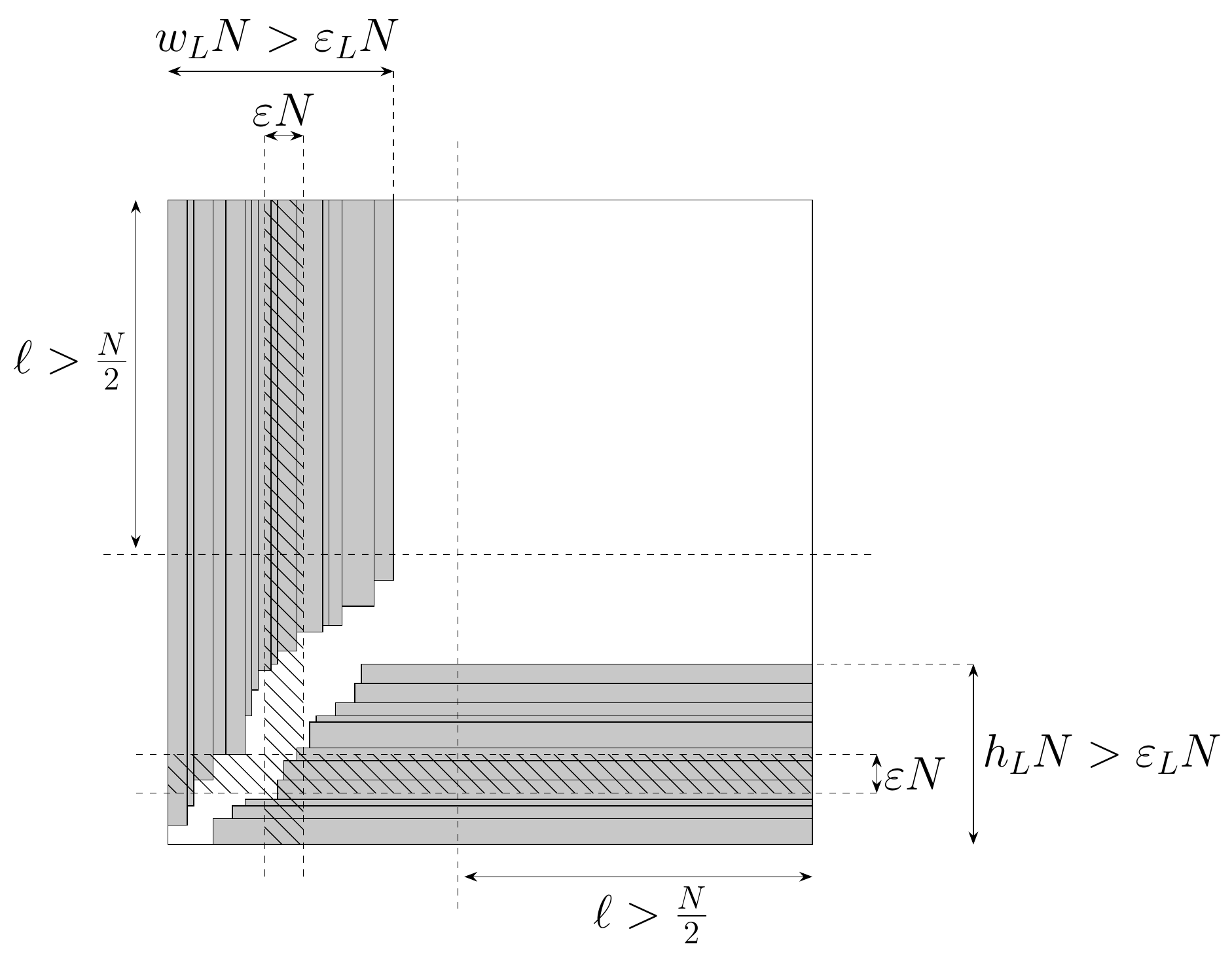}
        \caption{Packing of $\fontL$-region using rectangles from $\ilopt$. Striped strips are cheapest $\eps N$-width and cheapest $\eps N$-height.}
        \label{fig:cardworotcase2a1}
    \end{subfigure}%
    \hspace{40pt}
    \begin{subfigure}[b]{0.35\textwidth}
        \centering
        \includegraphics[width=2.5in]{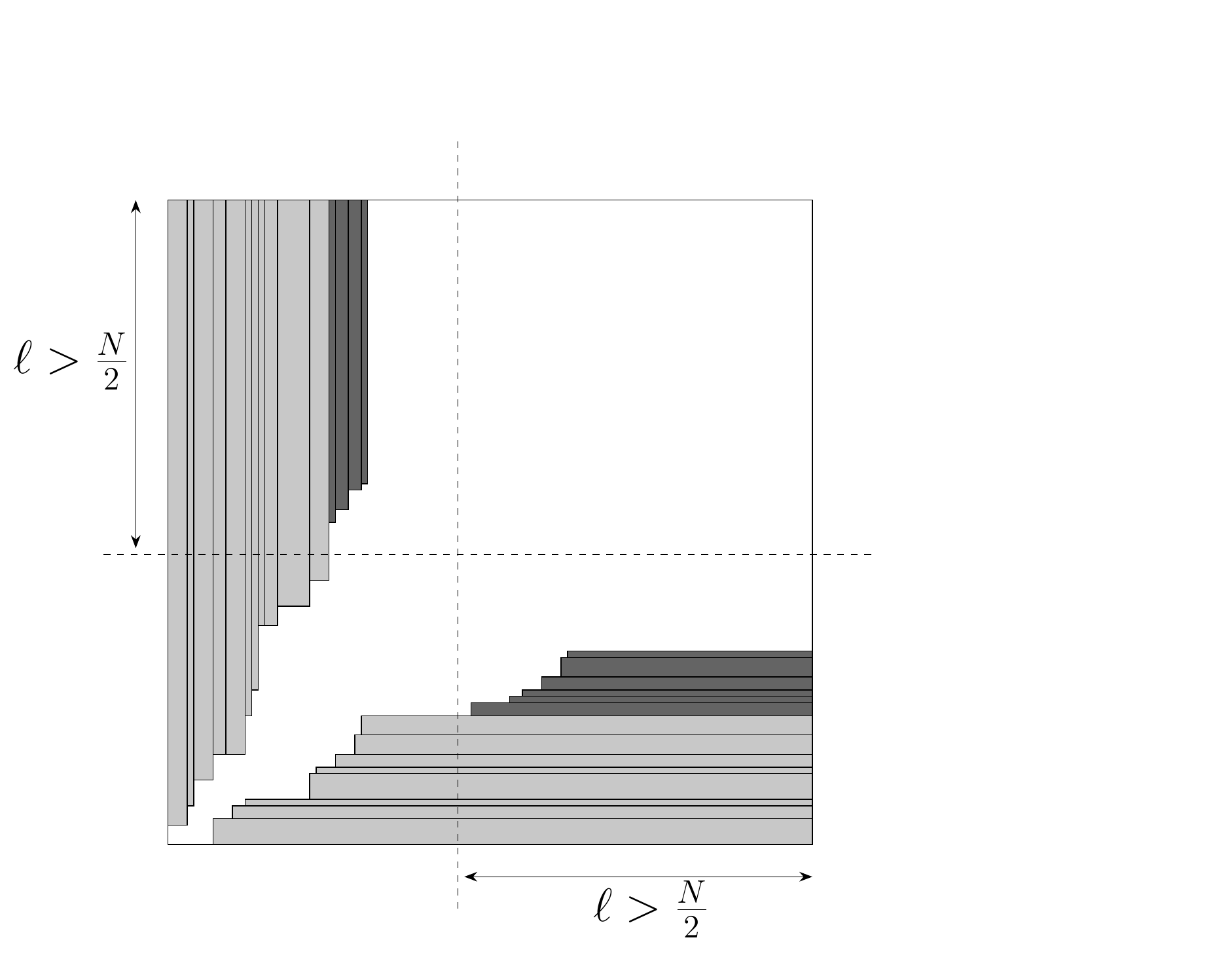}
        \caption{Packing of rectangles in $\ilopt \cup (\isthin)$. Dark gray rectangles are from $\isthin$. }
             \label{fig:cardworotcase2a2}
    \end{subfigure}
    ~
      \caption{The case for $w_{\fontL} > \eps_{\fontL} $ and  $h_{\fontL} > \eps_{\fontL}$.}
\end{figure*}

\walr{Maybe in Figure 6a we can emphasize that the space for horizontal $\isthin$ is free}
\begin{figure*}[t!]
    \centering
    \begin{subfigure}[b]{0.4\textwidth}    
        \centering
        \includegraphics[width=2.5in]{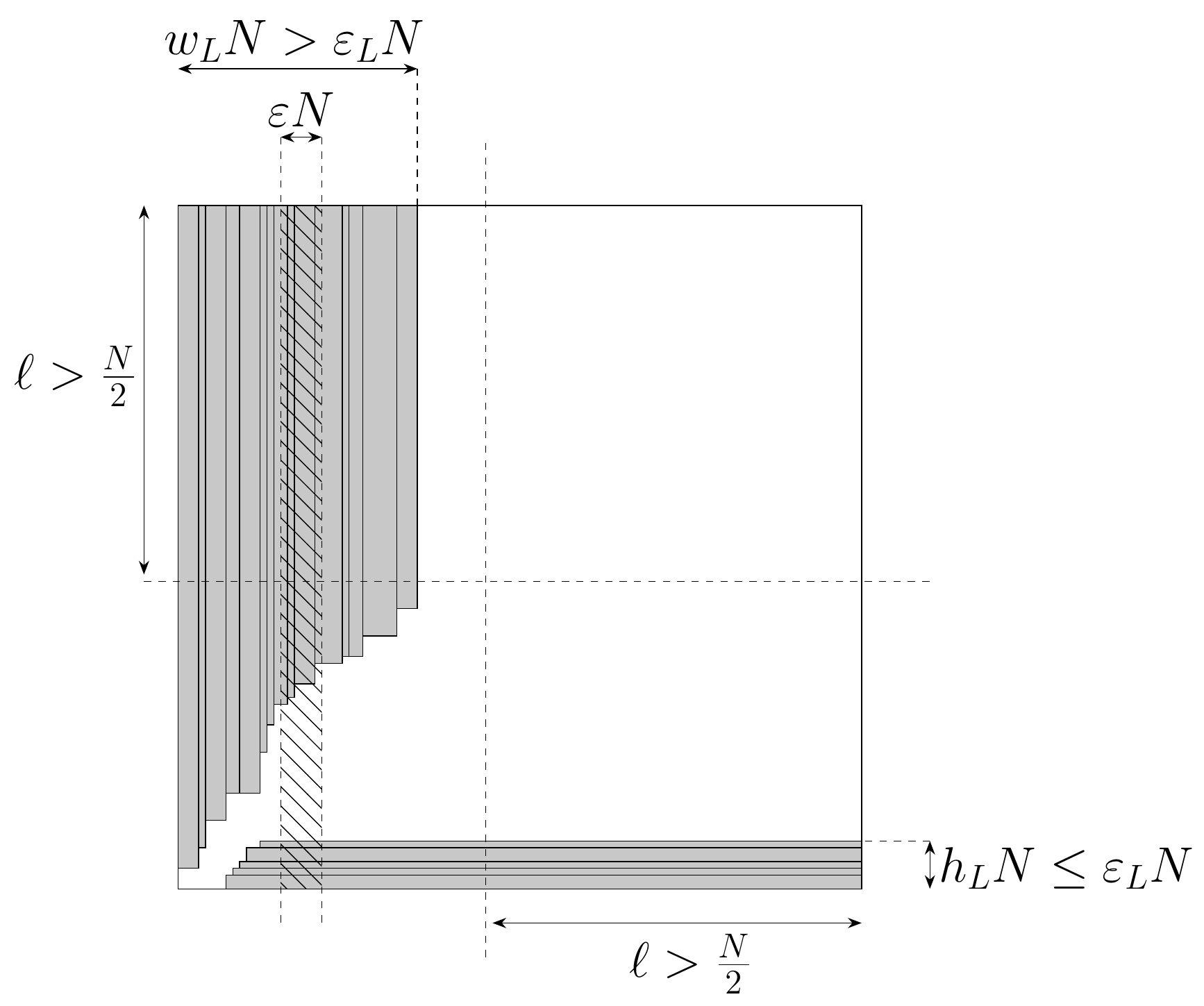}
        \caption{Packing of $\fontL$-region using rectangles from $\ilopt$. Striped strip is the cheapest $\eps N$-width strip.}
         \label{fig:cardworotcase2a3}
    \end{subfigure}%
    \hspace{40pt}
    \begin{subfigure}[b]{0.35\textwidth}   
        \centering
        \includegraphics[width=2.5in]{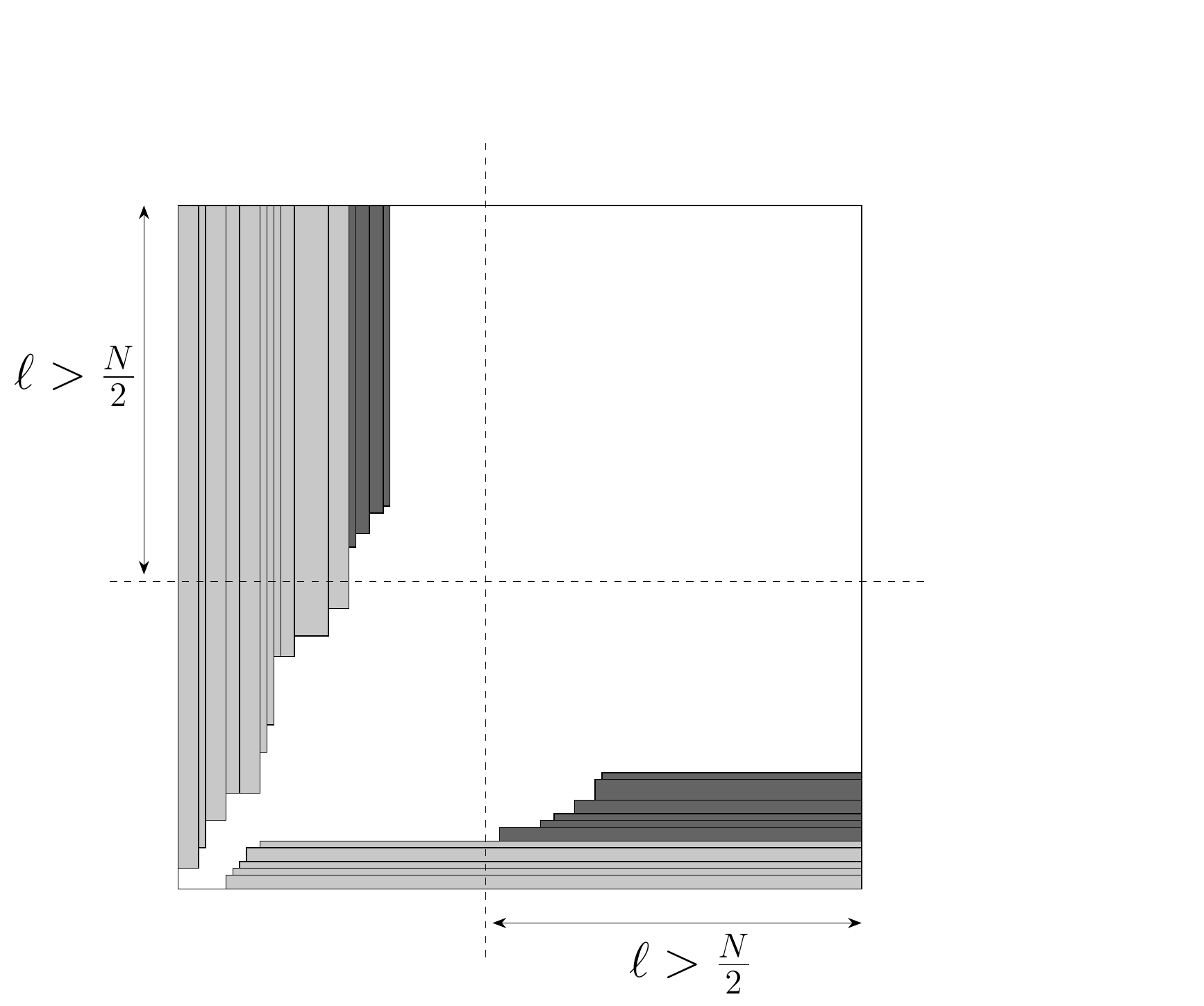}
        \caption{Packing of rectangles in $\ilopt \cup (\isthin)$. Dark gray rectangles are from $\isthin$.}
        \label{fig:cardworotcase2a4}
    \end{subfigure}
    ~
      \caption{The case for $w_{\fontL} > \eps_{\fontL} $ and  $h_{\fontL} \le \eps_{\fontL}$.}
\end{figure*}

\noindent {\bf Packing of subset of rectangles from $\ilopt \cup (\isthin)$ into $\fontL$-region.}\\
If ($w_{\fontL} > \eps_{\fontL} $ and  $h_{\fontL} > \eps_{\fontL} $), we partition the vertical part of the $\fontL$-region into consecutive strips of width $\eps N$.
Consider the strip that intersects the least number of vertical rectangles from $\ilopt$ among all strips, and we call it to be the {\em cheapest $\eps N$-width vertical strip} (See Figure \ref{fig:cardworotcase2a1}). 
Clearly the cheapest $\eps N$-width vertical strip intersects at most a $\frac{\eps+2\epss}{\eps_\fontL} \le 3\efl$ fraction of the rectangles in the vertical part of the $\fontL$-region, so we can remove all such vertical rectangles intersected by that strip \wal{at a small loss of profit}.
Similarly, we remove the horizontal rectangles intersected by the cheapest $\eps N$-height horizontal strip in the boundary $\fontL$-region.
We now pack the horizontal \wal{container for $\isthin$} in the free region left by the removed horizontal strip, and \wal{the vertical container for} $\isthin$ in the free region left by the removed vertical strip. Similarly to the proof of Lemma~\ref{lem:LoftheRing} we can sort rectangles in the vertical (resp. horizontal) pile of the $\fontL$-region according to their height (resp. width), obtaining a feasible $\fontL\&C$-packing (See Figure \ref{fig:cardworotcase2a2}). \\
In the other case ($w_{\fontL} > \eps_{\fontL} , h_{\fontL} \le \eps_{\fontL} $ and $w(V_{1-2\eps_{\fontL}})\le \left(\frac{1}{2}-2\eps_{large}\right)N$), we can again remove the  cheapest $\eps N$-width vertical strip  in the boundary $\fontL$-region and pack \wal{the vertical container for} $\isthin$ there (See Figure \ref{fig:cardworotcase2a3}). Now we show how to pack horizontal items from $\isthin$.
In the packing of the boundary $\fontL$-region, we can assume that the vertical rectangles are sorted non-increasingly by height from left to right and pushed upwards until they touch the top boundary.
Then, since $w(V_{1- 2 \eps_{\fontL}})\le \left(\frac{1}{2}-2\eps_{large}\right)N$ and ($h_{\fontL} \le \eps_{\fontL}$), the region $\left[\left(\frac{1}{2}-2\eps_{large}\right)N, N\right] \times \left[\eps_{L} N, 2\eps_{L} N\right]$ will be completely empty and thus we will have enough space to pack the \wal{horizontal container for} $\isthin$ on top of the horizontal part of the $\fontL$-region (See Figure \ref{fig:cardworotcase2a4}). This leads to a  packing in a boundary $\fontL$-region of width at most $w_L N$ and height at most $(h_L+\eps_{L})N$ with total profit at least $\frac{3(1-O(\eps))}{4}|\ilopt|+|\isthin|$. The last case, when $w_{\fontL} \le \eps_{\fontL}$, is analogous, leading to a packing into a boundary $\fontL$-region of width at most $(w_L+\eps_{L})N$ and height at most $h_L N$ with at least the same profit.
Thus, 
\begin{equation}\label{CardCase2a}
|\optrc| \ge \frac{3(1-O(\eps_{L}))}{4}|\ilopt|+|\isthin|
\end{equation} 

\noindent{\bf Packing of a subset of rectangles from $\isfat$ into  the remaining region.}\\
Note that after packing \wal{at least} $\frac{3(1-O(\eps))}{4}|\ilopt|+|\isthin|$ many rectangles in the boundary $\fontL$-region, the remaining rectangular region of width $(1-w_{\fontL}-\eps_{\fontL})N$ and height $(1-h_{\fontL}-\eps_{\fontL})N$ is completely empty.
Now we will show the existence of a packing of some rectangles from $\isfat$ \wal{in the remaining space of the packing (even using some space from the $\fontL$-boundary region)}.
Let $\isfhor := ((\isfat) \cap \Rho) \cup ((\isfat) \cap I_{small}) $ and $\isfver := (\isfat) \cap \Rve$. 
Let us assume w.l.o.g. that vertical rectangles are shifted as much as possible to the left and top of the knapsack and horizontal ones are pushed as much as possible to the right and bottom. We divide the analysis in three subcases depending on $a(\isfat)$. \\
{\em $-$  Subcase (i).} If $a(\isfat) \le \frac35 N^2$, from inequalities \eqref{lem2card}, \eqref{lem3card}, \eqref{lem4card}, \eqref{CardCase2a} and Lemma \ref{lem4cardgen}, we get a solution with good enough approximation factor. Check Section~\ref{sec:bound_apx} and Table~\ref{tab:summary} for details.

\noindent {\em $-$ Subcase (ii).} If $a(\isfat) > (\frac34+\eps+\eps_{large}) N^2$, from inequalities  \eqref{lem2card}, \eqref{lem3card}, \eqref{lem4card}, \eqref{CardCase2a} and Lemma \ref{lem4cardgen2}, we get a solution with good enough approximation factor. Check Section~\ref{sec:bound_apx} and Table~\ref{tab:summary} for details.

\noindent{\em $-$ Subcase (iii).} Finally, if $\frac35 N^2 \le a(\isfat) \le (\frac34+\eps+\eps_{large}) N^2$, from inequality \eqref{steinLar1} we get $\alpha + \beta \le 2(1-\frac35) = \frac45$. Now we consider two subcases.\\
{\em $\odot$ Subcase (iii a): $w_{\fontL}\le \frac{1}{2}$ and $h_{\fontL}\le \frac{1}{2}$.}
\arir{changed from $w_{\fontL}\le \frac{1}{2}-2\epsl$  to remove the previous gap between iiia and iiib.}
Note that in this case if $w_{\fontL}\ge \frac{1}{2}-2\eps_{large}-2\eps_{\fontL}$ (resp., $h_{\fontL}\ge \frac{1}{2}-2\eps_{large}-2\eps_{\fontL}$),
we can remove the cheapest $2(\eps_{\fontL}+\epsl) N$-width vertical (resp., horizontal) strip from the $\fontL$-region by removing $O(\eps_{\fontL})$ fraction of rectangles in $\ilopt$. 
Otherwise we have $w_{\fontL}<\frac{1}{2}-2\eps_{large}-2\eps_{\fontL}$ and $h_{\fontL}< \frac{1}{2}-2\eps_{large}-2\eps_{\fontL}$.
So there is a free rectangular region that has both side lengths at least $N(\frac12+2\eps_{large}+\eps_\fontL)$; we will keep $\eps_\fontL N$ width and $\eps_\fontL N$  height for resource augmentation and use the rest of the rectangular region (with both sides length at least $\left(\frac{1}{2}+2\eps_{large}\right)N$) for showing existence of a packing using Steinberg's theorem.

Note that this free rectangular region has area at least $N^2(1-w_{\fontL}-2\eps_{\fontL})(1-h_{\fontL}-2\eps_{\fontL})$. 
Now consider rectangles from $\isfhor$ (by sorting them non-decreasingly by area and picking them iteratively) until their total area becomes at least $min\{ a(\isfhor), \frac{N^2(1-w_{\fontL}-2\eps_{\fontL})(1-h_{\fontL}-2\eps_{\fontL})}{2} -\epss N^2\}$. 
Thus their total area is at most $\le  \frac{N^2(1-w_{\fontL}-2\eps_{\fontL})(1-h_{\fontL}-2\eps_{\fontL})}{2}$ as the area of any rectangle in $\isfhor$  is at most $\epss N^2$. Hence, from Steinberg's theorem, we can pack these rectangles in the free rectangular region.
Similarly, we can pack there rectangles from $\isfver$ with total area at least $min\{ a(\isfver), \frac{N^2(1-w_{\fontL}-2\eps_{\fontL}) (1-h_{\fontL}-2\eps_{\fontL})}{2}-\epss N^2\}$. 


Since rectangles are sorted non-decreasingly according to their areas, the total profit of the aforementioned packings is bounded below by $min\{1, \left(\frac{(1-w_{\fontL})(1-h_{\fontL})}{2a(\isfhor)} -O(\eps_{L})\right) N^2\}|\isfhor|$ and $min\{ 1, \left(\frac{(1-w_{\fontL}) (1-h_{\fontL})}{2 a(\isfver)}-O(\eps_{L})\right) N^2\}|\isfver|$ respectively.
We claim that if we keep the best of the two packings, we can always pack at least $\left(\frac{7}{48}-O(\eps_L)\right)|\isfat|$ many rectangles. To show this we will consider the four possible cases:

\begin{itemize}[leftmargin=10pt]
	\item If $min\{ 1, \left(\frac{(1-w_{\fontL})(1-h_{\fontL})}{2a(\isfhor)} -O(\eps_{L})\right) N^2\}=min\{ 1, \left(\frac{(1-w_{\fontL}) (1-h_{\fontL})}{2 a(\isfver)}-O(\eps_{L})\right) N^2\}=1$, then, by an averaging argument, the best among the two packings has profit at least $\frac{1}{2}(|\isfver|+|\isfhor|) = \frac{1}{2}|\isfat|$.
	\item If $\left(\frac{(1-w_{\fontL})(1-h_{\fontL})}{2a(\isfhor)} -O(\eps_{L})\right) N^2<1$ and $\left(\frac{(1-w_{\fontL}) (1-h_{\fontL})}{2 a(\isfver)}-O(\eps_{L})\right) N^2<1$, then by an averaging argument we pack at least \begin{align*} & \resizebox{0.96\hsize}{!}{$\frac{N^2}{2}\left( (\frac{(1-w_{\fontL})(1-h_{\fontL})}{2a(\isfhor)} -O(\eps_{L}))|\isfhor|+ (\frac{(1-w_{\fontL}) (1-h_{\fontL})}{2 a(\isfver)}-O(\eps_{L}))|\isfver|\right)$} \\ \ge & \resizebox{0.445\hsize}{!}{$\frac{N^2}{2}\left(  \frac{(1-w_{\fontL})(1-h_{\fontL})}{2a(\isfat)} -O(\eps_{L}) \right)|\isfat|$} \end{align*} where the inequality follows from the fact that $\frac{a}{b}+\frac{c}{d} \ge \frac{(a+c)}{(b+d)}$ for $a,b,c,d \ge 0$.
Since $a(\isfat)\le (N^2 -a(\ilopt)) \le (1-\frac{\alpha}{2}-\frac{\beta}{2})N^2 \le (1-\frac{w_{\fontL}}{2}-\frac{h_{\fontL}}{2})N^2
$
and $w_\fontL+ h_\fontL \le \alpha+\beta \le \frac45$, the amount of rectangles we are packing from $\isfat$ is bounded below by the minimum of $$f(h_L,w_L) = \left(  \frac{(1-w_{\fontL})(1-h_{\fontL})}{ (4-2w_{\fontL}-2h_{\fontL})} - O(\eps_L)\right)N^2|\isfat|$$ over the domain $\{w_{\fontL} + h_{\fontL}\le \frac{4}{5}, 0\le w_{\fontL}\le \frac{1}{2}, 0\le h_{\fontL} \le \frac{1}{2}\}$. 
Since $\frac{\partial f(h_L,w_L)}{\partial h_L} =\frac{-2(1-w_L)^2}{(4-2w_L -2h_L)^2}\le 0$ and $\frac{\partial f(h_L,w_L)}{\partial w_L} =\frac{-2(1-h_L)^2}{(4-2w_L -2h_L)^2} \le 0$, the function is decreasing with respect to both its arguments, implying that the minimum value must be attained when $h_L+w_L=\frac{4}{5}$. This in turn implies that the amount of rectangles from $\isfat$ we are packing is bounded below by the minimum of $f(h_L,\frac{4}{5}-h_L)$ over the interval $[\frac{3}{10}, \frac{1}{2}]$. Since $$f(h_L,\frac{4}{5}-h_L) = \left(\frac{5}{12}(1-h_L)(\frac{1}{5}-h_L)-O(\eps_L)\right)N^2 |\isfat|$$ describes a parabola centered at $h_L=\frac{2}{5}$, the minimum value on the aforementioned interval is attained at both limits $h_L=\frac{3}{10}$ and $h_L=\frac{1}{2}$ with a value of $\left(\frac{7}{48}-O(\eps_L)\right)|\isfat|$. 
\item If $min\{ 1, \left(\frac{(1-w_{\fontL})(1-h_{\fontL})}{2a(\isfhor)} -O(\eps_{L})\right) N^2\}=1$ and $\left(\frac{(1-w_{\fontL}) (1-h_{\fontL})}{2 a(\isfver)}-O(\eps_{L})\right) N^2<1$ (the remaining case being analogous), then we are packing at least \begin{align*} & \frac{1}{2}\left(|\isfhor|+ (\frac{(1-w_{\fontL}) (1-h_{\fontL})}{2 a(\isfver)}-O(\eps_{L}))N^2|\isfver|\right) \\ \ge & \frac{N^2}{2}\left( \frac{(1-w_{\fontL})(1-h_{\fontL})}{2a(\isfver)} -O(\eps_{L}) \right)(|\isfhor| + |\isfver|) \\ \ge & \frac{N^2}{2}\left(  \frac{(1-w_{\fontL})(1-h_{\fontL})}{2a(\isfat)} -O(\eps_{L}) \right)|\isfat| \\ \ge & \left(\frac{7}{48} - O(\eps_{L})\right)|\isfat|, \end{align*} where the last inequality comes from the analysis of the previous case.
\end{itemize}
From this we conclude that $$|OPT_{\fontL\& C}|\ge \frac{3(1-O(\eps_L))}{4}|\optln|+|\isthin|+\left(\frac{7}{48}-O(\eps_L)\right)|\isfat|.$$
This together with inequalities \eqref{lem2card}, \eqref{lem3card}, \eqref{lem4card} and Lemma \ref{lem4cardgen} gives us a solution with good enough approximation factor. Check Section~\ref{sec:bound_apx} and Table~\ref{tab:summary} for details.

\begin{figure}[t!]
            \centering
        \includegraphics[width=3.5in]{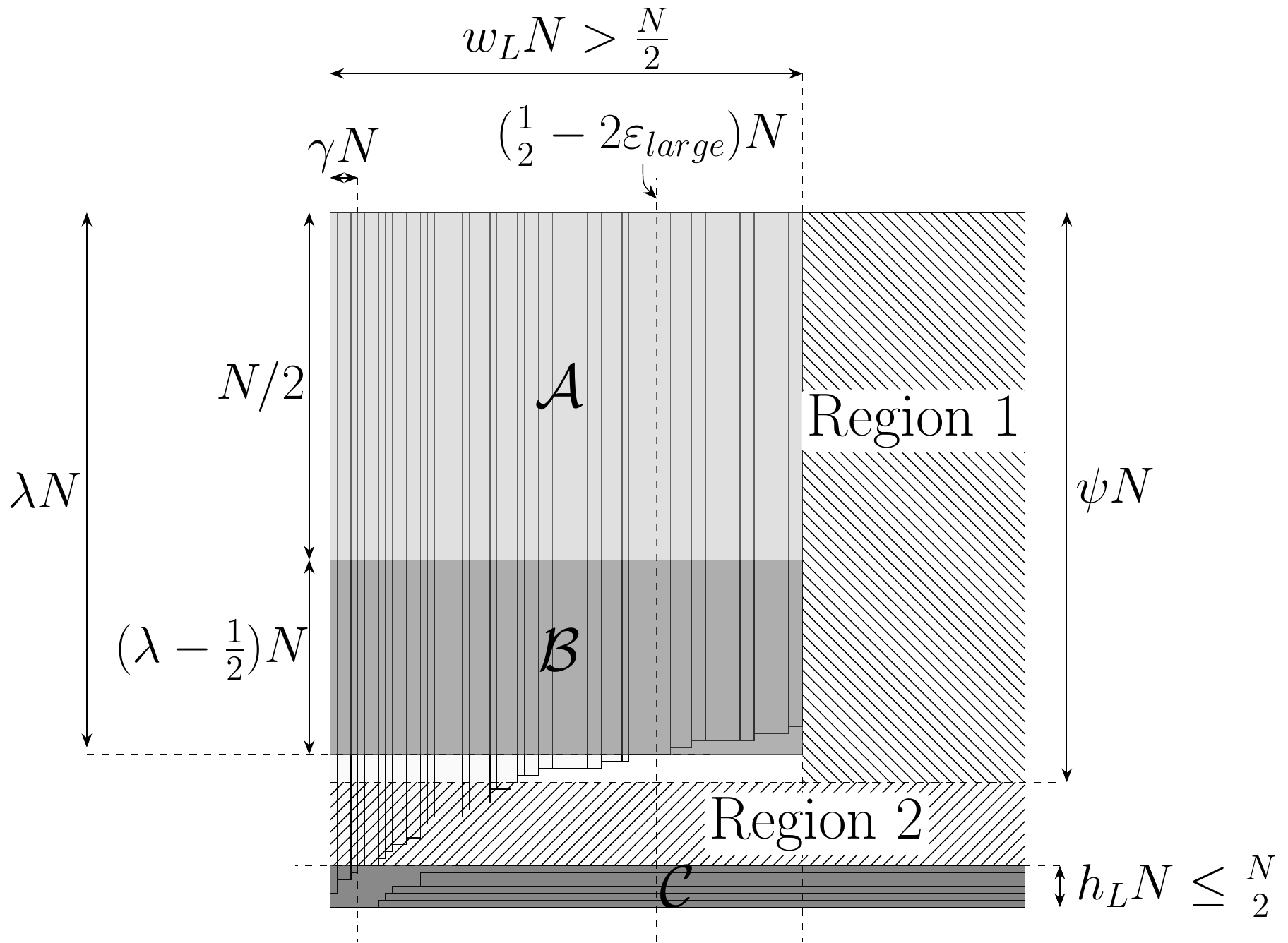}
        \caption{Case 2A(iii)b in the proof of Theorem~\ref{thm:cardWorot}}
        \label{fig:cardworotcase2aiii}
    \end{figure}%

\noindent {\em $\odot$ Subcase (iii b): } $w_{\fontL}>\frac{1}{2}$ (then from inequality \eqref{steinLar1}, $h_{\fontL} \le \frac{3}{10}$).
Note that $a(\ilopt) \le (1-\frac35)N^2=\frac25 N^2$.\\
Let us define some parameters from the current packing to simplify the calculations. Let $\lambda N$ be the height of the rectangle in the packing that intersects or  touches the vertical line $x=\left(\frac{1}{2}-2\eps_{large}\right)N$ (if two rectangles touch such line, we choose that tallest one) and $\gamma N$ be the total width of vertical rectangles having height greater than $(1-h_\fontL)N$. We define also the following three regions in the knapsack: $\mathcal{A}$, the rectangular region of width $w_{\fontL}N$ and height $\frac{1}{2}N$ in the top left corner of the knapsack; $\mathcal{B}$, the rectangular region of width $w_{\fontL} N$ and height $(\lambda-\frac{1}{2})N$ below $\mathcal{A}$ and left-aligned with the knapsack; and $\mathcal{C}$, the rectangular region of width $N$ and height $h_{\fontL}N$ touching the bottom boundary of the knapsack. Notice that $\mathcal{A}$ is fully occupied by vertical rectangles, $\mathcal{B}$ is almost fully occupied by vertical rectangles except for the right region of width $w_{\fontL} N- \left(\frac{1}{2}-2\eps_{large}\right)N$, and at least half of $\mathcal{C}$ is occupied by horizontal rectangles (some vertical rectangles may overlap with this region).
Our goal is to pack some rectangles from $\isfat$ in the \reflectbox{\ffmfamily{L}}-shaped region outside $\mathcal{A}\cup\mathcal{B}\cup\mathcal{C}$. Let $\psi \in [\lambda, 1-h_\fontL]$ be a parameter to be fixed later. We will use, when possible, the following regions for packing items from $\isfat$: Region 1 on the top right corner of the knapsack with width $N(1-w_{\fontL})$ and height $\psi N$ and Region 2 which is the rectangular region $[0,N]\times [h_{\fontL} N, (1-\psi )\cdot N]$ (see Figure~\ref{fig:cardworotcase2aiii}). Region $1$ is completely free but Region 2 may overlap with vertical rectangles. 

\arir{changed this para, check!}
We will now divide Region 2 into a constant number of boxes such that: they do not overlap with vertical rectangles, the total area inside Region 2 which is neither overlapping with vertical rectangles nor covered by boxes is at most $O(\eps_\fontL) N^2$ and each box has width at least $\left(\frac{1}{2}+2\eps_{large}\right)N$ and height at least $\eps N$. That way we will be able to pack rectangles from $\isfver$ into the box defined by Region $1$ and rectangles from $\isfhor$ into the boxes defined inside Region 2 \wal{using almost completely its free space}.
In order to create the boxes inside Region 2 we first create a monotone chain by doing the following: Let $(x_1,y_1)=(\gamma N, h_{\fontL})$. Starting from position $(x_1,y_1)$, we draw an horizontal line of length $\eps_\fontL N$ and then a vertical line from bottom to top until it touches a vertical rectangle, reaching position $(x_2, y_2)$. From $(x_2, y_2)$ we start again the same procedure and iterate until we reach the vertical line $x=\left(\frac{1}{2}-2\eps_{large}\right)N$ or the horizontal line $y=(1-\psi )N$. 
Notice that the area above the monotone chain and below $y=(1-\psi) N$ that is not occupied by vertical rectangles, is at most $\sum_i{\eps_{\fontL} N (y_{i+1}-y_i)} \le \eps_\fontL N^2$.
The number of points $(x_i, y_i)$ defined in the previous procedure is at most $1/\eps_\fontL$. By drawing an horizontal line starting from each $(x_i,y_i)$ up to $(N,y_i)$, together with the drawn lines from the monotone chain and the right limit of the knapsack, we define $k\le 1/\eps_\fontL$ boxes. We discard the boxes having height less than $\eps N$, whose total area is at most $\frac{\eps}{\eps_\fontL}N^2= \eps_\fontL N^2$, and have all the desired properties for the boxes. 

Note that the area that is occupied for sure by rectangles in $\ilopt$   in regions $\mathcal{A}, \mathcal{B}$ and $\mathcal{C}$ by rectangles \wal{in $\ilopt$} is at least $(\frac{1}{2}w_{\fontL}+ (\lambda-\frac{1}{2})(\frac{1}{2}-2\epsl) + \frac{1}{2}h_{\fontL})N^2$.
Since the total area of rectangles from $\ilopt$ is at most $\frac{2}{5}N^2$, the total area occupied by rectangles in $\ilopt$ in Region 2 is at most $N^2(\frac{2}{5} - \frac{1}{2}w_{\fontL} - (\lambda-\frac{1}{2})(\frac{1}{2}-2\epsl)  - \frac{1}{2}h_{\fontL})\le  N^2(\frac{13}{20}-\frac{w_{\fontL}}{2}-\frac{\lambda}{2}-\frac{h_{\fontL}}{2}+\epsl)$. This implies that the total area of the horizontal boxes is at least $N^2(1-\psi -h_{\fontL})-N^2(\frac{13}{20}-\frac{w_{\fontL}}{2}-\frac{\lambda}{2}-\frac{h_{\fontL}}{2})-O(\eps_\fontL) N^2$ and the area of the vertical box is $(1-w_{\fontL})\psi N^2$. 
Ignoring the $O(\epsl)$-term, these two areas become equal if we set $\psi =\frac{7+10(w_{\fontL}+\lambda -h_{\fontL})}{40-20w_{\fontL}}$. It is not difficult to verify that in this case $\psi \le 1-h_{\fontL}$.
If $\frac{7+10(w_{\fontL}+\lambda -h_{\fontL})}{40-20w_{\fontL}} \ge \lambda$, then we set 
$\psi= \frac{7+10(w_{\fontL}+\lambda -h_{\fontL})}{40-20w_{\fontL}}$. Otherwise we set $\psi=\lambda$.\walr{Erased $-O(\eps_{\fontL})$ above. I think they are not needed}

First, consider the case $\psi= \frac{7+10(w_{\fontL}+\lambda -h_{\fontL})}{40-20w_{\fontL}}$.
Since $\psi \ge \lambda$, boxes inside Region 2 have width at least $\left(\frac{1}{2}+2\eps_{large}\right)N$ and height at least $\eps N \gg \eps_{small}N$ (recall that $\eps_{small}$ differs by a large factor from $\eps_{large}\le\eps$), and the box in Region $1$ has height at least $\left(\frac{1}{2}+2\eps_{large}\right)N$ and width at least $\frac{1}{5}N\gg\eps_{small} N$. By using Steinberg's theorem, we can always pack in these boxes at least \begin{equation*} \resizebox{\hsize}{!}{$\left(\min\left\{1,\frac{\frac{1}{2}(N-w_{\fontL}N)\psi N}{a(\isfhor)}\wal{-\epss N^2} \right\} \right)|\isfhor| + \left(\min\left\{1,\frac{\frac{1}{2}(N-w_{\fontL}N)\psi N}{a(\isfver)}\wal{-\epss N^2}\right\}\right)|\isfver|.$}\end{equation*} 
Note that from each box $B'$ of height $h (\ge \eps N)$, we can remove the cheapest $\eps h$-horizontal strip and use resource augmentation to get a container based packing with nearly the same profit as $B'$.
Thus by performing a similar analysis to the one done in Subcase (iii a), and using the fact that $a(\isfat)\le N^2-(\frac{\alpha}{2} + (\lambda-\frac{1}{2})\frac{1}{2} + \frac{\beta}{2})N^2 \le N^2 - N^2(\frac{w_{\fontL}}{2} +(\lambda-\frac{1}{2})\frac{1}{2} + \frac{h_{\fontL}}{2})$, we can minimize the whole expression over the domain $\{\frac{w_{\fontL}}{2} + (\lambda-\frac{1}{2})\frac{1}{2}+\frac{h_{\fontL}}{2}\le \frac{2}{5}, \lambda\le \psi, \frac{1}{2}\le w_{\fontL}\le \frac{4}{5}, \frac{1}{2}\le \lambda \le 1, 0\le \wal{h_L} \le \frac{3}{10}\}$ and prove that this solution packs at least \begin{equation}\label{eq:caseiiib} \left(\frac{3-O(\eps_L)}{4}\right)|\ilopt| + |\isthin| + \left(\frac{5}{36} - O(\eps_\fontL)\right) |\isfat|.\end{equation} Thus, using the above inequality along with \eqref{lem2card}, \eqref{lem3card}, \eqref{lem4card} and Lemma \ref{lem4cardgen}, we get a solution with good enough approximation factor. Check Section~\ref{sec:bound_apx} and Table~\ref{tab:summary} for details.

Finally, if $\psi = \lambda > \frac{7+10(w_{\fontL}+\lambda -h_{\fontL})}{40-20w_{\fontL}}$, the bound for the area of horizontal boxes will not be equal to the area of the vertical box constructed to pack rectangles from $isfat$. 
In this case we \wal{change the width of the box inside Region $1$ to be $w_{\fontL}'<N(1-w_{\fontL})$ fixed in such a way that the area of this box is equal to the bound we have for the area of the boxes in Region $2$, i.e., $N^2(1-\lambda-h_{\fontL})-(\frac{13}{20}-\frac{w_{\fontL}}{2}-\frac{h_{\fontL}}{2}-\frac{\lambda}{2}+O(\eps_\fontL)) N^2$.} Performing the same analysis as before, it can be shown that in this case we pack at least \[\left(\frac{(1-\lambda -h_{\fontL})N^2-(\frac{13}{20}-\frac{w_{\fontL}}{2}-\frac{h_{\fontL}}{2}-\frac{\lambda}{2}) N^2}{2a(\isfat)} - O(\eps_{\fontL})N^2\right)|\isfat|,\] which is at least $(\frac{1}{6}-O(\eps_\fontL))|\isfat|$ over the domain $\{\frac{w_{\fontL}}{2} + (\lambda-\frac{1}{2})\frac{1}{2}+\frac{h_{\fontL}}{2}\le \frac{2}{5}, \frac{1}{2} \le w_{\fontL}\le \frac{4}{5}, \frac{7+10(w_L+\lambda-h_L)}{40-20w_L}< \lambda \le 1, 0\le h_L\le \frac{3}{10}\}$ (and this solution leads to a better bound than \eqref{eq:caseiiib}).\\
\begin{figure*}[t!]
    \centering
    \begin{subfigure}[b]{0.35\textwidth}
            \centering
        \includegraphics[width=2.5in]{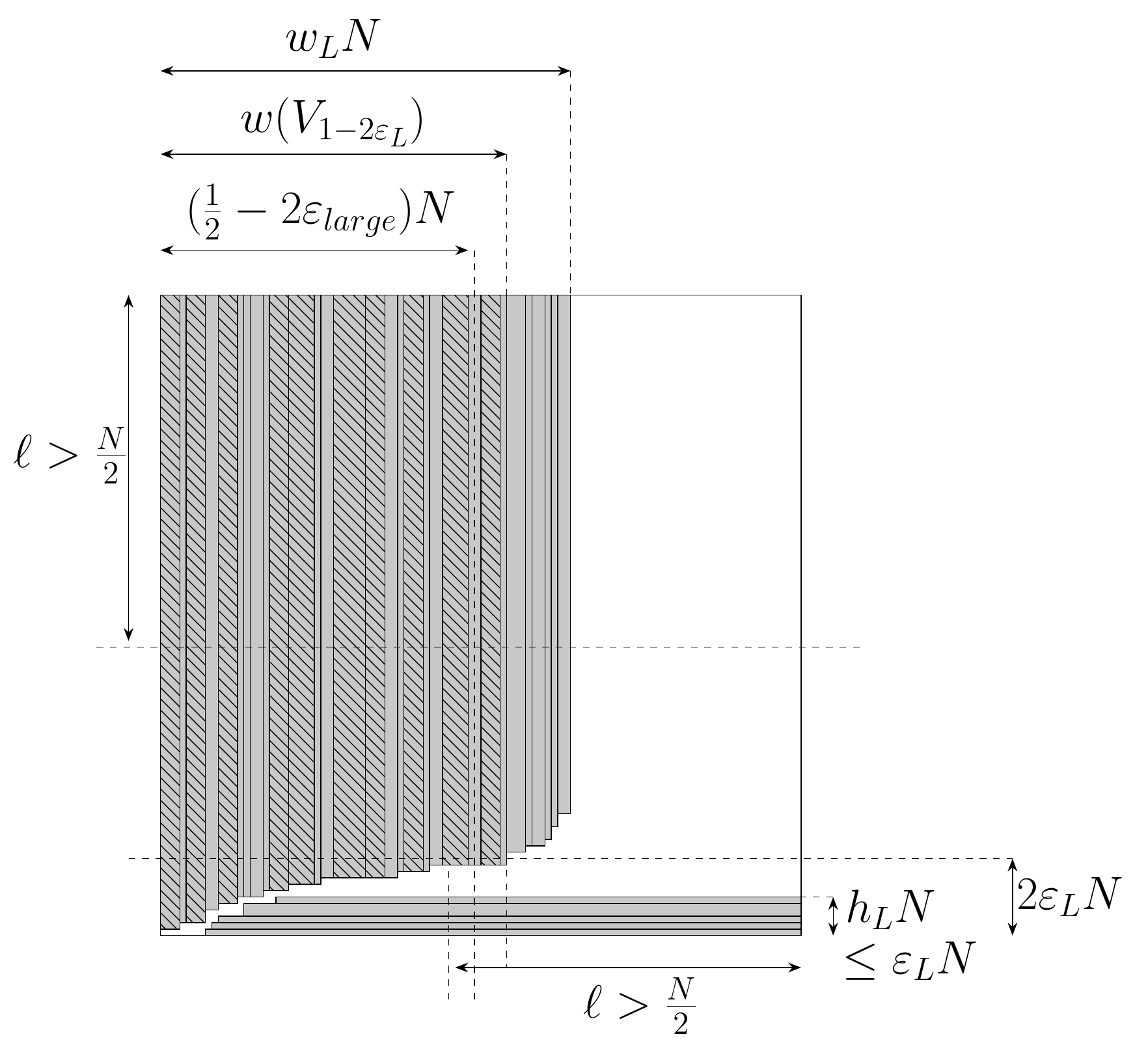}
        \caption{Packing of $\fontL$-region using rectangles from $\ilopt$. Striped rectangles are removed.}
        \label{fig:cardworotcase2b1}
    \end{subfigure}%
    \hspace{35pt}
    \begin{subfigure}[b]{0.4\textwidth}
        \centering
        \includegraphics[width=2.5in]{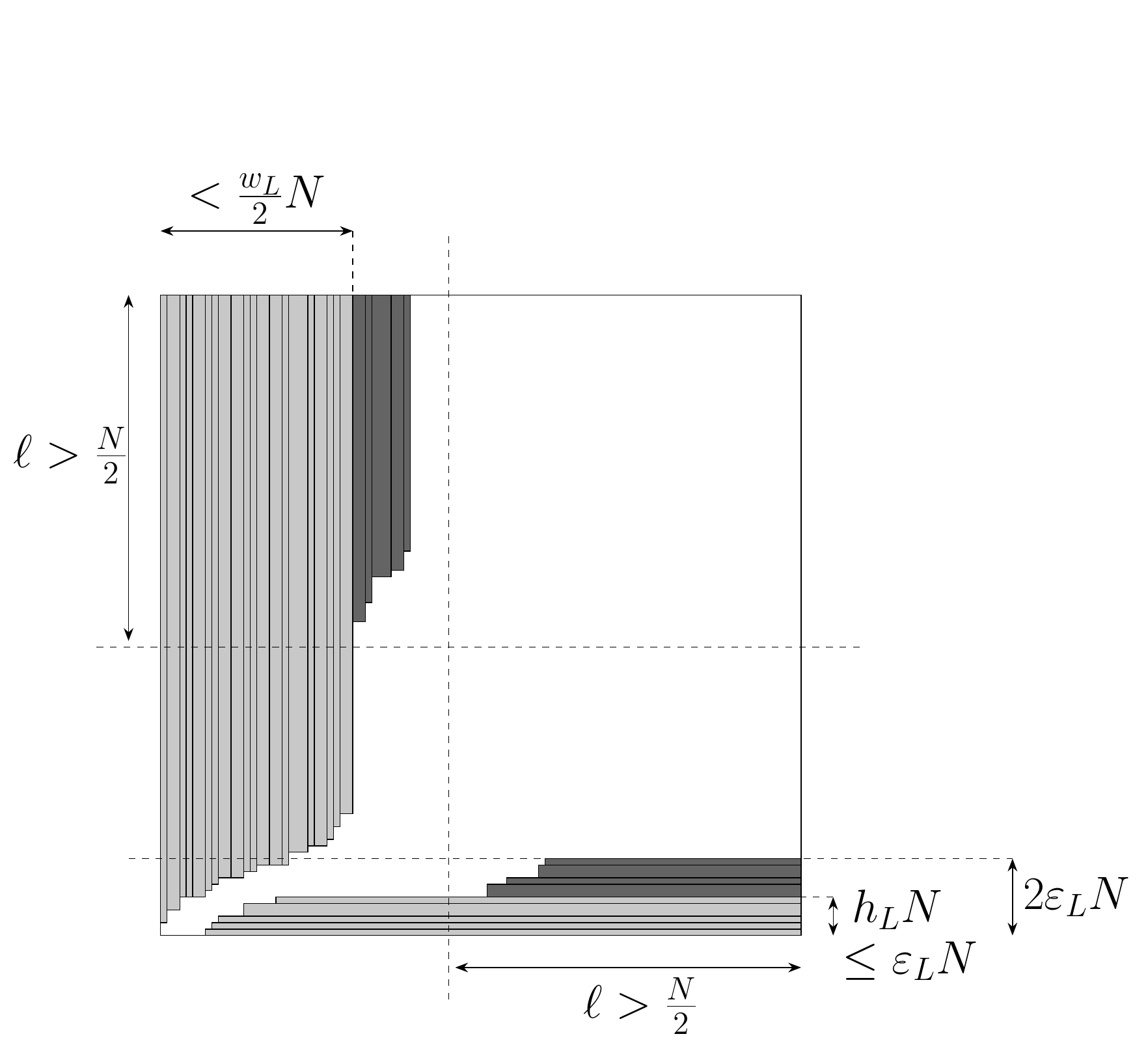}
        \caption{Packing of rectangles in $\ilopt \cup (\isthin)$. Dark gray rectangles are from $\isthin$. }
             \label{fig:cardworotcase2b2}
    \end{subfigure}
    ~
      \caption{The case 2B.}
\end{figure*}
\textit{$\Diamond$ Case 2B. \Big($h_{\fontL} \le  \eps_{\fontL} $ and $w_\fontL N \ge w(V_{1-2\eps_{\fontL}}) > \left(\frac{1}{2}-2\eps_{large}\right)N$\Big) or \Big($w_{\fontL} \le \eps_{\fontL} $ and $ h_\fontL N \ge h(H_{1-2 \eps_{\fontL}}) > \left(\frac{1}{2}-2\eps_{large}\right)N$\Big)}\\
In the first case, area of rectangles in $V_{1-2\eps_{\fontL}} > \left(\frac{1}{2}-2\eps_{large}\right)(1-2\eps_{\fontL})N^2$.
Remaining rectangles in $\ilopt$ have area at least $(w_\fontL -\frac12+2\eps_{large})N\cdot \frac{N}{2}$.
So, $a(\ilopt) > \left(\frac{1}{2}-2\eps_{large}\right)(1-2\eps_{\fontL})N^2 +(w_\fontL -\frac12)N\cdot \frac{N}{2} \ge (\frac14+\frac{w_\fontL}{2}-\eps_{\fontL}-2\eps_{large})N^2$.
Thus $a(\isfat)  \le a(\isopt) < (\frac34 -\frac{w_\fontL}{2} +\eps_{\fontL}+2\eps_{large})N^2$. \\
Now consider the vertical rectangles in the boundary $\fontL$-region sorted non-increasingly by width and pick them iteratively
until their total width crosses $(\frac{w_\fontL}{2}+3\eps_{\fontL}+2\eps_{large})N$.
Remove these rectangles and push the remaining vertical rectangles in the $\fontL$-region to the left as much as possible.
This modified $\fontL$-region will have profit at least $(\frac{1}{2} - O(\eps_{\fontL}))|\ilopt|$.
Now we can put an $\eps N$-strip for the vertical items from $\isthin$ next to the vertical part of $\fontL$-region. On the other hand, the horizontal items of $\isthin$ can be placed on top of the horizontal part of the $\fontL$-region.
The remaining space will be a free rectangular region of height at least $(1-2\eps_{\fontL}) N$ and width $(1-\frac{w_\fontL}{2}+2 \eps_\fontL+2\eps_{large})N$.\walr{total width of vertical in the L-region is $w_\fontL$. Changed from $\alpha_1+\alpha_2$.}
We will use a part of this rectangular region of height $(1-3\eps_{\fontL}) N$ and width $(1-\frac{w_\fontL}{2}+\eps_\fontL)N$ to pack rectangles from $\isfat$ and the rest of the region for resource augmentation.
Since $\frac{w_L}{2}-\eps_L \le \frac{1}{2}-\eps_{large}$, we can use Corollary~\ref{lem:smallStein}\arir{Lemma 3} to pack short rectangles in this region with profit at least $\Big( \frac{(1-\frac{w_\fontL}{2})/2}{\frac34 -\frac{w_\fontL}{2}}-O(\efl)\Big)|\isfat| \ge (\frac34 -O(\eps_{\fontL}))|\isfat|$ as the expression is increasing with respect to $w_L$ and $w_\fontL > \frac12-2\eps_{large}$. Thus, we get, 
\begin{equation}
\label{CardCase2b1}
|\optrc| \ge \left(\frac{1}{2}-O(\eps_{\fontL})\right)|\ilopt|+|\isthin|+\left(\frac34-O(\eps_{\fontL})\right)|\isfat|.
\end{equation}
On the other hand, as $a(\isfat) \le (\frac34  -\frac{w_\fontL}{2} +\eps_{\fontL}+2\eps_{large})N^2$ and $w_\fontL>\frac{1}{2}-2\eps_{large}$, we get $a(\isfat) \le (\frac12+3\eps_{large}+\eps_\fontL) N^2$ and thus from  Lemma \ref{lem4cardgen} we get
\begin{eqnarray}
\label{CardCase2b1a}
|\optrc| & \ge & \frac{3}{4}|\ilthin|+|\isthin| + (1-O(\eps_{\fontL}))|\isfat| \nonumber\\ & \ge & \frac{3}{4}|\ilthin|+(1-O(\eps_{\fontL}))|\isopt|.
\end{eqnarray}
From inequalities \eqref{lem1card}, \eqref{lem3card}, \eqref{lem4card},
\eqref{CardCase2b1} and \eqref{CardCase2b1a} we get a solution with good enough approximation factor. Check Section~\ref{sec:bound_apx} and Table~\ref{tab:summary} for details.

Now we consider the last case when $w_{\fontL} \le \eps_{\fontL} $ and $h_\fontL N \ge h(H_{1-2 \eps_{\fontL}}) > \left(\frac{1}{2}-2\eps_{large}\right)N$.  
\arir{Added back the description. Actually the argument is not totally analogous. We remove a different subring not the cheapest one.}
Note that as we assumed the cheapest subring was the top subring, after \wal{removing it} we might be left with only $|\ilopt \cap \Rho|/2$ profit in the horizontal part of $\fontL$-region. So, further removal of items from the horizontal part might not give us a good solution.
Thus we show an alternate good packing.  We restart with the ring packing and delete the cheapest vertical subring instead of the cheapest subring (i.e., the top subring) and create a new boundary $\fontL$-region.
Here, consider the horizontal rectangles in the boundary $\fontL$-region in non-increasing order of height and take them
until their total height crosses $(\frac{\beta_{bottom}+\beta_{top}}{2}+\eps_{small}+\eps)N$. Remove these rectangles and push the remaining horizontal rectangles to the bottom as much as possible.
Then, following similar arguments as before, we will obtain the same bounds for the constructed solution.

\subsection{Bounding the approximation factor}\label{sec:bound_apx} In each one of the cases listed before we are developing a set of different solutions in order to achieve a good approximation factor. Let $z=|\optrc|/|\opt|$, $x_1=|\ilthin|/|\opt|$, $x_2=|\ilfat|/|\opt|$, $x_3=|\isthin|/|\opt|$ and $x_4=|\isfat|/|\opt|$. The following list enumerates all the obtained inequalities in this section, and it is worth remarking that not all of them hold simultaneously.
\begin{enumerate} \item $z \ge \frac{3}{4}x_1+\frac{3}{4}x_2+x_3+\left(\frac{1}{2}-O(\eps_{L})\right)x_4$;
	\item $z\ge \frac{3}{4}x_1 + \frac{3}{4}x_2$;
	\item $z\ge \left(\frac{1}{2}-O(\eps_{L})\right)(x_1+x_2) + \left(\frac{3}{4}-O(\eps_{L})\right)(x_3+x_4)$;
	\item $z\ge (1-O(\eps_{L}))(x_2+x_4)$;
	\item $z\ge (1-O(\eps_{L}))\left(\frac{1}{2}x_1 + x_2 + \frac{1}{2}x_4\right)$;
	\item $z\ge \left(\frac{3}{4}-O(\eps_{L})\right)(x_1+x_2) + x_3$;
	\item $z\ge \frac{3}{4}x_1 + x_3 + \left(\frac{5}{6}-O(\eps_{L})\right)x_4$;
	\item $z\ge \frac{3}{4}(x_1+x_2) + \left(\frac{5}{12}-O(\eps_L)\right)(x_3+x_4)$;
	\item $z\ge\frac{3}{4}x_1 +x_3 + \left(\frac{2}{3}-O(\eps_{L})\right)x_4$;
	\item $z\ge\frac{3}{4}(x_1+x_2)+x_3+\left(\frac{7}{48}-O(\eps_L)\right)x_4$;
	\item $z\ge \left(\frac{3}{4}-O(\eps_{L})\right)(x_1+x_2) + x_3 + \left(\frac{5}{36}-O(\eps_{L})\right)x_4$;
	\item $z\ge\left(\frac{1}{2}-O(\eps_{L})\right)(x_1+x_2)+x_3+\left(\frac{3}{4}-O(\eps_{L})\right)x_4$;
	\item $z\ge\frac{3}{4}x_1 + (1-O(\eps_L))(x_3+x_4)$.
\end{enumerate}

For each case $i$, let $\mathcal{A}_i$ be the set of indexes of valid inequalities for case $i$. Then we can write the following linear program to compute the obtained approximation factor in that case: $$\begin{array}{cl} \min & z \\ s.t. & \text{Inequalities indexed by } \mathcal{A}_i \\ & \displaystyle\sum_{i =1}^{4}{x_i} = 1 \\ & z, x_i \ge 0 \qquad \text{for } i=1,2,3,4. \end{array}$$ Let $c_{j,k}$ be the coefficient accompanying $x_k$ in the constraint $j\in \mathcal{A}_i$, $k=1,2,3,4$. The dual of the program for case $i$ has the form $$\begin{array}{cl} \max & -w \\ s.t. & \displaystyle\sum_{j\in \mathcal{A}_i}{y_j} \le 1 \\ & \displaystyle\sum_{j\in\mathcal{A}_i}{c_{j,k}y_j} + w \ge 0 \qquad \text{for } k=1,2,3,4 \\ & y_j \ge 0 \qquad \text{for } j\in \mathcal{A}_i \\ & w\in \mathbb{R} \end{array}$$ Any feasible solution for the dual program of case $i$ is a lower bound on the fraction of $OPT$ packed in that case. Table~\ref{tab:summary} summarizes the analysis described along this section for all the cases, stating the valid inequalities and the approximation factor obtained in each one of them, together with a dual feasible solution. It is not difficult to see that the worst case is 2A(iii)b, implying that $|\optrc| \ge (\frac{325}{558}-O(\eps_{L}))|OPT|$. Applying Lemma~\ref{lem:containersPackPTAS} concludes the proof of Theorem~\ref{thm:cardWorot}. \end{proof}
\begin{table}[h]
	\begin{tabular}{ | c | c | c | c |}
		\hline
		&  &  &  \\[-11.5pt] Case & Valid inequalities & Dual feasible solution & \parbox[c]{2.6cm}{\centering Fraction of OPT\\packed $(w)$} \\ \hline
		 &  &  &  \\[-11.5pt] $1$ & $1, 3, 4, 5$ & $y_1=\frac{1}{2}, y_3=\frac{1}{2}, y_4=0, y_5=0$ & $\frac{5}{8}-O(\eps_{L})$ \\ & & & \\[-11.5pt] \hline
		&  &  &  \\[-11.5pt] $2A(i)$ & $3,4,5,6,7$ & $y_3=\frac{17}{54}, y_4=0, y_5=\frac{1}{3}, y_6=\frac{7}{54}, y_7=\frac{2}{9}$ & $\frac{127}{216}-O(\eps_{L})$ \\ & & & \\[-11.5pt] \hline
		&  &  &  \\[-11.5pt] $2A(ii)$ & $3,4,5,6,8$ & $y_3=\frac{4}{7}, y_4=0, y_5=0, y_6=0, y_8=\frac{3}{7}$ & $\frac{17}{28}-O(\eps_{L})$ \\ & & & \\[-11.5pt] \hline
		&  &  &  \\[-11.5pt] $2A(iii)a$ & $3,4,5,9,10$ & $y_3=\frac{124}{369}, y_4=0, y_5=\frac{1}{3}, y_9=\frac{2}{9}, y_{10}=\frac{40}{369}$ & $\frac{215}{369} - O(\eps_L)$ \\ & & & \\[-11.5pt] \hline
		&  &  &  \\[-11.5pt] $2A(iii)b$ & $3,4,5,9,11$ & $y_3=\frac{94}{279}, y_4=0, y_5=\frac{1}{3}, y_9=\frac{2}{9}, y_{11}=\frac{10}{93}$ & $\frac{325}{558} - O(\eps_{L})$ \\ & & & \\[-11.5pt] \hline
		&  &  &  \\[-11.5pt] $2B$ & $2,4,5,12,13$ & $y_2=\frac{8}{41}, y_4=0, y_5=\frac{9}{41}, y_{12}=\frac{18}{41}, y_{13}=\frac{6}{41}$ & $\frac{24}{41} - O(\eps_{L})$ \\[1.5pt] \hline
	\end{tabular}
	\caption{Summary of the case analysis in Theorem~\ref{thm:cardWorot}.}
	\label{tab:summary}
\end{table}

%% file: 2Dknapsack8f-With-rotations.tex
\section{Cardinality Case With Rotations}
\label{sec:cardRot}
In this section we present a polynomial time $(4/3+\eps)$-approximation algorithm for \tdkr~for the cardinality case. We next assume w.l.o.g. that $\eps>0$ is a sufficiently small constant and also that $\height(i) \ge \width(i)$ for all items $i$ in the input.

Consider some optimal solution $OPT$ for \tdkr, with an associated packing in the knapsack. We crucially exploit the following resource contraction lemma, which is our main new idea in the rotation case.

\begin{lem}\label{uwrescontr1}
(Resource Contraction Lemma)
For given positive constants $\eps \le 1/13$ and $\epss < \eps^{\frac{1}{2\eps}+1}$, suppose that there exists a feasible packing of a set of items $M$, with $|M|\geq 1/\epss^3$. Then it is possible to pack a subset $M'\subseteq M$ of  cardinality at least $\frac23(1-O(\eps))|M|$ into $[0,\left(1-\eps^{\frac{1}{2\eps}+1}\right)N]\times [0,N]$ if rotations are allowed.
\end{lem}

We defer the proof of the lemma to the end of this section.

As in the case without rotations, we will first show the existence of a container packing that packs at least $(3/4-O(\eps))|OPT|$ many items.
Let $APX$ be the container packing with largest possible profit. 
As in Section \ref{sec:weighted}, we assume all items to be skewed. Note that the small items can be handled with the techniques used in Lemma \ref{lem:smallPack}.
We start with the corridor partition as in Section \ref{sec:weighted} and define {\em thin}, {\em fat} and {\em killed}
rectangles accordingly.
Let $T$ and $F$ be the set of thin and  fat rectangles respectively.

We will show that 
$|APX|\ge (3/4-O(\eps))|OPT|$.
\begin{lem}
\label{uwrot1}
$|APX| \ge (1-\eps)|F|$.
\end{lem}
\begin{proof}
After removal of $T$, we can get a container packing for almost all items in $F$ as discussed in Lemma \ref{lem:onlyFat} in Section \ref{sec:weighted}.
\end{proof}

Now we show that using Lemma~\ref{uwrescontr1} we can prove the following :

\begin{lem}
\label{uwrot2}
$|APX| \ge (1-O(\eps)) (|T|+\frac23|F|)$.
\end{lem}
\begin{proof}
Thanks to \ref{lem:boxProperties} we can ensure that the total height of rectangles in $T$ is at most $\frac{\eps^{\frac{1}{2\eps}+1} N}{2}$. 
So we can pack them in a vertical container of width $\frac{\eps^{\frac{1}{2\eps}+1} N}{2}$.

Now if  $|F| \ge \frac{1}{\epss^3}$, then  by Lemma \ref{uwrescontr1} there exists $F' \subseteq F$ of cardinality at least $\frac23(1-O(\eps))|F|$ that can be packed inside $K':=[0,\left(1-\eps^{\frac{1}{2\eps}+1}\right)N]\times [0,N]$. 
 Then we can use the resource augmentation PTAS in \cite{jansen2007new} to get a container packing of  $(2/3-O(\eps)) |F|$ many rectangles in the area $K'':=[0,\left(1-\eps^{\frac{1}{2\eps}+1}/2\right)N]\times [0,N]$, and  pack the vertical container for items in $T$ to the right of $K''$  in the area $[\left(1-\eps^{\frac{1}{2\eps}+1}/2\right)N, 1]\times [0,N]$.
 
 Otherwise,  if  $|F| < \frac{1}{\epss^3}$, there are two cases.
 If $|T| < \frac{1}{\epss^4}$, then   $|F \cup T| < \frac{2}{\epss^4}$ and we can find the packing just by brute-force.
 Otherwise if, $|T| \ge \frac{1}{\epss^4} \ge |F|/\epss$, then $APX \ge |T| \ge (1-O(\eps))(|T|+|F|)$.
\end{proof}

Thus we get the following theorem:
\begin{thm}
$|APX| \ge (3/4-O(\eps))|OPT|$.
\end{thm}
\begin{proof}
The claim follows by combining Lemma \ref{uwrot2} and \ref{uwrot1}. Up to a factor $(1-O(\eps))$, the worst case is obtained when $|F|=|T|+2/3 \cdot  |F|$,
i.e., $|F|=3 |T|$. This gives a total profit of $3/4\cdot |T \cup F|$. 
\end{proof}


It remains to prove Lemma \ref{uwrescontr1}. Let us remove from $M$ all items that are larger than $\epss N$ in both dimensions. Let $M_2$ be the resulting set: observe that $|M_2|\geq (1-\epss)|M|$.

We next show how to remove from $M_2$ a set of cardinality at most $\eps |M_2|$ such that the remaining items $M_3$
are either \emph{very tall} or \emph{not too tall}, where the exact meaning of this will be given next. We use the notation $[k]=\{1,\ldots,k\}$ for a positive integer $k$.
\begin{lem} 
\label{lem:medWeight}
Given any constant $1/2>\eps>0$, there exists a value $i \in [\lceil1/(2\eps)\rceil]$ such that all items in $M_2$ having height in $((1-2\eps^i)N, (1-\eps^{i+1})N]$  have total cardinality at most $\eps |M_2|$.
\end{lem}
\sanr{Formally, $[1/(2\eps)]$ is not defined, as we need an integer here. So I replaced it with the rounded up value.}
\begin{proof}
Let $K_i$ be the set of items in $M_2$ with height  in $((1-2\eps^i)N, (1-\eps^{i+1})N]$
for $i \in [\fab{\lceil1/(2\eps)\rceil}]$.
An item can belong to at most two such sets as $\eps < 1/2$. 
Thus, the smallest such set has cardinality  at most $\eps |M_2|$.
\end{proof}

We remove from $M_2$ the elements from the set $K_i$ of minimum cardinality guaranteed by the above lemma, and let $M_3$ be the resulting set. We also define $\eps_s=\eps^i$  for the same $i$. Thus, $\eps_s \geq \eps^{1/2\eps} > \epss/\eps$.
Note that the items in $M_3$ have height either at most $(1-2\eps_s)N$ or above $(1-\eps\cdot \eps_s)N$.

\begin{figure}[h]
\centering
\begin{subfigure}[b]{.34\textwidth}
\resizebox{!}{135pt}{
	\begin{tikzpicture}
		
		
		\draw[thick] (0,0) rectangle (8,8);
		
		\fill[pattern = north east lines, pattern color = lightgray] (0,0) rectangle (1.5,8);
		\fill[pattern = north east lines, pattern color = lightgray] (6.5,0) rectangle (8,8);
		\fill[pattern = north west lines, pattern color = lightgray] (0,0) rectangle (8,1.5);
		\fill[pattern = north west lines, pattern color = lightgray] (0,6.5) rectangle (8,8);
		
		
		\draw[dashed] (-1,1.5) -- (9,1.5);
		\draw[dashed] (-1,6.5) -- (9,6.5);
		\draw[dashed] (1.5,-1) -- (1.5,9);
		\draw[dashed] (6.5,-1) -- (6.5,9);
		
		
		\draw [decorate,decoration={brace,amplitude=7pt}] (8,8) -- (8,6.5); 
		\draw (8.2,7.25) node [anchor = west] {\Large $\gamma N$};
		\draw [decorate,decoration={brace,amplitude=6pt}] (8,1.5) -- (8,0); 
		\draw (8.2,0.75) node [anchor = west] {\Large $\delta N$};
		\draw [decorate,decoration={brace,amplitude=7pt}] (6.5,8) -- (8,8); 
		\draw (7.25,8.2) node [anchor = south] {\Large $\beta N$};
		\draw [decorate,decoration={brace,amplitude=7pt}] (0,8) -- (1.5,8); 
		\draw (0.75,8.2) node [anchor = south] {\Large $\alpha N$};
		
		\draw (-0.6,1.5) node[anchor=north] {\Large $S_{B,\delta}$};
		\draw (1.5,-0.6) node[anchor=east] {\Large $S_{L,\alpha}$};
		\draw (6.5,-0.6) node[anchor=west] {\Large $S_{R,\beta}$};
		\draw (-0.6,6.5) node[anchor=south] {\Large $S_{T,\gamma}$};
		\end{tikzpicture}}
		\caption{Strips $S_{L,\alpha}, S_{R,\beta}, S_{B,\delta}, S_{T,\gamma}$}
	\end{subfigure}
	\hspace{35pt}
	\begin{subfigure}[b]{.4\textwidth}
	\resizebox{!}{135pt}{
		\begin{tikzpicture}
			
			
			\draw[thick] (0,0) rectangle (8,8);
			
			
			\draw[fill=lightgray] (0.5,6.5)  rectangle (5,7.2);
			\draw[fill=lightgray] (1.5,2)  rectangle (2.4,4);
			\draw[fill=darkgray] (0.5,1.5)  rectangle (1.3,5);
			\draw[fill=gray] (5,1)  rectangle (7,2.5);
			
			
			\draw[dashed] (2,-1) -- (2,9);
			
			
			\draw [decorate,decoration={brace,amplitude=7pt}] (0,8) -- (2,8); 
			\draw (0.75,8.2) node [anchor = south] {\Large $\alpha N$};
			
			\draw (2,-0.6) node[anchor=east] {\Large $S_{L,\alpha}$};
			\filldraw[color=white] (-1,0) rectangle (-0.9,0.1);
			\end{tikzpicture}}
			\caption{ $C_{L,\alpha}, D_{L,\alpha}$ are dark and light gray resp.}
		\end{subfigure}
	\caption{Definitions for cardinality 2DK with rotations.}\label{fig:unweighted_rotations_def}
	\end{figure}
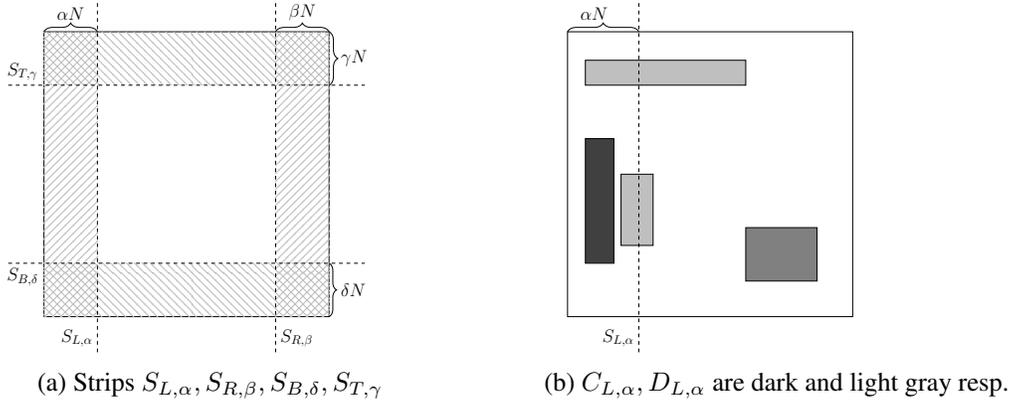

\andy{For any $\delta>0$} denote the strips of width $N$ and height $\delta N$ at the top and bottom of the knapsack by $S_{T,\delta}:=[0,N]\times[(1-\delta)N,N]$ and 
$S_{B,\delta} :=[0,N]\times[0, \delta N]$, resp. Similarly, denote the strips of height $N$ and width $\delta N$ to the left and right of the knapsack
by $S_{L,\delta}:=[0, \delta N] \times [0,N]$ and  $S_{R,\delta} :=[(1-\delta)N, N]\times [0,N]$, resp. 
The set of items in $M_3$ intersected by and fully contained in strip $S_{K, \delta}$ are denoted by $E_{K, \delta}$ and $C_{K, \delta}$, resp. Obviously $C_{K, \delta} \subseteq E_{K, \delta}$, and we define, $D_{K, \delta} = E_{K, \delta} \setminus C_{K, \delta}$ . 
Let $a(I)$ denote the total area of items in $I$, i.e., $a(I)=\sum_{i\in I}\width(i)\cdot \height(i)$.
\begin{lem}
\label{lem:stripHalf}
Either $a(E_{L, \eps_s} \cup E_{R, \eps_s}) \le \frac{(1+8 \eps_s)}{2} N^2$
or $a(E_{T, \eps_s} \cup E_{B, \eps_s}) \le  \frac{(1+8 \eps_s)}{2} N^2$.
\end{lem}
\begin{proof}
\andyr{Intuitively, I would define $V := E_{L, \eps_s} \cup E_{R, \eps_s}$ since these items are in the vertical strips.}\walr{Changed.}
Let us define $V := E_{L, \eps_s} \cup E_{R, \eps_s}$ 
and $H := E_{T, \eps_s} \cup E_{B, \eps_s}$.
Note that,
$a(V) + a(H)
=
a(V \cup H)+
a(V \cap H)$. 
Clearly $a(V \cup H) \le N^2$ since all items fit into the knapsack.
On the other hand, except possibly four items (the ones that contain at least one of the points $(\eps_s N,\eps_s N),((1-\eps_s) N,\eps_s N),(\eps_s N,(1-\eps_s)N),((1-\eps_s)N, (1-\eps_s)N)$) all other items in $V \cap H$ lie entirely within 
the four strips $S_{L,\eps_s}\cup S_{R,\eps_s}\cup S_{T,\eps_s}\cup S_{B,\eps_s}$.
Thus $a(V \cap H) \le 4\eps_s N^2 +4\eps_{small} N^2 \le 
8\eps_s N^2$, as $\eps_{small} \le \eps_s$. 
We can conclude that $\min\{ a(V),a(H) \} \le \frac{a(V)+a(H)}{2} = \frac{a(V \cup H)+
a(V \cap H)}{2} \le \frac{(1+8 \eps_s)}{2} N^2$.
\end{proof}


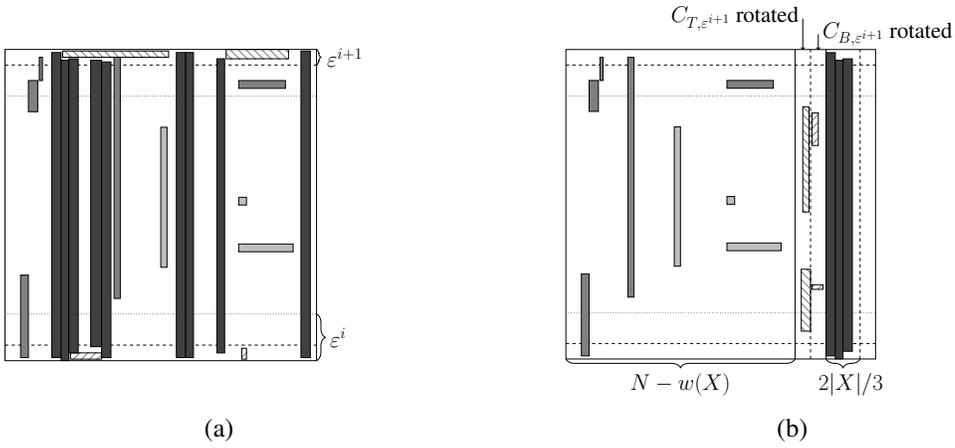
\begin{figure}
\centering
\captionsetup[subfigure]{justification=centering}
	\begin{subfigure}[b]{.35\textwidth}
	\resizebox{!}{150pt}{
\begin{tikzpicture}
\draw (0,0) -- (10,0) -- (10,10)-- (0,10) -- (0,0);
\draw[dashed]  (0,0.5) -- (10,0.5);
\draw[dashed]  (0,9.5) -- (10, 9.5);
\draw[dotted]  (0,1.5) -- (10,1.5);
\draw[dotted]  (0,8.5) -- (10, 8.5);
\filldraw[fill=darkgray, draw=black] (1.5,0.1) rectangle (1.8, 9.9);
\filldraw[fill=darkgray, draw=black] (1.8,0.005) rectangle (2.05, 9.65);
\filldraw[fill=darkgray, draw=black] (2.05,0.25) rectangle (2.35, 9.7);
\filldraw[fill=darkgray, draw=black] (2.75, 0.45) rectangle (3.1, 9.65);
\filldraw[fill=darkgray, draw=black] (3.1,0.1) rectangle (3.4, 9.6);
\filldraw[fill=darkgray, draw=black] (5.5,0.1) rectangle (5.8, 9.9);
\filldraw[fill=darkgray, draw=black] (5.8,0.1) rectangle (6.05, 9.9);
\filldraw[fill=darkgray, draw=black] (6.8,0.25) rectangle (7.05, 9.7);
\filldraw[fill=darkgray, draw=black] (9.5,0.1) rectangle (9.8, 9.95);
\draw[pattern=north west lines, pattern color=gray] (1.85, 9.75) rectangle (5.25, 9.95);
\draw[pattern=north west lines, pattern color=gray] (7.1, 9.7) rectangle (9.1, 10);
\draw[pattern=north east lines, pattern color=gray] (2.1, 0.05) rectangle (3.1, 0.25);
\draw[pattern=north east lines, pattern color=gray] (7.6, 0.05) rectangle (7.75, 0.4);
\filldraw[fill=gray, draw=black] (0.5,0.1) rectangle (0.75, 2.75);
\filldraw[fill=gray, draw=black] (0.75,8) rectangle (1.05, 9);
\filldraw[fill=gray, draw=black] (1.1,9) rectangle (1.2, 9.75);
\filldraw[fill=gray, draw=black] (3.5,2) rectangle (3.7, 9.75);
\filldraw[fill=gray, draw=black] (7.5, 8.75) rectangle (9, 9);
\filldraw[fill=lightgray, draw=black] (5,3) rectangle (5.2, 7.5);
\filldraw[fill=lightgray, draw=black] (7.5,3.5) rectangle (9.25, 3.75);
\filldraw[fill=lightgray, draw=black] (7.5,5) rectangle (7.75, 5.25);

\draw [decorate,decoration={brace,amplitude=4pt}, thick] (10,10) -- (10,9.5); 
\draw (10.2,9.75) node [anchor = west] {\huge $\eps^{i+1}$}; 
\draw [decorate,decoration={brace,amplitude=8pt},thick] (10,1.5) -- (10,0); 
\draw (10.3,0.75) node [anchor = west] {\huge $\eps^{i}$};

\fill[color=white] (0,-1.23) rectangle (1,-0.5);
\fill[color=white] (0,11) rectangle (1,11.5);
\end{tikzpicture}}
\caption{\label{f:uwrot1} }
\end{subfigure}%
\hspace{45pt}
	\begin{subfigure}[b]{.37\textwidth}
\resizebox{!}{150pt}{
\begin{tikzpicture}
\draw (0,0) -- (10,0) -- (10,10)-- (0,10) -- (0,0);
\draw (7.4,0) -- (7.4,10);
\draw[dashed] (7.9,0) -- (7.9,10);
\draw[dashed] (8.4,0) -- (8.4,10);
\draw[dashed]  (0,0.5) -- (10,0.5);
\draw[dashed]  (0,9.5) -- (10, 9.5);
\draw[dotted]  (0,1.5) -- (10,1.5);
\draw[dotted]  (0,8.5) -- (10, 8.5);
\filldraw[fill=gray, draw=black] (0.5,0.1) rectangle (0.75, 2.75);
\filldraw[fill=gray, draw=black] (0.75,8) rectangle (1.05, 9);
\filldraw[fill=gray, draw=black] (1.1,9) rectangle (1.2, 9.75);
\filldraw[fill=gray, draw=black] (2,2) rectangle (2.2, 9.75);
\filldraw[fill=gray, draw=black] (5.2, 8.75) rectangle (6.7, 9);
\filldraw[fill=lightgray, draw=black] (3.5,3) rectangle (3.7, 7.5);
\filldraw[fill=lightgray, draw=black] (5.2,3.5) rectangle (6.95, 3.75);
\filldraw[fill=lightgray, draw=black] (5.2,5) rectangle (5.45, 5.25);

\draw[pattern=north west lines, pattern color=gray] (7.65, 8.15) rectangle (7.85, 4.75);
\draw[pattern=north west lines, pattern color=gray] (7.6, 2.9) rectangle (7.9, 0.9);
\draw[pattern=north east lines, pattern color=gray] (7.95, 7.95) rectangle (8.15, 6.9);
\draw[pattern=north east lines, pattern color=gray] (7.95,2.4) rectangle (8.3, 2.25);

\filldraw[fill=darkgray, draw=black] (8.4,0.1) rectangle (8.7, 9.9);
\filldraw[fill=darkgray, draw=black] (8.7,0.005) rectangle (8.95, 9.65);
\filldraw[fill=darkgray, draw=black] (8.95,0.25) rectangle (9.25, 9.7);
\draw[dashed]  (9.5,0) -- (9.5, 10);

\draw [decorate,decoration={brace,amplitude=8pt}] (7.4,0) -- (0,0); 
\draw (3.7,-0.3) node [anchor = north] {\huge $N-w(X)$}; 
\draw [decorate,decoration={brace,amplitude=8pt}] (9.5,0) -- (8.4,0); 
\draw (9.2,-0.3) node [anchor = north] {\huge $2|X|/3$};
\draw[->] (7.65,11) -- (7.65,10); 
\draw (7.65,11) node [anchor=east] {\huge $C_{T,\eps^{i+1}}$ rotated};
\draw[->] (8.15,10.5) -- (8.15,10); 
\draw (8.15,10.5) node [anchor=west] {\huge $C_{B,\eps^{i+1}}$ rotated};

\end{tikzpicture}}
\caption{\label{f:uwrot4} 
}
\end{subfigure}
	\caption{Case A for cardinality 2DK with rotations. Dark gray rectangles are $X$, light gray rectangles are $Z$, gray (and hatched) rectangles are $Y$, hatched rectangles are $C_{T,\eps^{i+1}}$ and  $C_{B,\eps^{i+1}}$. Figure (a): original packing in $N \times N$, Figure (b): modified packing leaving space for resource contraction on the right.}
	\label{f:uwrotA}
\end{figure}

Now we use Steinberg's Theorem \cite{steinberg1997strip} to prove the following lemma.

\begin{lem}
\label{lem:Stein}
Let $0<\eps_a<1/2$ be a constant and $\tilde{M}:=\{1, \ldots, k\}$ a set of items satisfying 
$\width(i) \le \epss N$ for all $i \in \tilde{M}$. If $a(\tilde{M}) \le (1/2+\eps_a)N^2$. Then, 
a subset of $\tilde{M}$ of cardinality at least $(1-2\eps_s-2\eps_a)|\tilde{M}|$ can be packed into $[0,(1-\eps_s)N]\times [0,N]$.
\end{lem}
\begin{proof}
W.l.o.g., assume the items in $\tilde{M}$ are given in nondecreasing order according to their area.
Note that $a(i) \le \epss N^2  \le \san{\frac{\eps_s}{2}} N^2$ for any $i \in \tilde{M}$.
Let $S:=\{1, \ldots, j\}$ be such that $\frac{(1-2\eps_s)}{2}N^2 \le \sum_{i=1}^j a(i)  \le \frac{(1-\eps_s)}{2}N^2$ and 
\andyr{I believe that for this you would need that $a(i) \le  \eps_s/2 N^2$ for each $i \in \tilde{M}$}\sanr{You are right. We get this from $\epss < \eps \eps_s$ and $\eps<1/2$. I adjusted the bound in the proof anyways.}
$\sum_{i=1}^{j +1} a(i) > \frac{(1-\eps_s)}{2}N^2$.
Then from Theorem \ref{thm:steinberg}, $S$ can be packed into $[0,(1-\eps_s)N]\times [0,N]$.
As we considered items in the order of  nondecreasing area, $\frac{|S|}{|\tilde{M}|} \ge \frac{\left(\frac12-\eps_s\right)}{\left(\frac12+\eps_a\right)}$.
Thus, $ |S| \ge 
\left(1- \frac{(\eps_a + \eps_s)}{\left(\frac12+\eps_a\right)}\right)|\tilde{M}|
> (1-2\eps_a -2\eps_s)|\tilde{M}|$. 
\end{proof}

From Lemma \ref{lem:stripHalf}, we can assume w.l.o.g.
that $a(E_{T, \eps_s} \cup E_{B, \eps_s}) \le  \frac{(1+8 \eps_s)}{2}N^2$.
Let $X$ be the set of items in $M_3$ that intersect both $S_{T, \eps_s}$ and $S_{B, \eps_s}$  and $Y := \{ E_{T, \eps_s} \cup E_{B, \eps_s} \} \setminus X$. 
Define $Z:=M_3 \setminus \{ X \cup Y\}$ to be the rest of the items.
Let us define $w(X)=\sum_{i \in X} \width(i)$. Now there are two cases.

\noindent \textbf{Case A.} $w(X) \ge 12\eps \cdot \eps_s N$.
From Lemma \ref{lem:medWeight}, all items in $X$
intersect both $S_{T, \eps\cdot \eps_s}$ and $S_{B, \eps\cdot \eps_s}$. 
So the removal of $X\cup C_{T, \eps\cdot \eps_s}\cup C_{B, \eps\cdot \eps_s}$ creates a set of empty strips of height $N$ and total width of $w(X)$.
By a simple permutation argument, all items in $Y \cup Z$ can be packed inside $[0,N-w(X)]\times [0,N]$, leaving an empty vertical strip of width $w(X)$ on the right side of the knapsack.
Next we rotate $C_{T, \eps\cdot \eps_s}$ and $C_{B, \eps\cdot \eps_s}$ and pack them in two vertical strips, each
of width $\eps\cdot \eps_sN$. Note that $\width(i) \le \eps\cdot \eps_sN$ for all $i \in X$.
Now take items in $X$ by nondecreasing width, till their total width is in $[w(X)-4\eps\cdot \eps_sN, w(X)-3\eps\cdot \eps_sN]$ and pack them into another vertical strip. The cardinality of this set is at least $\frac{(w(X)- 4\eps\cdot \eps_sN)}{w(X)}|X|
\ge \frac{2}{3}|X|$, where the last inequality follow by the Case A assumption. 
Hence, at least $\frac23|X|+|Y|+|Z|  \ge \frac23(|X|+|Y|+|Z|) $ items can be packed into $[0,(1-\eps\cdot \eps_s)N]\times [0,N]$.

\begin{figure}
\captionsetup[subfigure]{justification=centering}
\centering
\hspace{-20pt}
	\begin{subfigure}[b]{.35\textwidth}
	\resizebox{!}{120pt}{
\begin{tikzpicture}
\draw (0,0) -- (10,0) -- (10,10)-- (0,10) -- (0,0);
\draw[dashed]  (0,0.5) -- (10,0.5);
\draw[dashed]  (0,9.5) -- (10, 9.5);
\draw[dotted]  (0,1.5) -- (10,1.5);
\draw[dotted]  (0,8.5) -- (10, 8.5);
\filldraw[fill=darkgray, draw=black] (2.05,0.25) rectangle (2.35, 9.7);
\filldraw[fill=darkgray, draw=black] (3.1,0.1) rectangle (3.4, 9.6);
\filldraw[fill=darkgray, draw=black] (6.8,0.25) rectangle (7.05, 9.7);
\draw[pattern=north west lines, pattern color=gray] (7.1, 9.7) rectangle (9.1, 10);
\draw[pattern=north east lines, pattern color=gray] (2.1, 0.05) rectangle (3.1, 0.25);
\draw[pattern=north east lines, pattern color=gray] (7.6, 0.05) rectangle (7.75, 0.4);
\filldraw[fill=gray, draw=black] (0.5,0.1) rectangle (0.75, 2.75);
\filldraw[fill=gray, draw=black] (0.75,8) rectangle (1.05, 9);
\filldraw[fill=gray, draw=black] (1.1,9) rectangle (1.2, 9.75);
\filldraw[fill=lightgray, draw=black] (5,3) rectangle (5.2, 7.5);
\filldraw[fill=lightgray, draw=black] (7.5,3.5) rectangle (9.25, 3.75);
\filldraw[fill=lightgray, draw=black] (7.5,5) rectangle (7.75, 5.25);

\draw [decorate,decoration={brace,amplitude=4pt}, thick] (10,10) -- (10,9.5); 
\draw (10.2,9.75) node [anchor = west] {\huge $\eps^{i+1}$}; 
\draw [decorate,decoration={brace,amplitude=8pt}, thick] (10,1.5) -- (10,0); 
\draw (10.3,0.75) node [anchor = west] {\huge $\eps^{i}$}; 
\end{tikzpicture}}
\caption{\label{f:uwrot3} 
}
\end{subfigure}%
\hspace{50pt}
	\begin{subfigure}[b]{.35\textwidth}
\resizebox{!}{120pt}{
\begin{tikzpicture}
\draw (0,0) -- (10,0) -- (10,10)-- (0,10) -- (0,0);
\draw[dashed]  (0,0.5) -- (10,0.5);
\draw[dashed]  (0,9.5) -- (10, 9.5);
\draw[dotted]  (0,1.5) -- (10,1.5);
\draw[dotted]  (0,8.5) -- (10, 8.5);

\filldraw[fill=darkgray, draw=black] (0,8.5) rectangle (9.45, 8.8); 
\filldraw[fill=darkgray, draw=black] (0,8.8) rectangle (9.5, 9.1); 
\filldraw[fill=darkgray, draw=black] (0,9.1) rectangle (9.45, 9.35); 

\draw[pattern=north east lines, pattern color=gray] (2.1, 0.05) rectangle (3.1, 0.25);
\draw[pattern=north east lines, pattern color=gray] (7.6, 0.05) rectangle (7.75, 0.4);
\filldraw[fill=gray, draw=black] (0.5,0.1) rectangle (0.75, 2.75);
\filldraw[fill=lightgray, draw=black] (5,3) rectangle (5.2, 7.5);
\filldraw[fill=lightgray, draw=black] (7.5,3.5) rectangle (9.25, 3.75);
\filldraw[fill=lightgray, draw=black] (7.5,5) rectangle (7.75, 5.25);

\draw[->] (10.5,9) -- (10,9); 
\draw (10.5,9) node [anchor=west] {\LARGE $X$ rotated};

\fill[color=white] (0,10.1) rectangle (0.1,10.15);\end{tikzpicture}}
\caption{\label{f:uwrot4} 
}
\end{subfigure}
	\caption{Case B for cardinality 2DK with rotations. Dark gray rectangles are $X$, light gray rectangles are $Z$, gray (and hatched) rectangles are $Y$, hatched rectangles are $C_{T,\eps^{i+1}}$ and  $C_{T,\eps^{i+1}}$. Figure (a): original packing, Figure (b): modified packing leaving space for resource contraction on the top.}
	\label{f:uwrotB}
\end{figure}
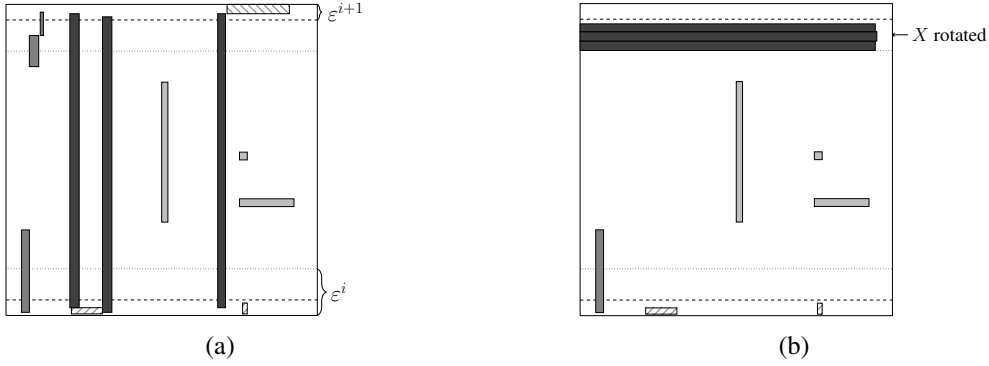

\noindent \textbf{Case B.}  $w(X) < 12 \eps \cdot \eps_sN$. Observe that $Y=(E_{T, \eps_s}\setminus X)\dot\cup (E_{B, \eps_s}\setminus X)$, hence $|Y|=|E_{T, \eps_s}\setminus X|+ |E_{B, \eps_s}\setminus X|$. Assume w.l.o.g. that  $|E_{B, \eps_s} \setminus X| \ge |Y|/2 \ge |E_{T, \eps_s}\setminus X|$.
Then remove $E_{T, \eps_s}$. We can pack $X$ on top of $M \setminus E_{T, \eps_s}$
 as $12 \eps\cdot \eps_s \le \eps_s - \eps\cdot \eps_s$ for $\eps \le 1/13$.
This gives a packing of $|X|+|Z|+\frac{|Y|}{2}$ many rectangles.
On the other hand, as $a(X \cup Y)=a(E_{T, \eps_s} \cup E_{B, \eps_s}) \le  \frac{(1+8 \eps_s)}{2}N^2$, from Lemma~\ref{lem:Stein}, it is possible to pack at least 
$(1-2 \eps_{s}-8\eps_s)|X \cup Y| \ge (1-10\eps_s)(|X|+|Y|)$ many items into $[0,(1-\eps\cdot \eps_s)N]\times [0,N]$.

Thus we can always pack a set of items of cardinality at least

$\begin{array}{rl} & \max\{(1-10\eps_s)(|X|+|Y|), |X|+|Z|+\frac{|Y|}{2}\} \\ \ge &
\frac{1}{3}(1-10\eps_s)(|X|+|Y|)+ \frac{2}{3}(|X|+|Z|+\frac{|Y|}{2}) \\
\ge & \frac{2}{3}(1-10\eps_s)(|X|+|Y|+|Z|) \\ = & \frac{2}{3}(1-10\eps_s)|M_3|.\end{array}$\\
This concludes the proof of Lemma \ref{uwrescontr1}.









%% file: 2Dknapsack8f-Weighted-case-Rotations.tex

\section{Weighted Case with Rotations}
\label{sec:weightedRot}

In this section we give a polynomial time $(3/2+\eps)$-approximation algorithm for the weighted 2-dimensional geometric knapsack problem when items are allowed to be rotated by 90 degrees.
In contrary to the unweighted case, where it is possible to remove a constant number of \emph{large} items, the same is not possible in the weighted case, where an item could have a big profit.

We call an item $i$ \emph{massive} if $\width(i) \geq (1 - \eps)N$ and $\height(i) \geq (1 - \eps)N$. The presence of such a big item in the optimal solution requires a different analysis, that we present below. In both the cases, we can show that there exists a container packing with roughly 2/3 of the profit of the optimal solution.

Let us assume that $\eps < 1/6$. We will prove the following result:

\begin{thm}\label{lem:structural_lemma_weighted}Let $\eps > 0$ and let $R$ be a set of items that can be packed into the $N \times N$ knapsack. Then there exists a container packing with $O_\eps(1)$ containers of a subset $R' \subseteq R$ into the $N \times N$ \andy{knapsack} such that $\profit(R') \geq (2/3 - O(\eps)) \profit(R)$, if rotations are allowed.
\end{thm}

We start by analyzing the case of a set $R$ that has a massive item.

\begin{lem}\label{lem:massiveitem}
Suppose that a set $R$ of items can be packed into a $N \times N$ bin and there is a massive item $m \in R$. Then, there is a container packing with at most $O_\eps(1)$ containers for a subset $R' \subseteq R$ such that $\profit(R') \geq \left(\dfrac{2}{3} - O(\eps) \right) \profit(R)$. 
\end{lem}
\begin{proof}
Assume, without loss of generality, that $1/(3\eps)$ is an integer. Consider the items in $R \setminus \{m\}$. Clearly, each of them has width or height at most $\eps$; moreover, $a(R \setminus \{m\}) \leq (1 - (1 - \eps)^2) N^2 = (2\eps - \eps^2)N^2 \leq \frac{N^2}{2(1+\eps)}$, as $\eps < 1/6$; thus, by possibly rotating each element so that the height is smaller than $\eps$, by Theorem~\ref{thm:steinberg} all the items in $R \setminus \{m\}$ can be packed in a $N \times \frac{N}{1 + \eps}$ bin; then, by Lemma~\ref{lem:structural_lemma_augm}, there is a container packing for a subset of $R \setminus \{m\}$ with $O_\eps(1)$ containers that fits in the $N \times N$ bin and has profit at least $(1 - O(\eps))\profit(R\setminus\{m\})$

Consider now the packing of $R$. Clearly, the region $[\eps N, (1 - \eps)N]^2$ is entirely contained within the boundaries of the massive item $m$. Partition the region with $x$-coordinate between $\eps N$ and $(1 - \eps)N$ in $k = 1/(3\eps)$ strips of width $3\eps(1 - 2\eps)N \geq 2\eps N$ and height $N$, let them be $S_1, \dots, S_{k}$; let $R(S_i)$ be the set of items in $R$ such that their left or right edge (or both) are contained in the interior of strip $S_i$. Since each item belongs to at most two of these sets, there exists $i$ such that $\profit(R(S_i)) \leq 6\eps \profit(R)$.\\
Symmetrically, we define $k$ horizontal strips $T_1, \dots, T_k$, obtaining an index $j$ such that $\profit(R(T_j)) \leq 6\eps \profit(R)$. Thus, no item in $\overline{R} := R \setminus (R(S_i) \cup R(T_j))$ has a side contained in the interior of $S_i$ or $T_j$, and $\profit(\overline{R}) \geq (1 - 12\eps)\profit(R)$. Let $M_V$ be the set of items in $\overline{R} \setminus \{m\}$ that overlap $T_j$, and let $M_H$ be the set of items in $\overline{R} \setminus \{m\}$ that overlap $S_i$. Clearly, the items in $M_H$ can be packed in a horizontal container with width $N$ and height $N - \height(m)$, and the items in $M_V$ can be packed in a vertical container of width $N - \width(m)$ and height $N$.

Let $H$ be the set of items of $\overline{R} \setminus M_H$ that are completely above the massive item $m$ or completely below it; symmetrically, let $V$ be the set of items of $\overline{R} \setminus M_V$ that are completely to the left or completely to the right of $m$. We will now show that there is a container packing for $M_H \cup V \cup \{m\}$. Since all the elements overlapping \ari{$T_j$} have been removed, $V$ can be packed in a bin of size $(N - \width(m)) \times (1 - 2\eps)N$ (see Figure~\ref{fig:massive_item}). Since $(1 - 2\eps)N \cdot (1 + \eps) < (1 - \eps)N \leq \height(m)$, Lemma~\ref{lem:structural_lemma_augm} implies that there is a container packing of a subset of $V$ with profit at least $(1 - O(\eps))\profit(V)$ in a bin of size $(N - \width(m)) \times \height(m)$ and using $O_\eps(1)$ containers; thus, by adding a horizontal container of the same size as $m$ and a horizontal container of size $N \times (N - \height(m))$, we obtain a container packing for $M_H \cup V \cup \{m\}$ with $O_\eps(1)$ containers and profit at least $(1 - O(\eps))\profit(M_H \cup V \cup \{m\})$. Symmetrically, there is a container packing for a subset of $M_V \cup H \cup \{m\}$ with profit at least $(1 - O(\eps))\profit(M_V \cup H \cup \{m\})$ and $O_\eps(1)$ containers.

Let $R_{MAX}$ be the set of maximum profit among the sets $R\setminus\{m\}$, $M_H \cup V \cup \{m\}$ and $M_V \cup H \cup \{m\}$. By the discussion above, there is a container packing for $R' \subseteq R_{MAX}$ with $O_\eps(1)$ containers and profit at least $(1 - O(\eps)) \profit(R_{MAX})$. Since each element in $\overline{R}$ is contained in at least two of the above three sets, it follows that:

\begin{align*}
\profit(R') & \geq (1 - O(\eps))\profit(R_{MAX}) \geq \left(1 - O(\eps)\right)\left(\frac{2}{3} \profit(\overline{R})\right)\\
      & \geq \left(\frac{2}{3} - O(\eps)\right) \profit(R)
\end{align*}
\end{proof}

\begin{figure}
	\centering
	\begin{subfigure}[b]{.44\textwidth}
	\resizebox{!}{180pt}{
	\begin{tikzpicture}
	
	
	\draw (0,0) rectangle (10,10);
	
	
	\fill[color=gray] (1,1) rectangle (9.25,9.25);
	\draw (1,1) rectangle (9.25,9.25);
	\draw (5.125,5.125) node {\textbf{\huge $m$}};
	
	
	\fill[color=lightgray, pattern=north east lines] (3.5,0.2) rectangle (6.5,0.4);
	\draw (3.5,0.2) rectangle (6.5,0.4);
	
	\fill[color=lightgray, pattern=north east lines] (3,0.4) rectangle (7.5,0.8);
	\draw (3,0.4) rectangle (7.5,0.8);
	
	\fill[color=lightgray, pattern=north east lines] (2.5,0.8) rectangle (6,1);
	\draw (2.5,0.8) rectangle (6,1);
	
	\fill[color=lightgray, pattern=north east lines] (3.5,9.4) rectangle (8,9.6);
	\draw (3.5,9.4) rectangle (8,9.6);
	
	\fill[color=lightgray, pattern=north east lines] (3,9.6) rectangle (6.7,9.8);
	\draw (3,9.6) rectangle (6.7,9.8);
	
	
	\fill[color=lightgray, pattern=north west lines] (0,3.5) rectangle (0.3,7);
	\draw (0,3.5) rectangle (0.3,7);
	
	\fill[color=lightgray, pattern=north west lines] (0.5,4) rectangle (0.8,7.2);
	\draw (0.5,4) rectangle (0.8,7.2);
	
	\fill[color=lightgray, pattern=north west lines] (9.3,2.5) rectangle (9.45,6.5);
	\draw (9.3,2.5) rectangle (9.45,6.5);
	
	\fill[color=lightgray, pattern=north west lines] (9.5,4) rectangle (9.8,7.5);
	\draw (9.5,4) rectangle (9.8,7.5);
	
	
	\fill[color=lightgray] (0.2,2.5) rectangle (0.5,0.5);
	\draw (0.2,2.5) rectangle (0.5,0.5);
	
	\fill[color=lightgray] (0.7,0.3) rectangle (1,0.9);
	\draw (0.7,0.3) rectangle (1,0.9);
	
	\fill[color=lightgray] (0.2,7.5) rectangle (0.6,9.7);
	\draw (0.2,7.5) rectangle (0.6,9.7);
	
	\fill[color=lightgray] (9.5,0.5) rectangle (9.8,3);
	\draw (9.5,0.5) rectangle (9.8,3);
	
	
	\draw[dashed, ultra thick] (4,0) -- (4,1);
	\draw[dashed, ultra thick] (5.5,0) -- (5.5,1);
	
	\draw[dashed, ultra thick] (4,9.25) -- (4,10);
	\draw[dashed, ultra thick] (5.5,9.25) -- (5.5,10);
	
	\draw[dashed, ultra thick] (0,4.5) -- (1,4.5);
	\draw[dashed, ultra thick] (0,6) -- (1,6);
	
	\draw[dashed, ultra thick] (9.25,4.5) -- (10,4.5);
	\draw[dashed, ultra thick] (9.25,6) -- (10,6);
	
	\draw[dashed, ultra thick] (1,-1) -- (1, 11);
	\draw[dashed, ultra thick] (9.25,-1) -- (9.25, 11);
	
	\draw[dashed, ultra thick] (-1,1) -- (11,1);
	\draw[dashed, ultra thick] (-1,9.25) -- (11,9.25);
	
	
	\draw (1,-0.75) node[anchor=east] {\Large{$V$}};
	
	\draw (9.25,-0.75) node[anchor=west] {\Large{$V$}};
	
	\draw (-0.75,1) node[anchor=north] {\Large{$H$}};
	
	\draw (-0.75,9.25) node[anchor=south] {\Large{$H$}};
	
	\draw (0,5.25) node[anchor=east] {\Large{$M_V$}};
	
	\draw (4.75,0) node[anchor=north] {\Large{$M_H$}};
	
	\end{tikzpicture}}
	\caption{Massive item case. Items intersecting strips $M_H$ and $M_V$ (hatched rectangles) cross them completely.}\label{fig:stripes}
	\end{subfigure}
	\hspace{35pt}
		\begin{subfigure}[b]{.44\textwidth}
		\resizebox{!}{180pt}{
		\begin{tikzpicture}
		
			
		\draw[thick] (0,0) rectangle (10,10);
		
		
		\draw[fill=gray] (5,2.5) rectangle (8,5.5);
		\fill[pattern = vertical lines] (5.5,0.5) rectangle (7.5,2);	
		\draw (5.5,0.5) rectangle (7.5,2);	
		\fill[pattern = horizontal lines] (2,6) rectangle (4,9);
		\fill[pattern = north east lines] (2,6) rectangle (4,9);
		\draw (2,6) rectangle (4,9);
		\fill[pattern = north east lines] (2,3) rectangle (4,5);
		\draw (2,3) rectangle (4,5);
		\fill[pattern = north west lines] (8.5,1.5) rectangle (9.5,4.5);
		\draw (8.5,1.5) rectangle (9.5,4.5);
		\fill[pattern = north west lines] (8.2,6.5) rectangle (9.6,8.5);
		\fill[pattern = horizontal lines] (8.2,6.5) rectangle (9.6,8.5);
		\draw (8.2,6.5) rectangle (9.6,8.5);
		
		
		\draw[dashed] (5,-1) -- (5,11);
		\draw[dashed] (8,-1) -- (8,11);
		\draw[dashed] (-1,2.5) -- (11,2.5);
		\draw[dashed] (-1,5.5) -- (11,5.5);
		
		
		\draw (6.5,4) node {\huge $i$};
		\filldraw (5,2.5) circle (2.5pt);\Large
		\draw (5,2.5) node[anchor = north east] { $(x_i,y_i)$};
		\filldraw (8,5.5) circle (2.5pt);
		\draw (8,5.5) node[anchor = south west] {\Large $(x_i',y_i')$};
		\draw [decorate,decoration={brace,amplitude=7pt}] (5,10) -- (8,10); 
		\draw (6.5,10.3) node [anchor = south] {\Large $w_i$};
		\draw [decorate,decoration={brace,amplitude=7pt}] (10,5.5) -- (10,2.5); 
		\draw (10.3,4) node [anchor = west] {\Large $h_i$};
		\draw (5,-0.6) node[anchor = east] {\Large $Left(i)$};
		\draw (8,-0.6) node[anchor = west] {\Large $Right(i)$};
		\draw (-0.6,2.5) node[anchor = north] {\Large $Bottom(i)$};
		\draw (-0.6,5.5) node[anchor = south] {\Large $Top(i)$};
		
		\end{tikzpicture}}
		\caption{$Bottom(i), Top(i), Left(i), Right(i)$ are represented by vertical, horizontal, north east  and north west stripes respectively.}\label{fig:massive_item1}
		\end{subfigure}
		\caption{}\label{fig:massive_item}
\end{figure}

If there is no massive item, we will show existence of  two {\em container packings} 
 and show the maximum of them always packs items with total profit at least $\left(\frac23-O(\eps)\right)$ fraction of the optimal profit. 

First, we follow the corridor decomposition and the classification of items as in Section \ref{sec:weighted} to define
sets $LF, SF, LT, ST, OPT_{small}$.
Let $T:= LT \cup ST$  be the set of thin items.
Also let $APX$ be the best {\em container packing} and $OPT$ be the optimal solution.
Then similar to Lemma \ref{lem:weighted-apx}, we can show $\profit(\apx)\ge (1-\eps)(\profit(LF)+\profit(SF)+\profit(\optsm))$.
Thus, 
\begin{equation}
\label{rotweightpack1}
\profit(\apx)\ge (1-\eps)\profit(OPT)-\profit(T).
\end{equation}

In the second case, we define the set $T$ as above. 
Then in Resource Contraction Lemma (Lemma \ref{lem:weightResContract}), we will show that one can pack $1/2$ of the remaining profit  in the optimal solution, i.e., $\profit(OPT \setminus T)/2$ in a 
knapsack of size $N \times (1-\eps/2)N$.
Now, we can pack $T$ in a horizontal container of height $\eps/4$ and using Lemma \ref{lem:weightResContract} and resource augmentation we can pack $\profit(OPT \setminus T)/2$ in the remaining space $N \times (1-\eps/4)N$.
Thus,
\begin{equation}
\label{rotweightpack1}
\profit(\apx)\ge \profit(T)+(1-\eps)(\profit(OPT)-\profit(T))/2.
\end{equation}
Hence, up to $(1-O(\eps))$ factor, we pack at least $\max\{(\profit(T)+\profit(OPT \setminus T)/2), \profit(OPT \setminus T)\} \ge 2/3 \cdot  \profit(OPT)$, thus proving Theorem~\ref{lem:structural_lemma_weighted}.

Note that using techniques similar to \andy{Appendix} \ref{sec:weighted}, we can get a PTAS for the best container packing.
Now to complete the proof of Theorem~\ref{thm:mainNoRotation}, it only remains to prove Lemma~\ref{lem:weightResContract}.
\begin{lem}(Resource Contraction Lemma)
\label{lem:weightResContract}
If a set of items $M$ contains no massive item and can be packed into a $N \times N$ bin, then it is possible to pack a set $M'$ of profit at least $\profit(M) \cdot \frac{1}{2}$ into a $N \times  (1-\frac{\eps}{2})N$ bin (or a $(1-\frac{\eps}{2})N \times N$ bin), if rotations are allowed.
\end{lem}

\begin{proof}
Let $\eps_s=\eps/2$.
We will partition $M$ into two sets $M_1, M \setminus M_1$ and show that both these sets can be packed into $N\times (1-\eps_s)N$ bin.
If an item $i$ is embedded in position $(x_i, y_i)$, we define $x'_i:=x_i+\width(i), y'_i :=y_i+\height(i)$.

In a packing of a set of items $M$, for item $i$ we define 
$Left(i):=\{k \in M: x'_k \le x_i \}$, 
$Right(i):=\{k \in M: x_k \ge x'_i \}$, 
$Top(i):=\{k \in M: y_k \ge y'_i \}$, 
$Bottom(i):=\{k \in M: y'_k \le y_i \}$, 
i.e., the set of items that lie completely on left, right, top and bottom of $i$ respectively.
Now consider four strips $S_{T,3\eps_s}, S_{B,\eps_s}, S_{L,\eps_s} ,S_{R,\eps_s}$ (see Figure~\ref{fig:resource_contraction}).\\

\begin{figure}
	\captionsetup[subfigure]{justification=centering}
	\begin{subfigure}[b]{.35\textwidth}
	\resizebox{!}{135pt}{
	\begin{tikzpicture}
		
		
		\draw[thick] (0,0) rectangle (8,8);
		
		
		\fill[color=lightgray] (4,0.6)  rectangle (5,2.5);
		\draw (4,0.6)  rectangle (5,2.5);
		\fill[color=lightgray] (5,3.5)  rectangle (5.7,6.2);
		\draw (5,3.5)  rectangle (5.7,6.2);
		\fill[color=lightgray] (0.5,6.5)  rectangle (5,7.2);
		\draw (0.5,6.5)  rectangle (5,7.2);
		\fill[color=lightgray] (1.5,2)  rectangle (2.4,4);
		\draw (1.5,2)  rectangle (2.4,4);
		
		
		\draw[dashed] (-1,0.8) -- (9,0.8);
		\draw[dashed] (-1,5.6) -- (9,5.6);
		\draw[dashed] (0.8,-1) -- (0.8,9);
		\draw[dashed] (7.2,-1) -- (7.2,9);
		
		
		\draw [decorate,decoration={brace,amplitude=6pt}] (8,8) -- (8,5.6); 
		\draw (8.2,6.8) node [anchor = west] {\Large $3\eps_s N$};
		\draw [decorate,decoration={brace,amplitude=6pt}] (8,0.8) -- (8,0); 
		\draw (8.2,0.4) node [anchor = west] {\Large $\eps_s N$};
		\draw [decorate,decoration={brace,amplitude=6pt}] (7.2,8) -- (8,8); 
		\draw (7.7,8.2) node [anchor = south] {\Large $\eps_s N$};
		\draw [decorate,decoration={brace,amplitude=6pt}] (0,8) -- (0.8,8); 
		\draw (0.3,8.2) node [anchor = south] {\Large $\eps_s N$};
		
		\draw (-0.6,0.7) node[anchor=north] {\Large $S_{B,\eps_s}$};
		\draw (0.8,-0.6) node[anchor=east] {\Large $S_{L,\eps_s}$};
		\draw (7.2,-0.6) node[anchor=west] {\Large $S_{R,\eps_s}$};
		\draw (-0.6,5.6) node[anchor=south] {\Large $S_{T,3\eps_s}$};
		\end{tikzpicture}}
		\caption{Case 1: \\$D_{B,\eps_s} \cap D_{T,3\eps_s}=\emptyset$.}
	\end{subfigure}
	\begin{subfigure}[b]{.3\textwidth}
	\resizebox{!}{135pt}{
	\begin{tikzpicture}
		
		
		\draw[thick] (0,0) rectangle (8,8);
		
		
		\draw[fill=gray] (3,0.5) rectangle (4,6.2);
		\draw (3.5,3.35) node {\huge $i$};
		
		
		\draw[dashed] (-0.5,0.8) -- (8.5,0.8);
		\draw[dashed] (-0.5,5.6) -- (8.5,5.6);
		\draw[dashed] (0.8,-1) -- (0.8,9);
		\draw[dashed] (7.2,-1) -- (7.2,9);
		
		\end{tikzpicture}}
		\caption{Case 2A: $i$ does not intersect $S_{L, \eps_s}$ or $S_{R, \eps_s}$.}
	\end{subfigure}
		\begin{subfigure}[b]{.3\textwidth}
		\resizebox{!}{135pt}{
			\begin{tikzpicture}
					
					
				\draw[thick] (0,0) rectangle (8,8);
					
					
				\draw[fill=gray] (0.6,0.6) rectangle (7.3,6);
				\draw (3.95,3.3) node {\huge $i$};
				
				
				\draw[dashed] (-0.5,0.8) -- (8.5,0.8);
				\draw[dashed] (-0.5,5.6) -- (8.5,5.6);
				\draw[dashed] (0.8,-1) -- (0.8,9);
				\draw[dashed] (7.2,-1) -- (7.2,9);
					
			\end{tikzpicture}}
			\caption{Case 2B: $i$ intersects both $S_{L, \eps_s}$ and $S_{R, \eps_s}$.}
		\end{subfigure}
	
	\vspace{3pt}
	\hspace{6.9pt}
	\begin{subfigure}[b]{.3\textwidth}
		\resizebox{!}{121.75pt}{
			\begin{tikzpicture}
						
						
				\draw[thick] (0,0) rectangle (8,8);
						
						
				\draw[fill=gray] (0.6,0.4) rectangle (2.4,6.2);
				\draw (1.5,3.3) node {\huge $i$};
						
						
				\draw[dashed] (-0.5,0.8) -- (8.5,0.8);
				\draw[dashed] (-0.5,5.6) -- (8.5,5.6);
				\draw[dashed] (0.8,-0.5) -- (0.8,8.5);
				\draw[dashed] (7.2,-0.5) -- (7.2,8.5);
				\draw[ultra thick] (4,0) -- (4,8);
						
				\end{tikzpicture}}
				\caption{Case 2C: $i \in D_{L,\eps_s}$ and $x_i'\le N/2$.}
		\end{subfigure}	
		\hspace{7.5pt}
	\begin{subfigure}[b]{.3\textwidth}
	\resizebox{!}{121.75pt}{
	\begin{tikzpicture}
				
				
		\draw[thick] (0,0) rectangle (8,8);
				
				
		\draw[fill=gray] (0.5,0.3) rectangle (5.1,7);
		\draw (2.8,3.65) node {\huge $i$};
				
				
		\draw[dashed] (-0.5,0.8) -- (8.5,0.8);
		\draw[dashed] (-0.5,5.6) -- (8.5,5.6);
		\draw[dashed] (0.8,-0.5) -- (0.8,8.5);
		\draw[dashed] (7.2,-0.5) -- (7.2,8.5);
		\draw[ultra thick] (4,0) -- (4,8);
				
		\end{tikzpicture}}
		\caption{Case 2C: $i \in D_{L,\eps_s}$ and $x_i' > N/2$.}
	\end{subfigure}
	\begin{subfigure}[b]{.3\textwidth}
		\resizebox{!}{121.75pt}{
\begin{tikzpicture}
	
	
	\draw[thick] (0,0) rectangle (8,8);
	
	
	\draw[fill=gray] (0.1,0.5) rectangle (0.6,6.3);
	\draw (0.35,3.4) node {\huge $i$};
	
	
	\draw[dashed] (-0.5,0.8) -- (8.5,0.8);
	\draw[dashed] (-0.5,5.6) -- (8.5,5.6);
	\draw[dashed] (0.8,-0.5) -- (0.8,8.5);
	\draw[dashed] (7.2,-0.5) -- (7.2,8.5);
	
	\end{tikzpicture}}
	\caption{Case 2C: \\ $i \in C_{L,\eps_s}$.}
	\end{subfigure}
	\caption{Cases for Resource Contraction Lemma (Lemma \ref{lem:weightResContract}).}\label{fig:resource_contraction}
\end{figure}
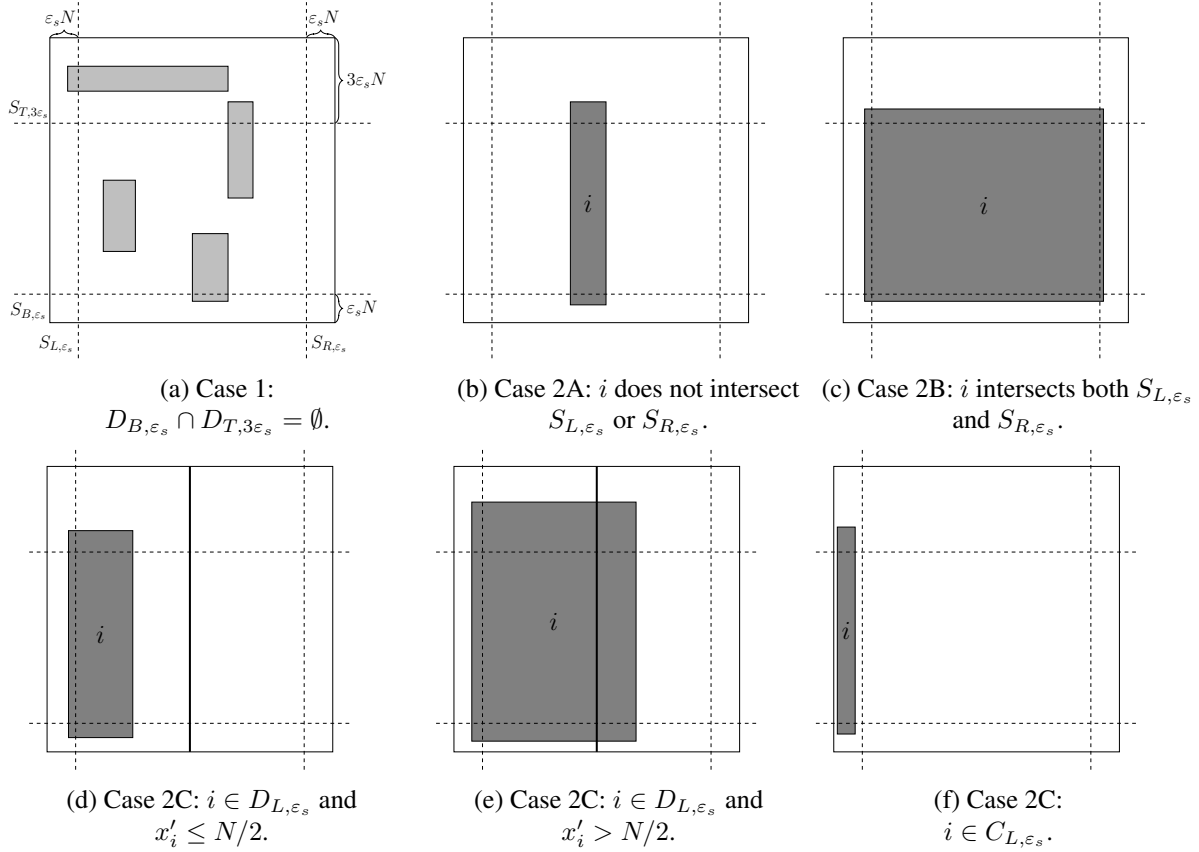

\noindent \textbf{Case 1.} $D_{B,\eps_s} \cap D_{T,3\eps_s}=\andy{\emptyset}$, i.e., no item intersecting $S_{B,\eps_s}$ intersects $S_{T,3\eps_s}$.
Define $M_1:=E_{T,3\eps_s}$. As these items in $M_1$ do not intersect $S_{B,\eps_s}$,
$M_1$ can be packed into a $(N, N(1-\eps_s))$ bin.
For the remaining items, pack $M \setminus (M_1 \cup C_{L, \eps_s} \cup C_{R, \eps_s})$ as it is. 
Now rotate $C_{L, \eps_s}$ and $C_{R, \eps_s}$ and pack on top of $M \setminus (M_1 \cup C_{L, \eps_s} \cup C_{R, \eps_s})$ into two strips of height $\eps_s N$ and width $N$. This packing will have total height $\le (1-3\eps_s+2\eps_s)N \le (1-\eps_s)N$.\\

\noindent \textbf{Case 2.} $D_{B,\eps_s} \cap D_{T,3\eps_s} \neq \andy{\emptyset}$, i.e., there is some item  intersecting $S_{B,\eps_s}$ that also crosses $S_{T,3\eps_s}$. Now, there are three subcases:
\\ \textbf{Case 2A.} \textit{There exists an item $i$ that does neither intersect $S_{L, \eps_s}$ nor $S_{R, \eps_s}$}.
Then item $i$ partitions the items in $M \setminus ( C_{T, 3\eps_s} \cup C_{B, \eps_s} \cup \{ i \} )$
into two sets: $Left(i)$ and $Right(i)$.
W.l.o.g., assume $x_i \le 1/2$. Then remove $Right(i), i, C_{T, 3 \eps_s}$ and $C_{B, \eps_s} $ from the packing.
Now rotate $C_{T, 3 \eps_s}$ and $C_{B, \eps_s}$ to pack right of $Left(i)$.
Define this set $M \setminus (Right(i)\cup \{ i \})$ to be $M_1$. Clearly packing of $M_1$ takes height $N$ and width $x_i+4\eps_s N \le (\frac12+4\eps_s)N \le (1-\eps_s)N$ as $\eps_s \le  \frac{1}{10}$.
As the item $i$ does not intersect the strip $S_{L, \eps_s}$, $(Right(i)\cup \{i\})$ can be packed into height $N$ and width $(1-\eps_s)N$.
\\ \textbf{Case 2B.} \textit{There exists an item $i$ that intersects both $S_{L, \eps_s}$ and $S_{R, \eps_s}$}. 
Consider $M_1$ to be $M \setminus (C_{L, \eps_s} \cup C_{R, \eps_s}  \cup Top(i) )$.
 As there is no massive item, $M_1$ is packed in height $(1-\eps_s)N$ and width $N$.
 Now, pack $Top(i)$ and then  rotate $C_{L, \eps_s}$ and $C_{R, \eps_s}$  to pack on top of it.
 These items can be packed into height $(1-y'_i+2\eps_s)N \le 5\eps_s N \le (1-\eps_s)N$ as $\eps_s \le 1/10$.
\\ \textbf{Case 2C.} \textit{If an item $i$ intersects both $S_{B, \eps_s}$ and $S_{T, 3 \eps_s}$,  then the item $i$ intersects exactly one of $S_{L, \eps_s}$ and $S_{R, \eps_s}$.} 
Consider the set of items in $D_{B, \eps_s} \cap D_{T, 3 \eps_s}$.\\
First, consider the case when the set $D_{B, \eps_s} \cap D_{T, 3 \eps_s}$ contains an item $i \in D_{L, \eps_s}$ (similarly one can consider $i \in D_{R, \eps_s}$). Now if $x'_i \le N/2$, take $M_1:=Right(i)$.
Then, we can rotate $Right(i)$ and pack into height $(1-\eps_s)N$ and width $N$. On the other hand, 
pack $M \setminus \{  M_1 \cup C_{T,3\eps_s} \cup C_{B,\eps_s} \}$ 
as it is. Then rotate $C_{T,3\eps_s} \cup C_{B,\eps_s}$ and pack on its side. Total width $\le (1/2+3\eps_s+\eps_s) N \le (1-\eps_s)N$ as $\eps_s \le 1/6$. Otherwise if $x'_i > N/2$ take $M_1:=Left(i)\cup i$.
Now, consider packing of $M \setminus \{  M_1 \cup C_{T,3\eps_s} \cup C_{B,\eps_s} \}$, rotate $C_{T,3\eps_s} \cup C_{B,\eps_s}$ and pack on its left. Total width $\le (1/2+4\eps_s)N \le (1-\eps_s)N$ as $\eps_s \le 1/10$. 
\arir{We need N/2 as the item i can have large width.}
\\Otherwise, no items in $S_{B, \eps_s} \cap S_{T, 3 \eps_s}$ are in $D_{L, \eps_s} \cup D_{R, \eps_s}$.
So let us assume that $i \in C_{L, \eps_s}$ (similarly one can consider $i \in C_{R, \eps_s}$), then  we take $M_1=E_{T,3\eps_s} \setminus (C_{L,\eps_s} \cup C_{R,\eps_s})$. Then we can rotate $C_{L, \eps_s}$ and $C_{R, \eps_s}$ and pack them on top of $M \setminus (M_1 \cup C_{L, \eps_s} \cup C_{R, \eps_s})$ as in Case 1. 
\end{proof}

%% file: 2Dknapsack8f-Tools.tex

\section{Some Tools}

In this section, we review some standard building blocks that we rely on in our construction.

\subsection{Next Fit Decreasing Height}

One of the most recurring tools used as a subroutine in countless results on geometric packing problems is the Next Fit Decreasing Height (NFDH) algorithm, which was originally analyzed in \cite{coffman1980performance} in the context of Strip Packing. We will use a variant of this algorithm to pack items inside a box, and analyze its properties. We provide a full proof for the sake of self-completeness.

Suppose you are given a box $C$ of size $w\times h$, and a set of items $I'$ each one fitting in the box (without rotations). NFDH computes in polynomial time a packing (without rotations) of $I''\subseteq I'$ as follows. It sorts the items $i\in I'$ in non-increasing order of height $h_i$, and considers items in that order $i_1,\ldots,i_n$. Then the algorithm works in rounds $j\geq 1$. At the beginning of round $j$ it is given an index $n(j)$ and a horizontal segment $L(j)$ going from the left to the right side of $C$. Initially $n(1)=1$ and $L(1)$ is the bottom side of $C$. In round $j$ the algorithm packs a maximal set of items $i_{n(j)},\ldots,i_{n(j+1)-1}$, with bottom side touching $L(j)$ one next to the other from left to right (a \emph{shelf}). The segment $L(j+1)$ is the horizontal segment containing the top side of $i_{n(j)}$ and ranging from the left to the right side of $C$. The process halts at round $r$ when either all items have being packed or $i_{n(r+1)}$ does not fit above $i_{n(r)}$. 

We prove the following:
\begin{lem}\label{lem:nfdhPack}
	Assume that, for some given parameter $\eps\in (0,1)$, for each $i\in I'$ one has $\width(i) \leq \eps w$ and $\height(i) \leq \eps h$. Then NFDH is able to pack in $C$ a subset $I''\subseteq I'$ of area at least $a(I'')\geq \min\{a(I'),(1-2\eps)w\cdot h\}$. In particular, if $a(I')\leq (1-2\eps)w\cdot h$, all items in $I'$ are packed.
\end{lem}

\begin{proof}
	The claim trivially holds if all items are packed.  Thus suppose that this is not the case.
	Observe that $\sum_{j=1}^{r+1}\height({i_{n(j)}}) > h$, otherwise item $i_{n(r+1)}$ would fit in the next shelf above $i_{n(r)}$; hence $\sum_{i=2}^{r+1}\height({i_{n(j)}}) > h - \height({i_{n(1)}}) \geq (1 - \eps)h$. Observe also that the total width of items packed in each round $j$ is at least $w-\eps w = (1 - \eps)w$, since $i_{n(j+1)}$, of width at most $\eps w$, does not fit to the right of $i_{n(j+1)-1}$. It follows that the total area of items packed in round $j$ is at least $(w-\eps w)\height({n(j+1)-1})$, and thus
	$$
	a(I'')\geq \sum_{j=1}^{r}(1-\eps)w \cdot \height({n(j+1)-1})\geq (1-\eps )w\sum_{j=2}^{r+1}\height({n(j)})\geq(1-\eps)^2w\cdot h\geq (1 - 2\eps)w\cdot h.
	$$
\end{proof}

\subsection{Maximum Generalized Assignment Problem}
\label{sec:GAP}
In this section we show that there is a PTAS for the Maximum Generalized Assignment Problem (GAP) if the number of bins is constant.
In GAP, we are given a set of $k$ bins with capacity constraints and a set of $n$ items
that have a possibly different size and profit for each bin and the goal is to  
pack a maximum-profit subset of items into the bins.
Let us assume that if item $i$ is packed in bin $j$, then it requires size  $s_{ij} \in \mathbb{Z}$ and profit $p_{ij} \in \mathbb{Z}$.

GAP is known to be APX-hard and the best known polynomial time approximation algorithm has ratio $(1-1/e+\eps)$ \cite{fleischer2011tight, feige2006approximation}. 
In fact, for arbitrarily small constant $\delta > 0$ (which can even be a function of $n$) GAP remains APX-hard even on the following instances: bin capacities are identical, and for each item $i$ and bin $j$, $p_{ij} = 1$, and $s_{ij} = 1$ or $s_{ij}=1+ \delta$ \cite{chekuri2005polynomial}. The complementary case, where item sizes do not vary across bins but profits do, is also APX-hard \cite{chekuri2005polynomial}.
However, when all profits and sizes are same across all bins (i.e., $p_{ij}=p_{ik}$ and $s_{ij}=s_{ik}$ for all bins $j,k$), the problem is known as multiple knapsack problem (MKP) and it admits PTAS \cite{chekuri2005polynomial}.

On the other hand, for our purposes we only need instances where $k = O(1)$. 
A PTAS for GAP for a constant number of bins follows from extending known techniques from the literature \cite{shmoys-tardos, shmoys1993approximation}.
However, we did not find an explicit proof in the literature and thus, for the sake of completeness, in this section we present a full, self-contained description of such an algorithm.

Let $C_j$ be the capacity of bin $j$ for $j \in [k]$.
Let $p(OPT)$ be the cost of the optimal assignment.

\begin{lem}
\label{lem:GAPresaug}
There is a {$O\left({\left(\frac{1 + \eps}{\eps}\right)}^k n^{k+1}\right)$} time algorithm for the maximum generalized assignment problem with $k$ bins, which returns a solution with profit at least $p(OPT)$ if we are allowed to augment the bin capacities by a $(1+\eps)$-factor {for any fixed $\eps>0$}.
\end{lem}
\begin{proof}
For each $i \in [n]$ and $c_j \in [C_j]$ for $j \in [k]$, let $S_{i, c_1,c_2, \dots, c_k}$ denote a subset of the set of items $\{1, 2, \dots, i\}$ packed into the bins such that the profit is maximized and capacity of bin $j$ is at most $c_j$. 
Let $P[i, c_1, c_2, \dots, c_k]$ denote the profit of $S_{i, c_1,c_2, \dots, c_k}$.
Clearly $P[1, c_1, c_2, \dots, c_k]$ is known for all $c_j \in [C_j]$ for $j \in [k]$. Moreover, we define $P[i, c_1, c_2, \dots, c_k] = 0$ if $c_j < 0$ for any $j \in [k]$.
We can compute the value of $P[i, c_1, c_2, \dots, c_k]$ by using a dynamic program (DP), that exploits the following recurrence:
\begin{align*}
 P[i, c_1, c_2, \dots, c_k] = \max\{&P[i-1, c_1, c_2, \dots, c_k],\\
 & \max_j \{p_{ij}+ P[i-1, c_1, \dots, c_j -s_{ij},  \dots, c_k]\}\}
\end{align*}
With a similar recurrence, we can easily compute a corresponding set $S_{i, c_1,c_2, \dots, c_k}$.

The running time of the above program is $O\Big(n \prod\limits_{j=1}^k C_j\Big)$. If each $C_j$ is polynomially bounded, then this running time is polynomial. Therefore, we now create a modified instance where each bin size is polynomially bounded.

Let $\mu_j=\frac{\eps C_j}{n}$.
For item $i$ and bin $j$, define the modified size  $s'_{ij}= \left\lceil \frac{s_{ij}}{\mu_j} \right\rceil = \left\lceil \frac{n  s_{ij}}{\eps C_j} \right\rceil$
and $C'_{j}= \left\lfloor \frac{(1+\eps)C_j}{\mu_j} \right\rfloor$.
{Note that $C'_j = \left\lfloor\frac{(1+\eps)n}{\eps}\right\rfloor \leq \frac{(1 + \eps)n}{\eps}$, so the above DP requires time at most $O\left(n \cdot \left(\frac{(1 + \eps)n}{\eps}\right)^k\right)$
}

The above DP finds the optimal solution $OPT_{modified}$ for the modified instance.
Now consider the optimal solution for the original instance (i.e., with original item sizes and bin sizes) $OPT_{original}$. 
If we show the same assignment of items to the bins is a feasible solution  (with modified bin sizes and item sizes)  for the modified instance, we get $OPT_{modified} \ge OPT_{original}$  and that will conclude the proof.

Let $S_j$ be the set of items packed in bin $j$ in the $OPT_{original}$.
So, $\sum_{i \in S_j} s_{ij} \le C_j$. Hence,
$$\sum_{i \in S_j} s'_{ij} 
\le  \left\lfloor \sum_{i \in S_j} \left(\frac{s_{ij}}{\mu_j}+1\right) \right\rfloor  \le  \left\lfloor \frac{1}{\mu_j}\left(\sum_{i \in S_j} s_{ij} +|S_j| \mu_j \right)\right\rfloor
\le  \left\lfloor \frac{1}{\mu_j}(C_j + n \mu_j) \right\rfloor \le  \left\lfloor \frac{(1+\eps)C_j}{\mu_j} \right\rfloor=C'_j
$$
Thus $OPT_{original}$ is a feasible solution for the modified instance and the DP will return a packing with profit at least $p(OPT)$ under $\eps$-resource augmentation.
\end{proof}

Now we can show how to employ this result to obtain a feasible solution with an almost optimal profit using the original bin capacities.
\begin{lem}\label{lem:GAP}
{There is an algorithm for maximum generalized assignment problem with $k$ bins that runs in time $O\left(\left(\frac{1 + \eps}{\eps}\right)^k n^{k/\eps^2 + k + 1}\right)$ and returns a solution that has profit at least $(1 - 3\eps)p(OPT)$, for any fixed $\eps > 0$.}
\end{lem}
\begin{proof}

First, we claim the following:
\begin{claim}
If a set of items $R_j$ is packed in a bin $B_j$ with capacity $C_j$, then there exists a set of  at most $O(1/\eps^2)$ items $X_j$, and a set of items $Y_j$ with $p(Y_j) \le \eps p(R_j)$ such  that all items in $R_j\setminus (X_j \cup Y_j)$ have size at most $\eps (C_j - \sum_{i \in X_j} s_{ij})$. 
\end{claim}
\begin{proof}
Let $Q_1$ be the set of items $i$ with $s_{ij} {>} \eps C_j $. If $p(Q_1) \le \eps p(R_j)$, we are done by taking $Y_j=Q_1$ and $X_j=\phi$.
Otherwise, define $X_j:=Q_1$ and we continue the next iteration with the remaining items.
Let $Q_2$ be the items with {size greater than} $\eps (C_j - \sum_{i \in X_j} s_{ij})$ in $R_j \setminus X_j$. If $p(Q_2) \le \eps p(R_j)$, we are done by taking $Y_j=Q_2$.
Otherwise define $X_j:= Q_1 \cup Q_2$ and we continue with further iterations till  
we get a set $Q_t$ with $p(Q_t) \le \eps p(R_j)$. Note that we need at most $\frac{1}{\eps}$ iterations since the sets $Q_i$ are disjoint.
Otherwise, $\displaystyle p(R_j) \ge \sum\limits_{i=1}^{1/\eps} p(Q_i) > \sum\limits_{i=1}^{1/\eps} \eps p(R_j) \ge p(R_j)$,  which is a contradiction.
Thus, consider $Y_j=Q_t$ and $X_j=\bigcup_{l=1}^{t-1} Q_l$. One has $|X_j| \le 1/\eps^2$ and $p(Y_j)\le \eps p(R_j)$.  On the other hand, after removing $Q_t$, the remaining  items have size $< \eps (C_j - \sum_{i \in X_j} s_{ij})$. 
\end{proof}
Now consider a bin with bin capacity of $(C_j - \sum_{i \in X_j} s_{ij})$ where all packed items $R'_j$ have sizes $< \eps (C_j - \sum_{i \in X_j} s_{ij})$, then we can divide the bin into $1/\eps$ equal sized intervals $S_{j,1}, S_{j,2}, \dots, S_{j,{1/\eps}}$ of lengths $\eps(C_j - \sum_{i \in X_j} s_{ij})$. Let $R'_{j,l}$ be the set of items intersecting the interval $S_{j,l}$.
 As each packed item can belong to at most two such intervals, the cheapest  set $R''$ among $\{ R'_{j,1 }, \dots, R'_{j,1/\eps} \}$ has profit at most $2 \eps p(R'_j)$. Thus we can remove this set $R''$ and reduce the bin size by a factor of $(1-\eps)$.

Now consider the packing of $k$ bins $B_j$'s in the optimal packing $OPT$. Let $R_j$ be the set of items packed in bin $B_j$. 
Now the algorithm first guesses all $X_j$'s,  a constant number of items, in all $k$ bins. We assign them to corresponding bins in $O(n^{k/\eps^2})$ time. 
Then for bin $j$ we are left with capacity $r_j := C_j - \sum_{i \in X} s_{ij}$.
From previous discussion, we know that there is packing of $R''_j \subseteq R_j \setminus X_j$ of profit $(1-2\eps)p(R_j \setminus X_j)$ in a bin with capacity $(1-\eps)C_j$. 
Thus  we can use resource augmentation algorithm for GAP in Lemma~\ref{lem:GAPresaug} to pack remaining items in $k$ bins where for bin $j$ we use  original capacity to be $(1-\eps)C_j$ for $j \in [k]$ before the resource augmentation.
As Lemma~\ref{lem:GAPresaug} returns the optimal packing on this modified bin sizes we get total profit $\ge (1-3\eps)p(OPT)$. 
\end{proof}

%% file: 2Dknapsack9a-Resource-augmentation.tex

\section{Packing rectangles with resource augmentation}
\label{sec:resource-augmentation}

In this section we prove that it is possible to pack a high profit subset of rectangles into boxes, if we are allowed to augment one side of a knapsack by a small fraction.

The result is essentially proved in \cite{Jansen2009310}, although we introduced some modifications and extensions to obtain the additional properties relative to packing into containers and a guarantee on the area of the obtained packing. For the sake of completeness, we provide a full proof, which follows in spirit the proof of the original result, from which we also borrow several notations. We will prove the following stronger version of Lemma~\ref{lem:augmentPack}:

\sal{We say that a container $C'$ is smaller than a container $C$ if $w(C') \leq w(C)$ and $h(C') \leq h(C)$. Given a container $C$ and a positive $\eps < 1$, we say that a rectangle $R_j$ is $\eps$-small for $C$ if $w_j \leq \eps w(C)$ and $h_j \leq \eps h(C)$.}

\begin{lem}[Resource Augmentation Packing Lemma]\label{lem:structural_lemma_augm}
	Let $\R'$ be a collection of rectangles that can be packed into a box of size $a\times b$, and $\epsau>0$ be a given constant. Then there exists a container packing of $\R''\subseteq \R'$ inside a box of size $a\times (1+\epsau)b$ (resp., $(1+\epsau)a\times b$) such that:
	\begin{enumerate}\itemsep0pt
		\item $p(\R'')\geq (1-O(\epsau))p(\R')$;
		\item the number of containers is $O_{\epsau}(1)$ and their sizes belong to a set of cardinality $n^{O_{\epsau}(1)}$ that can be computed in polynomial time;
		\item the total area of the containers is at most $a(\R') + \epsau ab$.
	\end{enumerate}
\end{lem}
\sal{In this result, we assume that rectangles in an area container $C$ are $\epsau$-small for $C$.}

Note that we do not allow rotations, that is, rectangles are packed with the same orientation as in the original packing. However, as an existential result we can apply it also to the case with rotations. Moreover, since Lemma~\ref{lem:containersPackPTAS} gives a PTAS for approximating container packings, this implies a simple algorithm that does not need to solve any LP to find the solution, in both the cases with and without rotations.

For simplicity, in this section we assume that widths and heights are positive real numbers in $(0, 1]$, and $a = b = 1$: in fact, all elements, container and boxes can be rescaled without affecting the property of a packing of being a \emph{container packing} with the above conditions. Thus, without loss of generality, we prove the statement for the augmented $1 \times (1 + \epsau)$ box.

\sal{Let $\epsau' = \epsau/2 < \epsau$. We will first obtain a packing where all the elements of each area container $C$ are $\epsau'$-small for $C$, and in Section~\ref{sec:rounding_containers} we will obtain the final packing, where the sizes of each container are taken from a polynomially sized set of choices.}

We will use the following Lemma, that follows from the analysis in \cite{kenyon2000near}:

\begin{lem}[Kenyon and Rémila \cite{kenyon2000near}]\label{lem:Kenyon-Remila}
	Let $\overline{\eps} > 0$, and let $\mathcal{Q}$ be a set of rectangles, each of height and width at most $1$. Let $\mathcal{L} \subseteq \mathcal{Q}$ be the set of rectangles of width at least $\overline{\eps}$, and let $OPT_{SP}(\mathcal{L})$ be the minimum width such that the rectangles in $\mathcal{L}$ can be packed in a box of size $OPT_{SP}(\mathcal{L})\times 1$.
	
	Then $\mathcal{Q}$ can be packed in polynomial time into a box of height $1$ and width $\tilde{w} \le \max\{OPT_{SP}(\mathcal{L}) + \frac{18}{\overline{\eps}^2} w_{\max}, a(\mathcal{Q})(1+\overline{\eps}) + \frac{19}{\overline{\eps}^2}w_{\max}\}$, where $w_{\max}$ is the maximum width of rectangles in $\mathcal{Q}$. Furthermore, all the rectangles with both width and height less than $\overline{\eps}$ are packed into at most $\frac{9}{\overline{\eps}^2}$ boxes, and all the remaining rectangles into at most $\frac{27}{\overline{\eps}^3}$ vertical containers.
\end{lem}
Note that the boxes containing the rectangles that are smaller than $\overline{\eps}$ are not necessarily packed as containers.

We need the following technical lemma:
\begin{lem}\label{lem:ra_intermediate}
	Let $\eps>0$ and let $f(\cdot)$ be any positive increasing function such that $f(x)<x$ for all $x$. Then, there exist positive constant values $\delta,\mu \in \Omega_\eps(1)$, with \sal{$f(\eps)\geq \delta$} and \sal{$f(\delta)\geq \mu$} such that the total profit of all the rectangles whose width or height lies in $(\mu, \delta]$ is at most $\eps\cdot p(\R')$.
\end{lem}
\begin{proof}
	Define $k+1=2/\eps+1$ constants $\eps_1,\ldots,\eps_{k+1}$, with $\eps_1=f(\eps)$ and $\eps_{i}=f(\eps_{i+1})$ for each~$i$. Consider the $k$ ranges of widths and heights of type $(\eps_{i+1},\eps_{i}]$. By an averaging argument there exists one index $j$ such that the total profit of the rectangles in $\R'$ with at least one side length in the range $(\eps_{j+1}N,\eps_{j}N]$ is at most $2\frac{\eps}{2}p(\R')$. It is then sufficient to set $\delta=\eps_j$ and $\mu=\eps_{j+1}$.    
\end{proof}
We use this lemma with $\eps = \epsau'$, and we will specify the function $f$ later. By properly choosing the function $f$, in fact, we can enforce constraints on the value of $\mu$ with respect to $\delta$, which will be useful several times; the exact constraints will be clear from the analysis. Thus, we remove from $\R'$ the rectangles that have at least one side length in $(\mu, \delta]$.

We call a rectangle $R_i$ \emph{wide} if $w_i > \delta$, \emph{high} if $h_i > \delta$, \emph{short} if $w_i \le \mu$ and \emph{narrow} if $h_i \le \mu$.\footnote{Note that the classification of the rectangles in this section is different from the ones used in the main results of this paper, although similar in spirit.}
From now on, we will assume that we start with the optimal packing of the rectangles in $R'$, and we will modify it until we obtain a packing with the desired properties.
We remove from $R'$ all the short-narrow rectangles, initially obtaining a packing. We will show in section \ref{sec:shortnarrow} how to use the residual space to pack them, with a negligible loss of profit.

As a first step, we round up the widths of all the \emph{wide} rectangles in $R'$ to the nearest multiple of $\delta^2$; moreover, we shift them horizontally so that their starting coordinate is an integer multiple of $\delta^2$ (note that, in this process, we might have to shift also the other rectangles in order to make space). Since the width of each wide rectangle is at least $\delta$ and $\frac{1}{\delta}\cdot 2\delta^2 = 2\delta$, it is easy to see that it is sufficient to increase the width of the box to $1 + 2\delta$ to perform such a rounding.

\subsection{Containers for short-high rectangles}\label{sec:containers-short-rectangles}

\begin{figure}
	\centering
	\includegraphics[width=10.7cm]{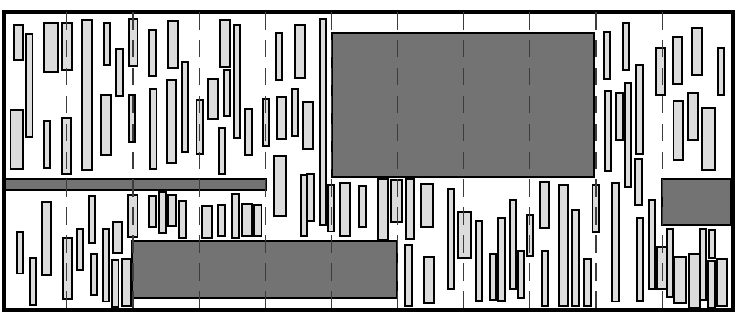}
	\caption{An example of a packing after the short-narrow rectangles have been removed, and the wide rectangles (in dark grey) have been aligned to the $M$ vertical strips. Note that the short-high rectangles (in light gray) are much smaller than the vertical strips.}
	\label{fig:augm-round-x}
\end{figure}

We draw vertical lines across the $1\times (1+2\delta)$ region spaced by $\delta^2$, splitting it into $M := \frac{1+2\delta}{\delta^2}$ vertical strips (see Figure~\ref{fig:augm-round-x}). Consider each maximal rectangular region which is contained in one such strip and does not overlap any wide rectangle; we define a box for each such region that contains at least one short-high rectangle, and we denote the set of such boxes by $\mathcal{B}$.

\begin{figure}
	\centering
	\includegraphics[width=8cm]{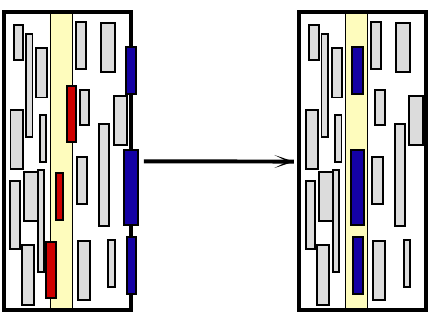}
	\caption{For each vertical box, we can remove a low profit subset of rectangles (red in the picture), to make space for short-high rectangles that cross the right edge of the box (blue).}
	\label{fig:augm-fix-vertical-box}
\end{figure}

Note that some short rectangles might intersect the vertical edges of the boxes, but in this case they overlap with exactly two boxes. Using a standard shifting technique, we can assume that no rectangle is cut by the boxes by losing profit at most $\epsau' OPT$: first, we assume that the rectangles intersecting two boxes belong to the leftmost of those boxes. For each box $B \in \mathcal{B}$ (which has width $\delta^2$ by definition), we divide it into vertical strips of width $\mu$. Since there are $\frac{\delta^2}{\mu} > 2/\epsau'$ strips and each rectangle overlaps with at most $2$ such strips, there must exist one of them such that the profit of the rectangles intersecting it is at most $2\mu p(B) \leq \epsau' p(B)$, where $p(B)$ is the profit of all the rectangles that are contained in or belong to $B$. We can remove all the rectangles overlapping such strip, creating in $B$ an empty vertical gap of width $\mu$, and then we can move all the rectangles intersecting the right boundary of $B$ to the empty space.
\begin{pro}
	The number of boxes in $\mathcal{B}$ is at most $\frac{1+2\delta}{\delta^2} \cdot \frac{1}{\delta} \le \frac{2}{\delta^3}$.
\end{pro}
First, by a shifting argument similar to above, we can reduce the width of each box to $\delta^2 - \delta^4$ while losing only an $\epsau'$ fraction of the profit of the rectangles in $B$. Then, for each $B\in \mathcal{B}$, since the maximum width of the rectangles in $B$ is at most $\mu$, by applying Lemma~\ref{lem:Kenyon-Remila} with $\overline{\eps} = \delta^2/2$ we obtain that the rectangles packed inside $B$ can be repacked into a box $B'$ of height $h(B)$ and width at most $w'(B) \le \max\{\delta^2 - \delta^4 + \frac{72}{\delta^4} \mu, (\delta^2 - \delta^4)(1+\frac{\delta^2}{2}) + \frac{76}{\delta^4}\mu\}\le \delta^2$, which is true if we make sure that $\mu \leq \delta^{10} / 76$. Furthermore, the short-high rectangles in $B$ are packed into at most $\dfrac{216}{\delta^6} \leq \dfrac{1}{\delta^7}$ vertical containers, assuming without loss of generality that $\delta \leq 1/216$. \sal{Note that all the rectangles are packed into vertical containers, because rectangles that have both width and height smaller than $\overline{\eps}$ are short-narrow and we already removed them.} Summarizing:

\begin{pro}\label{prop:basic_structure} There is a set $\R^+ \subseteq \R'$ of rectangles with total profit at least $(1-O(\epsau'))\cdot p(\R')$ and a corresponding packing for them in a $1 \times (1+2\delta)$ region such that:
	\begin{itemize}
		\item every wide rectangle in $\R^+$ has its length rounded up to the nearest multiple of $\delta^2$ and it is positioned so that its left side is at a position $x$ which is a multiple of $\delta^2$, and
		\item each box $B \in \mathcal{B}$ storing at least one short-high rectangle has width $\delta^2$, and the rectangles inside are packed into at most $1/\delta^7$ vertical containers.
	\end{itemize}
\end{pro}

\subsection{Fractional packing with $O(1)$ containers}

Let us consider now the set of rectangles $\R^+$ and an almost optimal packing $S^+$ for them according to Proposition~\ref{prop:basic_structure}. We remove the rectangles assigned to boxes in $\mathcal{B}$ and consider each box $B \in \mathcal{B}$ as a single pseudoitem. Thus, in the new almost optimal solution there are just pseudoitems from $\mathcal{B}$ and wide rectangles with right and left coordinates that are multiples of $\delta^2$. We will now show that we can derive a fractional packing with the same profit, and such that the rectangles and pseudoitems can be (fractionally) assigned to a constant number of containers. By \emph{fractional packing} we mean a packing where horizontal rectangles are allowed to be sliced horizontally (but not vertically); we can think of the profit as being split proportionally to the heights of the slices.

Let $\mathcal{K}$ be a subset of the horizontal rectangles of size $K$ that will be specified later. By extending horizontally the top and bottom edges of the rectangles in $\mathcal{K}$ and the pseudoitems in $\mathcal{B}$, we partition the knapsack into at most $2(|K| + |\mathcal{B}|) + 1 \leq 2(K + \frac{2}{\delta^3}) + 1 \leq 2(K + \frac{3}{\delta^3})$ horizontal stripes.

Let us focus on the (possibly sliced) rectangles contained in one such stripe of height $h$. For any vertical coordinate $y \in [0,h]$ we can define the \emph{configuration} at coordinate $y$ as the set of positions where the horizontal line at distance $y$ from the bottom cuts a vertical edge of a horizontal rectangle which is not in $\mathcal{K}$. There are at most $2^{M-1}$ possible configurations in a stripe.

\begin{figure}
	\centering
	\includegraphics[width=14.1cm]{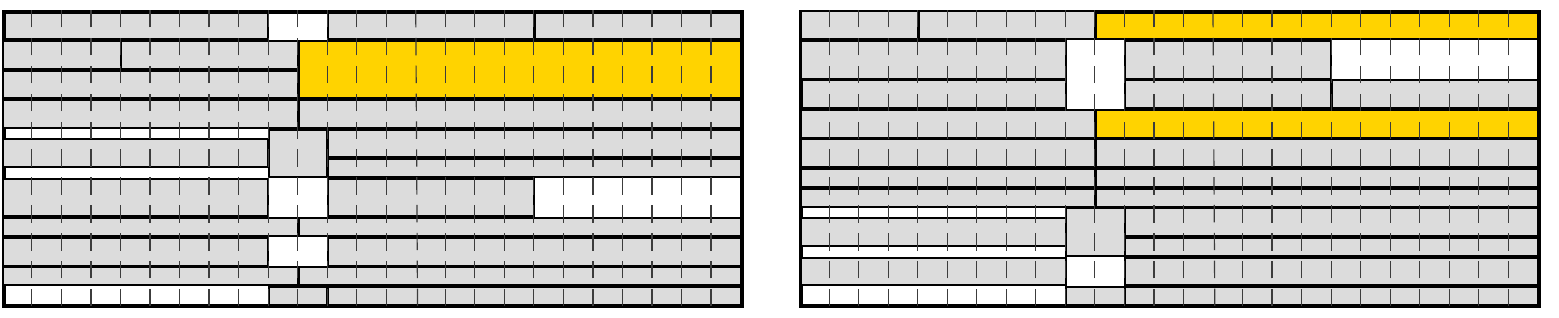}
	\caption{
		Rearranging the rectangles in a horizontal stripe. On the right, rectangles are repacked so that regions with the same configuration appear next to each other. Note that the yellow rectangle has been sliced, since it partakes in two regions with different configurations.
	}
	\label{fig:augm-fractional-repacking}
\end{figure}

We can further partition the stripe in maximal contiguous regions with the same configuration. Note that the number of such regions is not bounded, since configurations can be repeated. But since the rectangles are allowed to be sliced, we can rearrange the regions so that all the ones with the same configuration appear next to each other; see Figure~\ref{fig:augm-fractional-repacking} for an example. After this step is completed, we define up to $M$ horizontal containers per each configuration, where we repack the sliced horizontal rectangles. Clearly, all sliced rectangles are repacked.

Thus, the number of horizontal containers that we defined per each stripe is bounded by $M2^{M-1}$, and the total number overall is at most
\[
2\left(K + \frac{3}{\delta^3}\right) M 2^{M-1} = \left(K + \frac{3}{\delta^3}\right) M 2^{M}.
\]
\subsection{Existence of an integral packing}\label{sec:integral_packing}

%
%

We will now show the existence of an integral packing, at a small loss of profit.

Consider a fractional packing in $N$ containers. Since each rectangle slice is packed in a container of exactly the same width, it is possible to pack all but at most $N$ rectangles integrally by a simple greedy algorithm: choose a container, and greedily pack in it rectangles of the same width, until either there are no rectangles left for that width, or the next rectangle does not fit in the current container. In this case, we discard this rectangle and close the container, meaning that we do not use it further. Clearly, only one rectangle per container is discarded, and no rectangle is left unpacked.

The only problem is that the total profit of the discarded rectangles can be large. To solve this problem, we use the following shifting argument. Let $\mathcal{K}_0 = \emptyset$ and $K_0 = 0$. For convenience, let us define $f(K) = \left(K + \frac{3}{\delta^3}\right) M 2^{M}$.

First, consider the fractional packing obtained by choosing $\mathcal{K} = \mathcal{K}_0$, so that $K = K_0 = 0$. Let $\mathcal{K}_1$ be the set of discarded rectangles obtained by the greedy algorithm, and let $K_1 = |\mathcal{K}_1|$. Clearly, by the above reasoning, the number of discarded rectangles is bounded by $f(K_0)$. If the profit $p(\mathcal{K}_1)$ of the discarded rectangles is at most $\epsau' p(OPT)$, then we remove them and there is nothing else to prove. Otherwise, consider the fractional packing obtained by fixing $\mathcal{K} = \mathcal{K}_0 \cup \mathcal{K}_1$. Again, we will obtain a set $\mathcal{K}_2$ of discarded rectangles such that $K_2 := |\mathcal{K}_2| \leq f(K_0 + K_1)$. Since the sets $\mathcal{K}_1, \mathcal{K}_2, \dots$ that we obtain are all disjoint, the process must stop after at most $1/\epsau'$ iterations. Setting $p := M2^{M}$ and $q := \frac{3}{\delta^3} M2^{M}$, we have that $K_{i+1} \leq p(K_0 + K_1 + \dots K_i) + q$ for each $i \geq 0$. Crudely bounding it as $K_{i+1} \leq i \cdot pq \cdot K_i$, we immediately obtain that $K_i \leq (pq)^i$. Thus, in the successful iteration, the size of $\mathcal{K}$ is at most $K_{1/\epsau'-1}$ and the number of containers is at most $K_{1/\epsau'} \leq (pq)^{1/\epsau'} = (\frac{3}{\delta^2}M^2 2^{2M})^{1/\epsau'} = O_{\epsau', \delta}(1)$.

\subsection{Rounding down horizontal and vertical containers}\label{sec:shrinking_horiz_vert}

As per the above analysis, the total number of horizontal containers is at most  $(\frac{3}{\delta^2}M^2 2^{2M})^{\epsau'}$ and the total number of vertical containers is at most $\frac{2}{\delta^3} \cdot \frac{1}{\delta^7} = \frac{2}{\delta^{10}}$.

We will now show that, at a small loss of profit, it is possible to replace each horizontal and each vertical container defined so far with a constant number of smaller containers, so that the total area of the new containers is at most as big as the total area of the rectangles originally packed in the container. Note that in each container we consider the rectangles with the original widths (not rounded up). We use the following lemma:

\begin{lem}\label{lem:shrink_knapsack_container} Let $C$ be a horizontal (resp. vertical) container defined above, and let $\R_C$ be the set of rectangles packed in $C$. Then, it is possible to pack a set $\R'_C \subseteq \R_C$ of profit at least $(1 - 3\epsau')p(\R_C)$ in a set of at most $\left\lceil\log_{1 + \epsau'} (\frac{1}{\delta})\right\rceil / \epsau'^2$ horizontal (resp. vertical) containers that can be packed inside $C$ and such that their total area is at most $a(\R_C)$. 
\end{lem}
\begin{proof}
	Without loss of generality, we prove the result only for the case of a horizontal container.
	
	Since $w_i \geq \delta$ for each rectangle $R_i \in \R_C$, we can partition the rectangles in $\R_C$ into at most $\left\lceil\log_{1 + \epsau'} (\frac{1}{\delta})\right\rceil$ groups $\R_1, \R_2, \dots$, so that in each $\R_j$ the widest rectangle has width bigger than the smallest by a factor at most $1+\epsau'$; we can then define a container $C_j$ for each group $\R_j$ that has the width of the widest rectangle it contains and height equal to the sum of the heights of the contained rectangles.
	
	Consider now one such $C_j$ and the set of rectangles $\R_j$ that it contains, and let $P := p(\R_j)$. Clearly, $w(C_j) \leq (1+\epsau') w_i$ for each $R_i \in \R_j$, and so $a(C_j) \leq (1+\epsau')a(\R_j)$. If all the rectangles in $\R_j$ have height at most $\epsau' h(C_j)$, then we can remove a set of rectangles with total height at least $\epsau' h(C)$ and profit at most $2\epsau' p(\R_j)$. Otherwise, let $\mathcal{Q}$ be the set of rectangles of height larger than $\epsau' h(C_j)$, and note that $a(Q) \geq \epsau' h(C_j)w(C_j)/(1 + \epsau')$. If the $p(\mathcal{Q}) \leq \epsau' P$, then we remove the rectangles in $\mathcal{Q}$ from the container $C_j$ and reduce its height as much as possible, obtaining a smaller container $C'_j$; since $a(C'_j) \leq a(C_j) - \epsau' a(C_j) = (1 - \epsau')a(C_j) \leq (1 - \epsau')(1 + \epsau')a(\R_j) < a(\R_j)$, then the proof is finished. Otherwise, we define one container for each of the rectangles in $\mathcal{Q}$ (which are at most $1/\epsau'$) of exactly the same size, and we still shrink the container with the remaining rectangles as before; note that there is no lost area for each of the newly defined container. Since at every non-terminating iteration a set of rectangles with profit larger than $\epsau' P$ is removed, the process must end within $1/\epsau'$ iterations.
	
	Note that the total number of containers that we produce for each initial container $C_j$ is at most $1/\epsau'^2$, and this concludes the proof.
\end{proof}

Thus, by applying the above lemma to each horizontal and each vertical container, we obtain a modified packing where the total area of the horizontal and vertical containers is at most the area of the rectangles of $\R'$ (without the short-narrow rectangles, which we will take into account in the next subsection), while the number of containers increases at most by a factor $\left\lceil\log_{1 + \epsau'} (\frac{1}{\delta})\right\rceil / \epsau'^2$.

\subsection{Packing short-narrow rectangles}\label{sec:shortnarrow}

\begin{sloppypar}
	Consider the integral packing obtained from the previous subsection, which has at most $K' := \left(\frac{2}{\delta^{10}} + (\frac{3}{\delta^2}M^2 2^{2M})^{\epsau'}\right)\left\lceil\log_{1 + \epsau'} (\frac{1}{\delta})\right\rceil / \epsau'^2$ containers. We can create a non-uniform grid extending each side of the containers until they hit another container or the boundary of the knapsack. Moreover, we also add horizontal and vertical lines spaced at distance $\epsau'$. We call \emph{free cell} each face defined by the above lines that does not overlap a container of the packing; by construction, no free cell has a side bigger than $\epsau'$. The number of free cells in this grid plus the existing containers is bounded by $K_{TOTAL} = {(2K' + 1/\epsau')}^2 = O_{\epsau', \delta}(1)$. We crucially use the fact that this number does not depend on value of $\mu$.
\end{sloppypar}

Note that the total area of the free cells is no less than the total area of the short-narrow rectangles, as a consequence of the guarantees on the area of the containers introduced so far. We will pack the short-narrow rectangles into the free cells of this grid using NFDH, but we only use cells that have width and height at least $\frac{8\mu}{\epsau'}$; thus, each short-narrow rectangle will be assigned to a cell whose width (resp. height) is larger by at least a factor $8/\epsau'$ than the width (resp. height) of the rectangle. Each discarded cell has area at most $\frac{8\mu}{\epsau'}$, which implies that the total area of discarded cells is at most $\frac{8 \mu K_{TOTAL}}{\epsau'}$. Now we consider the selected cells in an arbitrary order and pack short narrow rectangles into them using NFDH, defining a new area container for each cell that is used. Thanks to Lemma~\ref{lem:nfdhPack}, we know that each new container $C$ (except maybe the last one) that is used by NFDH contains rectangles for a total area of at least $(1 - \epsau'/4)a(C)$. Thus, if all rectangles are packed, we remove the last container opened by NFDH, and we call $S$ the set of rectangles inside, that we will repack elsewhere; note that $a(S) \leq \epsau'^2 \leq \epsau'/3$, since all the rectangles in $S$ were packed in a free cell. Instead, if not all rectangles are packed by NFDH, let $S$ be the residual rectangles. In this case, the area of the unpacked rectangles is $a(S) \leq \frac{8 \mu K_{TOTAL}}{\epsau'} + \epsau'/4 \leq \epsau'/3$, assuming that $\mu \leq \frac{\epsau'^2}{96 K_{TOTAL}}$.

In order to repack the rectangles of $S$, we define a new area container $C_S$ of height $1$ and width $\epsau'/2$. Since $a(C_S) = \epsau'/2 \geq (\epsau'/3) / (1 - 2\epsau')$, all elements from $S$ are packed in $C_S$ by NFDH, and the container can be added to the knapsack by further enlarging its width from $1 + 2\delta$ to $1 + 2\delta + \epsau'/2 < 1 + \epsau'$. 

The last required step is to guarantee the necessary constraint on the total area of the area containers, similarly to what was done in Section~\ref{sec:shrinking_horiz_vert} for the horizontal and vertical containers.

Let $D$ be any full area container (that is, any area container except for $C_S$). We know that the area of the rectangles $R_D$ in $D$ is $a(R_D) \geq (1 - \epsau')a(D)$, since each rectangle $R_i$ inside $D$ has width less than $\epsau' w(D)/2$ and height less than $\epsau' h(D)/2$, by construction. We remove rectangles from $R_D$ in non-decreasing order of profit/area ratio, until the total area of the residual rectangles is between $(1 - 4\epsau')a(D)$ and $(1 - 3\epsau')a(D)$ (this is possible, since each element has area at most $\epsau'^2 a(D)$); let $R'_D$ be the resulting set. We have that $p(R'_D) \geq (1 - 4\epsau') p(R_D)$, due to the greedy choice. Let us define a container $D'$ of width $w(D)$ and height $(1 - \epsau')h(D)$. It is easy to verify that each rectangle in $R_D$ has width (resp. height) at most $\epsau' w(D')$ (resp. $\epsau' h(D')$). Moreover, since $a(R'_D) \leq (1 - 3\epsau')a(D) \leq (1 - 2\epsau')(1 - \epsau')a(C) \leq (1 - 2\epsau')a(C')$, then all elements in $R'_D$ are packed in $D'$. By applying this reasoning to each area container (except $C_S$), we obtain property (3) of Lemma~\ref{lem:structural_lemma_augm}.

Note that the constraints on $\mu$ and $\delta$ that we imposed are $\mu \leq \frac{\delta^{10}}{76}$ (from Section~\ref{sec:containers-short-rectangles}), and $\mu \leq \frac{\epsau'^2}{96 K_{TOTAL}}$. It is easy to check that both of them are satisfied if we choose $\sal{f(x) = (\epsau' x)^C}$ for a big enough constant $C$ that depends only on $\delta$ and $\epsau'$.

\subsection{Rounding containers to a polynomial set of sizes}\label{sec:rounding_containers}
In this subsection we show that it is possible to round down the size of each horizontal, vertical or area container so that the resulting sizes can be chosen from a polynomially sized set, while incurring in a marginal loss of profit.

\begin{sloppypar}
	For a set $\R$ of rectangles, we define $WIDTHS(\R) = \{w_j \, | \, R_j \in \R\} $ and $HEIGHTS(R) = \{h_j \, | \, R_j \in \R\}$.
\end{sloppypar}

Given a finite set $P$ of real numbers and a fixed natural number $k$, we define the set $P^{(k)} = \{(p_1 + p_2 + \dots + p_l) + i p_{l+1} \, | \, p_j \in P \text{ } \forall \, j, l\le k, 0 \leq i \leq n, i\mbox{ integer}\}$; note that if $|P| = O(n)$, then $|P^{(k)}| = O(n^{k+2})$. Moreover,  if $P \subseteq Q$, then obviously $P^{(k)} \subseteq Q^{(k)}$, and if $k' \leq k''$, then $P^{(k')} \subseteq P^{(k'')}$.

\begin{lem}\label{lem:round_knapsack_container}
	Let $\eps > 0$, and let $\R$ be a set of rectangles packed in a horizontal or vertical container $C$. Then, for any $k \geq 1/\eps$, there is a set $\R' \subseteq \R$ with profit $p(\R') \geq (1 - \eps)p(\R)$ that can be packed in a container $C'$ smaller than $C$ such that $w(C') \in WIDTHS(\R)^{(k)}$ and $h(C') \in HEIGHTS(R)^{(k)}$.
\end{lem}
\begin{proof}
	Without loss of generality, we prove the thesis for an horizontal container $C$; the proof for vertical containers is symmetric. Clearly, the width of $C$ can be reduced to $w_{max}(\R)$, and $w_{max}(\R) \in WIDTHS(\R) \subseteq {WIDTHS(\R)}^{(k)}$.
	
	If $|\R| \leq 1/\eps$, then $\sum_{R_i \in \R} h_i \in HEIGHTS(\R)^{(k)}$ and there is no need to round the height of $C$ down. Otherwise, let $\R_{TALL}$ be the set of the $1/\eps$ rectangles in $\R$ with largest height (breaking ties arbitrarily), let $R_j$ be the least profitable of them, and let $\R' = \R \setminus \{R_j\}$. Clearly, $p(\R') \geq (1 - \eps)p(\R)$.
	Since each element of $\R' \setminus \R_{TALL}$ has height at most $h_j$, it follows that $h(\R \setminus \R_{TALL}) \leq (n - 1/\eps) h_j$. Thus, letting $i = \left\lceil h(\R' \setminus \R_{TALL}) / h_j \right\rceil \leq n$, all the rectangles in $\R'$ fit in a container $C'$ of width $w_{max}(\R)$ and height $h(C') := h(\R_{TALL}) + i h_j \in {HEIGHTS(R)}^{(k)}$. Since $h(\R_{TALL}) + i h_j \leq h(\R_{TALL}) + h(\R' \setminus \R_{TALL}) + h_j = h(\R) \leq h(C)$, this proves the result.
\end{proof}

\begin{lem}\label{lem:round_area_container}
	Let $\eps > 0$, and let $\R$ be a set of rectangles that are assigned to an area container $C$. Then there exists a subset $\R' \subseteq \R$ with profit $p(\R') \geq (1 - 3\eps)p(\R)$ and a container $C'$ smaller than $C$ such that: $a(\R') \leq a(C)$, $w(C') \in {WIDTHS(\R)}^{(0)}$, $h(C') \in {HEIGHTS(\R)}^{(0)}$, and each $R_j \in \R'$ is $\dfrac{\eps}{1-\eps}$-small for $C'$. 
\end{lem}
\begin{proof}
	Without loss of generality, we can assume that $w(C) \leq n w_{max}(\R)$ and $h(C) \leq n h_{max}(\R)$: if not, we can first shrink $C$ so that these conditions are satisfied, and all the rectangles still fit in $C$.
	
	Define a container $C'$ so that it has width $w(C') = w_{max}(\R) \left\lfloor w(C)/w_{max}(\R) \right\rfloor$ and height $h(C') = h_{max}(\R) \left\lfloor h(C)/h_{max}(\R) \right\rfloor$, that is, $C'$ is obtained by shrinking $C$ to the closest integer multiples of $w_{max}(\R)$ and $h_{max}(\R)$. Observe that $w(C') \in {WIDTHS(\R)}^{(0)}$ and $h(C') \in {HEIGHTS(\R)}^{(0)}$. Clearly, $w(C') \geq w(C) - w_{max}(\R) \geq w(C) - \eps w(C) = (1 - \eps) w(C)$, and similarly $h(C') \geq (1 - \eps) h(C')$. Hence $a(C') \geq (1 - \eps)^2 a(C) \geq (1 - 2\eps) a(C)$.
	
	We now select a set $\R' \subseteq \R$ by greedily choosing elements from $\R$ in non-increasing order of profit/area ratio, adding as many elements as possible without exceeding a total area of $(1 - 2\eps) a(C)$. Since each element of $\R$ has area at most $\eps^2 a(C)$, then either all elements are selected (and then $p(\R') = p(\R)$), or the total area of the selected elements is at least $(1 - 2\eps - \eps^2)a(C) \geq (1 - 3\eps)a(C)$. By the greedy choice, we have that $p(\R') \geq (1 - 3\eps) p(\R)$.
	
	Since each rectangle in $\R$ is $\frac{\eps}{1 - \eps}$-small for $C'$, this proves the thesis.
\end{proof}

By applying Lemmas \ref{lem:round_knapsack_container} and \ref{lem:round_area_container} \sal{with $\eps=\epsau'$} to all the containers \sal{and noting that $\dfrac{\epsau'}{1-\epsau'} \leq \epsau$}, we completed the proof of Lemma~\ref{lem:structural_lemma_augm}.

\begin{remark}
	Note that in the above, the size of the container is rounded to a family of sizes that depends on the rectangles inside; of course, they are not known in advance in an algorithm that enumerates over all the container packings. On the other hand, if the instance is a set $\mathcal{\R}$ of $n$ rectangles, then for any natural number $k$ we have that $WIDTHS(\R)^{(k)} \subseteq WIDTHS(\R)^{(k)}$ and $HEIGHTS(\R)^{(k)} \subseteq WIDTHS(\mathcal{\R})^{(k)}$ for any $\R \subseteq \R$; clearly, $WIDTHS(\mathcal{\R})^{(k)} \times HEIGHTS(\R)^{(k)}$ has a polynomial size and can be computed explicitly.
	
	Similarly, when finding container packings for the case with rotations, one can compute the set $SIZES(\R) := WIDTHS(\R) \cup HEIGHTS(\R)$, and consider containers of width and height in $SIZES(\R)^{(k)}$ for a sufficiently high constant $k$. 
\end{remark}